\documentclass[reqno,a4paper]{amsart}

\usepackage{graphicx}
\usepackage{multicol}
\usepackage{color}
\usepackage{caption}
\usepackage{setspace}
\usepackage{amssymb,amsmath}  
\numberwithin{equation}{section}   
\usepackage{amsfonts,amsrefs}
\usepackage{amsthm} 
\usepackage{bm}    
\usepackage{bbm}      
\usepackage{tabularx}
\usepackage{enumerate}
\usepackage{algorithmic}
\usepackage{lscape}
\makeatletter
\def\verbatim@font{\linespread{1}\normalfont\ttfamily}
\makeatother
\usepackage{pgfpages}
\usepackage{changepage}
\usepackage{fancyhdr,a4wide}
\usepackage{forest}
\usepackage{listings}
\usepackage{placeins}
\usepackage[foot]{amsaddr} 

\usepackage[utf8]{inputenc}
\usepackage[T1]{fontenc}
\usepackage{amsmath}
\usepackage{tensor}
\usepackage{amssymb}
\usepackage{graphicx}
\usepackage{subfigure}
\usepackage{mathrsfs}
\usepackage{microtype}
\usepackage{rotating}
\usepackage{xspace}
\usepackage{mathtools}
\usepackage{tikz,pgfplots}
\usepackage{comment}
\usepackage{MnSymbol}
\pgfplotsset{compat=1.9}
\usetikzlibrary{external}
\tikzset{every label/.style={font=\footnotesize,inner sep=1.5pt}}
\usepackage{dutchcal} 
\usepackage{geometry}
\usepackage{hyperref}

\theoremstyle{plain}
\newtheorem{thm}{Theorem}[section]
\newtheorem{lem}[thm]{Lemma}
\newtheorem{prop}[thm]{Proposition}

\newtheorem*{thm*}{Theorem}

\theoremstyle{definition}
\newtheorem{Def}[thm]{Definition}

\theoremstyle{remark}
\newtheorem{rem}[thm]{Remark}

\input{Tdef.def}

\allowdisplaybreaks
\def\jf#1{\relax}

\title[Ultraviolet Cascades, Norm Inflation, and the Rendall Instability for the Einstein-Euler System with Positive Cosmological Constant]{Ultraviolet Cascades, Norm Inflation, and the Rendall Instability for the Einstein-Euler System with Positive Cosmological Constant}

\author[E.~Marshall]{Elliot Marshall}
\address{Institute of Applied and Computational Mathematics \\
FORTH \\
70013 Heraklion\\
Greece}
\email{elliot.marshall@iacm.forth.gr}

\begin{document}

\begin{abstract}
    We numerically study perturbations of spatially homogeneous and orthogonal solutions to the Einstein-Euler equations with positive cosmological constant and linear equation of state $p=K\rho$ for $K \in (\frac{1}{3},1)$. Recent numerical and analytic work has demonstrated that the fluid density gradient in these perturbed spacetimes can form steep gradients and blow up as future timelike infinity is approached, a phenomenon now known as the `Rendall instability'. In this article, we study the onset of this instability outside of one spatial dimension for the first time. Our results provide compelling evidence that the Rendall instability causes a forward ultraviolet cascade and $H^{s}$ norm inflation in the energy density of the fluid from small $U(1)$-symmetric perturbations. This is, to the best of our knowledge, the first example of small data norm inflation for solutions of the field equations with a positive cosmological constant.  
\end{abstract}

\maketitle

\section{Introduction}
Exponentially expanding, fluid-filled spacetimes satisfying a linear equation of state $p=K\rho$, $K \in [0,1]$ form the basis of modern cosmology. Due to their physical importance, a significant amount of research has been devoted to understanding the future stability of these models. The first rigorous results in this direction are due to Rodnianski and Speck \cites{RodnianskiSpeck:2013,Speck:2012} who established the non-linear future stability of FLRW solutions to the Einstein-Euler equations with positive cosmological constant and $K \in (0, \frac{1}{3})$. Stability at the endpoints $K=0$ and $K=\frac{1}{3}$ was subsequently established in \cite{HadzicSpeck:2015} and \cite{LubbeKroon:2013}, respectively. Related works have explored alternative topologies, equations of state, and rates of spacetime expansion, see for example \cites{Oliynyk:CMP_2016,Ringstrom:2009,Mondal:2021,Mondal:2024,Fajman_et_al:2021b,FOOW:2023,FOW:2021,Speck:2013,LiuWei:2021,Wei:2018,LeFlochWei:2021,Fajman_et_al:2025,Bernhardt:2025}.  \newline \par 

The aforementioned results all assume that $K \leq \frac{1}{3}$. Indeed, due to influential work of Rendall \cite{Rendall:2004}, it was widely expected that cosmological models with super-radiative ($K>\frac{1}{3}$) equations of state would be unstable. Specifically, Rendall conjectured that the density gradient $\frac{\del_{x^{\Omega}}\rho}{\rho}$ of the fluid would become unbounded towards future timelike infinity where the spatial fluid velocity vanishes. The blow up of density gradient in these models, now known as the `Rendall instability', has subsequently been numerically observed, both towards the future and the past, in Gowdy and $\mathbb{T}^{2}$-symmetry, see for example \cites{BMO:2023,BMO:2024,Marshall:2026,ColeyLim:2012,ColeyLim:2013,ColeyLim:2015,MarshallOliynyk:2022}. Moreover, this instability is of possible physical interest because it is inconsistent with the expected late time behaviour of exponentially expanding cosmologies where the density gradient remains bounded \cites{Lim_et_al:2004,Rendall:2004}. \newline \par

The mechanism for this instability, anticipated by Rendall in \cite{Rendall:2004}, is the difference in the late time behaviour of orthogonal and tilted fluids. Fluid tilt is defined with respect to a fixed timelike vector $u = u^{\mu}\del_{\mu}$. A fluid with  four velocity $v = v^{\mu}\del_{\mu}$ is said to be \textit{orthogonal} with respect to $u$ if $v$ is aligned with $u$ and \textit{tilted} with respect to $u$ otherwise. For a general spacetime without symmetries there is no preferred choice for $u$. However, for spatially homogeneous cosmologies there exists a preferred spacelike foliation by hypersurfaces of homogeneity. In this case, it is natural to take $u$ to be the normal vector of the foliation. \newline \par

The tilt classification is useful because homogeneous orthogonal and tilted fluids can have dramatically different asymptotic behaviour. This has been extensively studied in the dynamical systems literature \cites{EllisWainwright:1997,Sandin:2009,SandinUggla:2008,Hervik_et_al:2010,Hewitt_et_al:2001}, where it was first observed that $K=\frac{1}{3}$ acts a bifurcation point for the fluid dynamics. Fluids with $K<\frac{1}{3}$ approach orthogonal states while for $K>\frac{1}{3}$ the fluid velocity approaches a null vector at future timelike infinity. The latter behaviour is commonly referred to as \textit{extreme tilt}.  \newline \par

In the presence of a cosmological constant, we expect that solutions to the Einstein-Euler equations are pointwise well-approximated by spatially homogeneous solutions\footnote{That is to say, spatial derivatives become negligible towards the future.} towards future timelike infinity. Thus, based on the analysis of spatially homogeneous models in the dynamical systems literature, the fluid should become \textit{asymptotically} orthogonal or extremely tilted at each individual point in space. The Rendall instability occurs when asymptotically orthogonal points are adjacent to asymptotically extremely tilted points. At these two points, the density decays at different rates which causes sharp spiky features, reminiscent of the spikes seen towards the big bang \cites{BergerMoncrief:1993,BergerMoncrief:1998b,BergerMoncrief1998a,GarfinklePretorius:2020,BergerGarfinkle:1998,HernStewart:1998}, to develop. In turn, these sharp features in the density cause the density gradient to blow up. \newline \par

For $K<\frac{1}{3}$, this instability does not occur because the orthogonal state is an attractor for the fluid. However, for $K>\frac{1}{3}$, perturbations of an orthogonal spacetime, such as the FLRW solution, will, in general, contain both points where the fluid is asymptotically orthogonal and points where it is asymptotically extremely tilted. Thus, we expect that generic perturbations of spatially homogeneous and orthogonal solutions to the Einstein-Euler system with positive cosmological constant and $K>\frac{1}{3}$ will develop the Rendall instability. \newline \par

Note that not all cosmological models with $K>\frac{1}{3}$ are unstable. In fact, recent analytic work by the author in collaboration with G. Fournodavlos and T. Oliynyk \cites{Fournodavlos_et_al:2024,FMO:2025} has rigorously established that certain tilted, spatially homogeneous solutions to the Einstein-Euler equations with positive cosmological constant and super-radiative equations of state are non-linearly stable under suitably small inhomogeneous perturbations. In this case, the Rendall instability is avoided because there are no points in space where the fluid develops asymptotically orthogonal behaviour\footnote{That is to say, the fluid approaches extreme tilt at every point in space.}. Recent work by Beyer \cite{Beyer:2026} has also established analogous stability results for tilted fluids on Kasner backgrounds near the big bang. \newline \par

\subsection{Ultraviolet Cascades, Norm Inflation, and the Rendall Instability}
In \cite{Oliynyk:2024}, Oliynyk established a global-in-time, future non-linear stability result for $\mathbb{T}^{2}$-symmetric perturbations of the orthogonal solution to the spatially homogeneous Euler equations on a fixed de Sitter background with super-radiative sound speeds. Moreover, Oliynyk rigorously established that the Rendall instability occurs for this class of solutions. In particular, this result shows that the blow-up of the density gradient occurs at the points in space where the spatial fluid velocity vanishes, exactly as conjectured by Rendall. \newline \par

An interesting aspect of Oliynyk's result is that it requires a restriction on the sound speed of the fluid depending on the Sobolev regularity $s$, 
\begin{align*}
    \frac{1}{3} < K < \frac{1+s}{3s},
\end{align*}
to ensure that the $H^{s}$ norm of the perturbed solutions remains bounded. If this upper bound on the regularity is real, it would imply an $H^{s}$ instability for suitably large $s$ and each fixed $K$. Oliynyk conjectured that this instability, if it exists, is due to \textit{norm inflation}.

\begin{Def}[Norm Inflation]
    We say $H^{s}$ norm inflation occurs if, for fixed $s$ and any $\delta >0$, there exists initial data $u(0)$ with $\|u(0)\|_{H^{s}} = \delta$ such that
    \begin{align}
        \label{eqn:norm_inflation_inequality}
        \|u(T)\|_{H^{s}} > \frac{1}{\delta}.
    \end{align}
    at some finite time $T>0$.
\end{Def}

Norm inflation indicates that energy is being rapidly transferred from low to high frequencies which is known as a \textit{forward ultraviolet cascade}. Thus, Oliynyk's conjecture suggests that the Rendall instability is a mechanism for a spectral cascade in exponentially expanding spacetimes. Indeed, our numerical results support this claim. As discussed in Section \ref{sec:Numerical_Results}, we find that the Rendall instability does in fact lead to norm inflation and an ultraviolet cascade in the energy density of the fluid.  \newline \par

Although similar behaviour has previously been established for perturbations of Anti-de Sitter (AdS) spacetimes, see for example \cites{Moschidis:2020,Moschidis:2023,KehleMoschidis:2026,BizonRostworowski:2011}, our results, to the best of our knowledge, provide the first example of small data norm inflation for solutions of the field equations with a \textit{positive} cosmological constant. It is also important to emphasise that the mechanisms for the spectral cascade in these two settings are fundamentally different. In the context of AdS, the instability is driven by resonant mode mixing, see \cites{Bizon_et_al:2015,Jalmuzna_et_al:2011,BizonRostworowski:2011} for details, while the Rendall instability is driven by sharp spatial features emerging in the fluid energy density at the transition between tilted and orthogonal states.

\subsection{Outline of Paper}
The purpose of this article is to numerically investigate the development of ultraviolet cascades and norm inflation from the Rendall instability in solutions of the $U(1)$-symmetric Einstein-Euler equations with positive cosmological constant and super-radiative equations of state. To this end, we study small perturbations of spatially homogeneous and orthogonal spacetimes. In particular, this is the first study of the Rendall instability outside of one spatial dimension. \newline \par

As part of this article, we develop a new form of the orthonormal frame formalism \cites{UEWE:2003,Ellis:1967,EllisMaccallum:1969,MacCallum:1973,ElstUggla:1997} which is adapted to the presence of a single spacelike Killing vector. This generalises earlier work of Wainwright and collaborators on the orthonormal frame formalism for two spacelike Killing vectors \cites{Wainwright:1979,Wainwright:1981, Elst_et_al:2001,HewittWainwright:1990}. Our numerical simulations lead to the following conclusions:
\begin{enumerate}[(i)]
    \item The Rendall instability can occur outside of one spatial dimension and drives a forward ultraviolet cascade in the energy density of the fluid.
    \item $H^{s}$ norm inflation occurs for the energy density of the fluid, however the resolution of our simulations is insufficient to conclusively test the relationship between the sound speed and Sobolev regularity conjectured by Oliynyk. 
\end{enumerate}

The article is structured as follows. In Section \ref{sec:Orthonormal_Frame_Formalism_U(1)} we derive the Einstein equations for a symmetry adapted orthonormal frame. We then show that these equations can be expressed in a symmetric hyperbolic form in Section \ref{sec:Sym_Hyp_U1_Eqns}, followed by a derivation of the Euler equations, relative to our frame, in Section \ref{sec:U(1)_Euler_Equations}. In Section \ref{sec:Numerical_Implementation}, we cover our numerical implementation, choice of initial data, and code tests. Finally, in Section \ref{sec:Numerical_Results}, we discuss our numerical results.

\subsection{Notation and Indexing Conventions}
Our indexing conventions are as follows: lower case Greek letters, e.g. $\mu,\nu,\kappa$, will label spacetime coordinate indices that run from $0$ to $3$ while upper case Greek letters, e.g. $\Gamma,\Omega,\Sigma$, will label spatial coordinate indices that run from $1$ to $3$. We will primarily work with orthonormal frames $e_{a} = e^{\mu}_{a}\del_{\mu}$. Lower case Latin letters, e.g. $a$, $b$, $c$, that run from $0$ to $3$ will label spacetime frame indices while spatial frame indices will be labelled by upper case
Latin letters, e.g. $A$, $B$, $C$, and run from $1$ to $3$. Lower case calligraphic Latin indices, e.g. $\mathcal{a}$, $\mathcal{b}$, denote frame indices
which run from 1 to 2.

\subsection{Index Operations}
The \textit{symmetrisation}, \textit{anti-symmetrisation}, and \textit{symmetric trace-free} operations on pairs of spatial frame indices are defined by
\begin{gather*}
    \Sigma_{(AB)} = \frac{1}{2}(\Sigma_{AB} + \Sigma_{BA}), \quad \Sigma_{[AB]} = \frac{1}{2}(\Sigma_{AB}-\Sigma_{BA}),\intertext{and}
    \Sigma_{\langle AB \rangle} = \Sigma_{(AB)}- \frac{1}{3}\delta^{CD}\Sigma_{CD}\delta_{AB},
\end{gather*}
respectively.

\subsection{Acknowledgements}
I would like to thank Todd Oliynyk for helpful conversations. I gratefully acknowledge the support of the ERC starting grant 101078061 SINGinGR, under the European Union’s Horizon Europe program for research and innovation. I am also grateful to the referee whose comments greatly improved the clarity of this paper's results.

\section{The U(1)-Symmetric Orthonormal Frame Formalism}
\label{sec:Orthonormal_Frame_Formalism_U(1)}
As mentioned in the introduction, we are interested in studying $U(1)$-symmetric solutions of the Einstein-Euler equations with positive cosmological constant. The purpose of this section is to derive a new form of the orthonormal frame formalism adapted to a single spacelike Killing vector. For the most part, previous work using the frame formalism, outside of spatial homogeneity\footnote{See \cite{EllisWainwright:1997} for details on the orthonormal frame formalism in spatially homogeneous cosmologies.}, has considered either the full $3+1$ system or symmetry reduced $1+1$ systems, see for example \cites{Marshall:2026,BOZ:2025,FMO:2025,Zheng:2026,Garfinkle_et_al:2023,Garfinkle:2004,GarfinklePretorius:2020,Lim_Thesis:2004,ColeyLim:2012,Andersson_et_al:2005}. $U(1)$-symmetry presents a middle-ground between these two options; it is simpler than the $3+1$ system but also allows for more complex dynamical behaviour than $1+1$ systems. Another advantage is that numerical simulations of the $U(1)$ system require significantly less computational power than $3+1$ simulations of comparable resolution. \newline \par

The derivation in this section is somewhat long, but we have included almost all details as we expect it will be a useful resource for future numerical or analytic studies in $U(1)$-symmetry. For the impatient reader, the final symmetric hyperbolic form of the Einstein equations is given in equations \eqref{eqn:sigma_plus_evo}-\eqref{eqn:CH_U1}.

\begin{rem}
A similar approach has recently been utilised by Dong \cite{Dong:2026} to study big bang stability for the class of polarised\footnote{A spacetime is polarised if the Killing vector $\xi$ is hypersurface orthogonal.} $U(1)$-symmetric spacetimes. In this case, a variant of the orthonormal frame formalism adapted to the reduced $2+1$ dimensional spacetime was employed. Instead our formalism works with frames on the full 3+1 spacetime. Furthermore, our formalism is not restricted to polarised $U(1)$-symmetry. 
\end{rem}

\subsection{Group Invariant Frames}
To begin, we consider spacetimes of the form
\begin{align*}
    M = \mathbb{R} \times \Sigma.
\end{align*}
where $\Sigma$ is a three-dimensional, orientable, compact manifold without boundary. These conditions imply that $M$ admits a global frame\footnote{That is, $M$ is parallelisable, see \cite{Parker:1984} for a proof.} $\{e_{a}\}$. As discussed above, we are interested in the case where the spacetime $M$ has a spacelike Killing vector field $\xi$ which generates the $U(1)$-symmetry group.
\newline \par
 
First, we introduce an orthonormal frame $\{e_{i}\}$, $i=0,\cdots,3$, where the timelike vector $e_{0}$ is hypersurface orthogonal. We will assume that $e_{3}$ is tangent to the group orbits, that is
\begin{align}
\label{eqn:orbit_aligned_e3}
    e_{3} = f\xi
\end{align}
for some function $f$.  In particular, since $e_{3}$ is a unit vector, we have that
\begin{align*}
    f = \lambda^{\frac{-1}{2}}, \quad \lambda = \xi_{a}\xi^{a}.
\end{align*}
Frames satisfying \eqref{eqn:orbit_aligned_e3} are called \textit{orbit aligned}. Note that we are always allowed to (locally) rotate our spatial frame such that $e_{3}$ is orbit aligned. Indeed, \eqref{eqn:orbit_aligned_e3} amounts to a gauge choice for the spatial frame. \newline \par

Expressed in this frame, the components of the metric are given by
\begin{align*}
    g_{ab}:= g(e_{a},e_{b}) = \eta_{ab}
\end{align*}
where $\eta_{ab}$ is the Minkowski metric\footnote{Frame indices will always be raised and lowered using the frame metric $g_{ab} = \eta_{ab}$.}. The primary variables in the orthonormal frame formalism are the commutator coefficients, $\tensor{c}{_a^c_b}$, defined by
\begin{align}
    [e_{a},e_{b}] = \tensor{c}{_a^c_b}e_{c}.
\end{align}
Similarly, we can define the connection coefficients $\omega_{abc}$ by
\begin{equation}
\begin{aligned}
\label{eqn:connection_coefficient_defn}
\nabla_{e_{a}}e_{b} = \tensor{\omega}{_a^c_b}e_{c}, \;\;
\omega_{acb} = g_{cd}\tensor{\omega}{_a^d_b}, \;\; 
\omega_{acb} = - \omega_{abc}.
\end{aligned}
\end{equation}
Since the frame $e_{a}$ is orthonormal, the connection and commutator coefficients uniquely determine each other\footnote{Hence, we could equivalently work with the connection coefficients.} via the relations
\begin{align}
\label{eqn:commutator_connection_relations}
    \tensor{c}{_a^c_b} = \tensor{\omega}{_a^c_b} -\tensor{\omega}{_b^c_a}, \;\; \tensor{\omega}{_{abc}} = \frac{1}{2}(c_{bac}-c_{cba}-c_{acb}).
\end{align}

Next, we want to consider frames which are \textit{group invariant}.
\begin{Def}[Group Invariant Frame]
An orthonormal frame is group invariant if
\begin{align}
   \label{eqn:Group_Invariant_Frame_Condition}
    [e_{i},\xi] = 0,
\end{align}
for all $i=0,1,2,3$.
\end{Def}
As we will see, working with an orbit aligned, group invariant frame considerably simplifies the frame equations. \newline \par

In \cite{Geroch:1971} Geroch showed that vector fields on the 2+1 quotient space generated by the U(1) symmetry action are in one-to-one correspondence with vector fields $V$ on $M$ satisfying 
\begin{align*}
    \mathcal{L}_{\xi}V = 0, \quad g(\xi,V) = 0.
\end{align*}
where $g$ is the metric on $M$. Hence, we can \textit{locally} introduce an orthonormal frame on the 2+1 quotient space which pulls back to a collection of group invariant, orthonormal vector fields $\{e_{0},e_{1},e_{2}\}$ on $M$. This can then be completed to an orbit aligned, group invariant frame on $M$ by including $e_{3} = \lambda^{\frac{-1}{2}}\xi$, $\lambda \neq 0$. Now, by construction, 
\begin{align*}
    [e_{a},\xi] = 0, \quad g(e_{a},\xi) = 0, \quad a = 0,1,2.
\end{align*}
Similarly, for $e_{3}$ we have
\begin{align*}
       [e_{3},\xi] = -\xi(\lambda^{\frac{-1}{2}})\xi - f[\xi,\xi] = \frac{1}{2}\lambda^{\frac{-3}{2}}\xi(\lambda)\xi.
   \end{align*}
Hence, we can obtain \eqref{eqn:Group_Invariant_Frame_Condition} for $i=3$ provided that $\xi(\lambda)=0$. This is true since 
\begin{align}
\label{eqn:lambda_invariance}
       \xi(\lambda) = \mathcal{L}_{\xi}\lambda &= \xi^{a}\nabla_{a}\lambda \nonumber \\
       &= \xi^{a}\nabla_{a}(\xi^{b}\xi_{b}) \nonumber \\
       &= \xi^{a}\xi^{b}\big(\nabla_{a}\xi_{b} + \nabla_{b}\xi_{a}\big) \nonumber \\
       &= 0,
\end{align}
where we have used Killing's equation 
\begin{align}
    \label{eqn:Killing_eqn}
        \nabla_{a}\xi_{b} = -\nabla_{b}\xi_{a},
\end{align}
to obtain the last equality. Together the group invariance and orbit alignment of the frame $\{e_{a}\}$ imply
\begin{align}
    \label{eqn:orbit_alignment_simplifications}
    [e_{0},e_{3}] = \tensor{c}{_0^3_3}e_{3}, \quad [e_{1},e_{3}] = \tensor{c}{_1^3_3}e_{3}, \quad
       [e_{2},e_{3}] = \tensor{c}{_2^3_3}e_{3}. 
\end{align}

Thus, we can locally construct an orbit aligned, group invariant frame and, by \eqref{eqn:orbit_alignment_simplifications}, using such a frame implies 
\begin{align}
\label{eqn:Group_Invariant_Commutator_Simplifications}
    \tensor{c}{_a^c_3} = 0, \quad a,c = 0,1,2.
\end{align}
Furthermore, using the Jacobi identity applied to the three vectors $\xi$, $e_{a}$, $e_{b}$ ($a\neq b$) and the group invariance of the frame, we find that the commutator coefficients are constant on the group orbits
\begin{align}
\label{eqn:Jacobi_Identity_KillingVector}
   0 &= \big[\xi, [e_{a},e_{b}]\big] + \big[e_{a}, [e_{b},\xi]\big] + \big[e_{b}, [\xi,e_{a}]\big] \nonumber \\
   &= \big[\xi, \tensor{c}{_a^c_b}e_{c}\big] \nonumber \\
   &= \xi(\tensor{c}{_a^c_b})e_{c} + \tensor{c}{_a^c_b}[\xi,e_{c}] \nonumber \\
   &= \xi(\tensor{c}{_a^c_b})e_{c}.
\end{align}
Combining the above with the fact $e_{3} = \lambda^{\frac{-1}{2}}\xi$, we then obtain
\begin{align}
\label{eqn:e3_commutator_independence}
    e_{3}(\tensor{c}{_a^c_b}) = 0.
\end{align}

\begin{rem}[Polarised $U(1)$-symmetry]
Spacetimes with a hypersurface orthogonal Killing vector $\xi$ are called \textit{polarised} and form an invariant subset of the wider $U(1)$-symmetric class. In this case, $e_{0}$, $e_{1}$, and $e_{2}$ are surface forming,
\begin{align*}
    [e_{a},e_{b}] = \tensor{c}{_a^c_b}e_{c}, \quad a, b, c = 0,1,2,
\end{align*}
which implies
\begin{align}
\label{eqn:HO_Killing_Commutator_Simplifications}
    \tensor{c}{_0^3_1} = \tensor{c}{_0^3_2} = \tensor{c}{_1^3_2} = 0.
\end{align}
Polarised $U(1)$-symmetry is also equivalent to the twist of the Killing vector vanishing, cf. Appendix \ref{app:Frame_Quantities}.
\end{rem}

\subsection{Orthonormal Frame Formalism}
\label{sec:U1_OrthonormalFrame_Formalism}
We now want to express the orthonormal frame formalism using our symmetry adapted frame. See \cites{ElstUggla:1997,EllisWainwright:1997,RöhrUggla:2005} for detailed discussions of the general $3+1$ formalism. First, the commutator coefficients are decomposed into new irreducible variables\footnote{Note, since we have assumed $e_{0}$ is hypersurface orthogonal, our frame is irrotational and $w_{A}=0$.}, 
\begin{align}
\label{eqn:mixed_commutator_decomposition}
    [e_{0},e_{A}] &= \dot{u}_{A}e_{0} - \big[ H\delta_{A}^{B} + \tensor{\sigma}{_A^B} + \tensor{\eps}{_A^B_C}\Omega^{C}\big]e_{B}, \\
\label{eqn:spatial_commutator_decomposition}
    [e_{A},e_{B}] &= \big(\eps_{ABD}n^{CD} + 2a_{[A}\delta^{C}_{B]}\big) e_{C},
\end{align}
where $n_{AB}$ is symmetric, $\sigma_{AB}$ is symmetric and trace-free, and $\eps_{ABC}$ is the totally antisymmetric permutation symbol satisfying
\begin{align*}
    \eps_{123} = 1.
\end{align*}
In the framework of the $3+1$ decomposition, the variables $\dot{u}_{A}$, $H$, and $\tensor{\sigma}{_{AB}}$ correspond to the frame components of the acceleration vector, the trace of the extrinsic curvature, and the trace-free symmetric part of the extrinsic curvature, respectively. Similarly, $n_{AB}$ and $a_{A}$ are related to the Ricci tensor ${}^{(3)}R_{AB}$ of the $t=$ constant spacelike hypersurfaces via the formulas
\begin{align*}
{}^{(3)}S_{AB} &= e_{(A}(a_{B)}) - \frac{1}{3}e_{C}(a^{C})\delta_{AB} - (e_{C} - 2a_{C})n_{D(A}\tensor{\eps}{_{B)}^{CD}} + b_{AB} - \frac{1}{3}\tensor{b}{_C^C}\delta_{AB},\\
{}^{(3)}R &= 4e_{C}(a^{C}) - 6a_{C}a^{C} - \frac{1}{2}\tensor{b}{_C^C}, \\
b_{AB} &= 2\tensor{n}{_A^C}n_{CB} - \tensor{n}{_C^C}n_{AB},
\end{align*}
where ${}^{(3)}R$ and ${}^{(3)}S_{AB}$ and are the trace and trace-free parts of ${}^{(3)}R_{AB}$, respectively. Finally, $\Omega^{A}$ is the local angular velocity of the spatial frame $e_{A}$ with respect to a frame which is Fermi propagated along $e_{0}$. In particular, $\Omega^{A}=0$ when we choose our spatial frame to be Fermi propagated.  \newline \par

Now, by combining the conditions \eqref{eqn:Group_Invariant_Commutator_Simplifications} with the expressions for the frame commutators \eqref{eqn:mixed_commutator_decomposition}-\eqref{eqn:spatial_commutator_decomposition} we find
\begin{gather*}
    \dot{u}_{3} = 0, \quad \sigma_{13} + \Omega_{2} = 0, \quad \sigma_{23} - \Omega_{1} = 0, \\
    n_{22} = 0, \quad  n_{11} = 0, \quad n_{12} + a_{3} = 0, \quad n_{12} - a_{3} = 0.
\end{gather*}
The last two identities together imply
\begin{align*}
    n_{12}=a_{3} = 0.
\end{align*}
To summarise, choosing an orbit aligned, group invariant frame implies
\begin{align}
\label{eqn:U1_Frame_Simplifications}
    \dot{u}_{3}=a_{3}=n_{11}=n_{12}=n_{22}=0, \quad \sigma_{13}+\Omega_{2}=0, \quad \sigma_{23}-\Omega_{1}=0.
\end{align}
Similarly, by combining \eqref{eqn:HO_Killing_Commutator_Simplifications} with \eqref{eqn:mixed_commutator_decomposition}-\eqref{eqn:spatial_commutator_decomposition}, we see that 
\begin{align}
\label{eqn:Polarised_U1_Frame_Simplifications_a}
\sigma_{13}-\Omega_{2} = 0, \quad \sigma_{23}+\Omega_{1} = 0, \quad n_{33} = 0.
\end{align}
when $\xi$ is hypersurface orthogonal. The first two identities, when combined with \eqref{eqn:U1_Frame_Simplifications} then imply
\begin{align}
\label{eqn:Polarised_U1_Frame_Simplifications_b}
    \sigma_{13}=\sigma_{23}=\Omega_{2}=\Omega_{1}=0,
\end{align}
for a polarised spacetime.

\begin{rem}[Berger-Moncrief $U(1)$ Parameterisation]
    In \cites{BergerMoncrief1998a,BergerMoncrief:1998b}, Berger and Moncrief numerically studied $U(1)$-symmetric spacetimes using a parameterisation of the metric derived by Moncrief in \cite{Moncrief:1986}. An explicit example of an orbit aligned, group invariant frame and its corresponding commutator coefficients is given in Appendix \ref{app:BergerMoncrief_FrameExample} for this form of the metric.
\end{rem}

Next, we introduce symmetry-adapted local coordinates such that $\xi = \frac{\del}{\del x^{3}}$ and the local coordinate representation of our frame is given by
\begin{align}
\label{eqn:symmetry_adapted_frame_ansatz}
    e_{0} = \alpha^{-1}(\del_{t} - \beta^{\Omega}\del_{x^{\Omega}}), \quad e_{1} = e_{1}^{1}\del_{x^{1}} + e_{1}^{2}\del_{x^{2}} + e_{1}^{3}\del_{x^{3}}, \quad e_{2} = e_{2}^{1}\del_{x^{1}} + e_{2}^{2}\del_{x^{2}} + e_{2}^{3}\del_{x^{3}}, \quad e_{3} =  e_{3}^{3}\del_{x^{3}},
\end{align}
where $\alpha$ is the lapse, $\beta^{\Omega}$ is the shift vector, and $x^{1}$ and $x^{2}$ are local coordinates on a compact two-manifold. Using the group invariance of the frame, it is straightforward to verify that coordinate components of the frame are constant on the group orbits. Hence, in symmetry adapted local coordinates, all variables depend only on $(t,x^{1},x^{2})$. \newline \par

We can obtain the symmetry reduced frame equations by substituting \eqref{eqn:U1_Frame_Simplifications} and \eqref{eqn:symmetry_adapted_frame_ansatz} into the general 3+1 frame equations \eqref{eqn:frame_evo_3+1}-\eqref{eqn:sigma_evo_3+1} given in Appendix \ref{app:3+1_Frame_Eqns}. Additionally, we assume that the frame components of the stress-energy tensor are constant on the group orbits, 
\begin{align}
\label{eqn:U(1)_Symmetric_Matter_Condition}
    \xi(T_{ab}) = e_{3}(T_{ab}) = 0.
\end{align}
Ultimately, after a series of lengthy calculations, we obtain the following system of symmetry reduced evolution equations 
\begin{align}
\label{eqn:e11_evo}
e_{0}(e^{1}_{1}) &= -(H + \sigma_{11})e^{1}_{1} - (\sigma_{12} + \Omega_{3})e_{2}^{1}, \\
\label{eqn:e12_evo}
e_{0}(e_{1}^{2}) &= -(H + \sigma_{11})e_{1}^{2} - (\sigma_{12} + \Omega_{3})e_{2}^{2}, \\
\label{eqn:e13_evo}
e_{0}(e_{1}^{3}) &= -(H + \sigma_{11})e_{1}^{3} - (\sigma_{12} + \Omega_{3})e^{3}_{2} -2\sigma_{13}e_{3}^{3}, \\
\label{eqn:e21_evo}
e_{0}(e_{2}^{1}) &= -(\sigma_{12} - \Omega_{3})e_{1}^{1} - (H + \sigma_{22})e_{2}^{1}, \\
\label{eqn:e22_evo}
e_{0}(e_{2}^{2}) &= -(\sigma_{12} - \Omega_{3})e_{1}^{2} - (H + \sigma_{22})e_{2}^{2}, \\
\label{eqn:e23_evo}
e_{0}(e_{2}^{3}) &= -(\sigma_{12} - \Omega_{3})e_{1}^{3} - (H+\sigma_{22})e_{2}^{3} - 2\sigma_{23}e_{3}^{3}, \\
\label{eqn:e33_evo}
e_{0}(e_{3}^{3}) &= -(H + \sigma_{33})e_{3}^{3}, \\
\label{eqn:a1_evo}
e_{0}(a_{1}) &= -e_{1}(H) + \frac{1}{2}e_{1}(\sigma_{11}) + \frac{1}{2}e_{2}(\sigma_{12} + \Omega_{3}) - H(\dot{u}_{1}+a_{1}) + \sigma_{11}(\frac{1}{2}\dot{u}_{1}-a_{1}) \nonumber \\ 
&+ (\sigma_{12} + \Omega_{3})(\frac{1}{2}\dot{u}_{2}-a_{2}), \\
\label{eqn:a2_evo}
e_{0}(a_{2}) &=  -e_{2}(H) +\frac{1}{2}e_{2}(\sigma_{22}) + \frac{1}{2}e_{1}(\sigma_{12}-\Omega_{3}) -H(\dot{u}_{2} + a_{2}) + \sigma_{22}(\frac{1}{2}\dot{u}_{2} - a_{2}) \nonumber \\
&+ (\sigma_{12}-\Omega_{3})(\frac{1}{2}\dot{u}_{1}-a_{1}), \\
\label{eqn:n13_evo}
e_{0}(n_{13}) &= \frac{1}{2}e_{1}(\Omega_{3}-\sigma_{12}) + \frac{1}{2}e_{2}(\sigma_{11} - \sigma_{33}) + n_{13}(\sigma_{11} + \sigma_{33} - H) - n_{23}(\Omega_{3} - \sigma_{12}) \nonumber \\
&+ \frac{1}{2}\dot{u}_{1}(\Omega_{3} - \sigma_{12}) + \frac{1}{2}\dot{u}_{2}(\sigma_{11} - \sigma_{33}), \\
\label{eqn:n23_evo}
e_{0}(n_{23}) &= \frac{1}{2}e_{1}(\sigma_{33} - \sigma_{22}) + \frac{1}{2}e_{2}(\sigma_{12} + \Omega_{3}) + n_{23}(\sigma_{33} + \sigma_{22} - H) + n_{13}(\sigma_{12}+\Omega_{3}) \nonumber \\
&+ \frac{1}{2}\dot{u}_{1}(\sigma_{33}-\sigma_{22}) + \frac{1}{2}\dot{u}_{2}(\sigma_{12} + \Omega_{3}), \\
\label{eqn:n33_evo}
e_{0}(n_{33}) &= -2e_{1}(\sigma_{23}) + 2e_{2}(\sigma_{13}) +n_{33}(2\sigma_{33}-H) + 4n_{13}\sigma_{13} + 4n_{23}\sigma_{23} \nonumber \\
&- 2\dot{u}_{1}\sigma_{23} + 2\dot{u}_{2}\sigma_{13}, \\
\label{eqn:H_evo}
e_{0}(H) &= \frac{1}{3}e_{\mathcal{a}}(\dot{u}^{\mathcal{a}})  - H^{2} + \frac{1}{3}\dot{u}_{\mathcal{a}}(\dot{u}^{\mathcal{a}} - 2a^{\mathcal{a}}) - \frac{1}{3}\sigma_{AB}\sigma^{AB} - \frac{1}{6}(T_{00} + \tensor{T}{_A^A}), \\
\label{eqn:sigma11_evo}
e_{0}(\sigma_{11}) &= \frac{2}{3}e_{1}(\dot{u}_{1}-a_{1}) + \frac{1}{3}e_{2}(a_{2}-\dot{u}_{2}) + e_{2}(n_{13}) - 3H\sigma_{11} + \frac{2}{3}(\dot{u}_{1}^{2} + a_{1}\dot{u}_{1}) \nonumber \\
&-\frac{1}{3}(\dot{u}_{2}^{2}  + a_{2}\dot{u}_{2}) + n_{13}(\dot{u}_{2} - 2a_{2}) - 2(\sigma_{12}\Omega_{3} + \sigma_{13}^{2}) \nonumber \\
&+ \frac{1}{3}(4n_{23}^{2} + n_{33}^{2} - 2n_{13}^{2}) + T_{11} - \frac{1}{3}\tensor{T}{_A^A}, \\
\label{eqn:sigma12_evo}
e_{0}(\sigma_{12}) &= \frac{1}{2}e_{1}(\dot{u}_{2}-a_{2}) - \frac{1}{2}e_{1}(n_{13}) + \frac{1}{2}e_{2}(\dot{u}_{1} - a_{1}) + \frac{1}{2}e_{2}(n_{23}) -3H\sigma_{12} \nonumber \\
&+ \frac{1}{2}n_{23}(\dot{u}_{2} - 2a_{2}) + (\sigma_{11}-\sigma_{22})\Omega_{3} -2\sigma_{13}\sigma_{23} - \frac{1}{2}n_{13}(\dot{u}_{1} - 2a_{1}) \nonumber \\
&- 2n_{13}n_{23} +\frac{1}{2}\dot{u}_{1} (\dot{u}_{2} + a_{2}) + \frac{1}{2}\dot{u}_{2}(\dot{u}_{1} + a_{1}) + T_{12} , \\
\label{eqn:sigma13_evo}
e_{0}(\sigma_{13}) &= \frac{1}{2}e_{2}(n_{33}) -3H\sigma_{13} + \frac{1}{2}n_{33}(\dot{u}_{2} -2a_{2}) + \sigma_{13}(\sigma_{11} - \sigma_{33}) + \sigma_{23}( \sigma_{12} - \Omega_{3}) \nonumber \\
&- n_{13}n_{33} + T_{13} , \\
\label{eqn:sigma22_evo}
e_{0}(\sigma_{22}) &= \frac{1}{3}e_{1}(a_{1}-\dot{u}_{1}) - e_{1}(n_{23}) + \frac{2}{3}e_{2}(\dot{u}_{2}-a_{2}) -3H\sigma_{22} - \frac{1}{3}(\dot{u}_{1}^{2} + a_{1}\dot{u}_{1}) \nonumber \\
&+\frac{2}{3}(\dot{u}_{2}^{2} + a_{2}\dot{u}_{2}) - n_{23}(\dot{u}_{1} - 2a_{1}) + 2(\sigma_{12}\Omega_{3} - \sigma_{23}^{2}) + \frac{1}{3}(4n_{13}^{2}  + n_{33}^{2} - 2n_{23}^{2}) \nonumber \\
&+ T_{22} - \frac{1}{3}\tensor{T}{_A^A}, \\
\label{eqn:sigma23_evo}
e_{0}(\sigma_{23}) &= -\frac{1}{2}e_{1}(n_{33}) -3H\sigma_{23} + \frac{1}{2}n_{33}(2a_{1} -\dot{u}_{1}) + \sigma_{23}(\sigma_{22} - \sigma_{33}) + \sigma_{13}(\Omega_{3} + \sigma_{12}) \nonumber \\
&- n_{23}n_{33} + T_{23}, \\
\label{eqn:sigma33_evo}
e_{0}(\sigma_{33}) &= \frac{1}{3}e_{1}(a_{1}-\dot{u}_{1}) + e_{1}(n_{23}) + \frac{1}{3}e_{2}(a_{2}-\dot{u}_{2}) - e_{2}(n_{13}) -3H\sigma_{33} - \frac{1}{3}(\dot{u}_{1}^{2} + a_{1}\dot{u}_{1}) \nonumber \\
&- \frac{1}{3}(\dot{u}_{2}^{2} + a_{2}\dot{u}_{2}) + n_{13}(2a_{2} - \dot{u}_{2}) + n_{23}(\dot{u}_{1} - 2a_{1}) + 2(\sigma_{23}^{2} + \sigma_{13}^{2}) - \frac{2}{3}(n_{13}^{2} + n_{23}^{2} + n_{33}^{2}) \nonumber \\
&+ T_{33} - \frac{1}{3}\tensor{T}{_A^A}.
\end{align}
The variables $\alpha$, $\beta^{\Omega}$, $\dot{u}_{A}$, and $\Omega_{A}$ are gauge variables and do not have evolution equations. Similarly, the constraint equations \eqref{eqn:C1_Constraint_3+1}-\eqref{eqn:Hamiltonian_Constraint_3+1} reduce to
\begin{align}
\label{eqn:C1_Constraint_112}
0 =\tensor{(\mathcal{C}_{1})}{_1^1_2} &:= 2e_{[1}(e_{2]}^{1}) - 2a_{[1}e_{2]}^{1} - (n^{13}e_{1}^{1} + n^{23}e_{2}^{1}) , \\
\label{eqn:C1_Constraint_122}
0 =\tensor{(\mathcal{C}_{1})}{_1^2_2} &:= 2e_{[1}(e_{2]}^{2}) - 2a_{[1}e_{2]}^{2} - (n^{13}e_{1}^{2} + n^{23}e_{2}^{2}) , \\
\label{eqn:C1_Constraint_132}
0 =\tensor{(\mathcal{C}_{1})}{_1^3_2} &:= 2e_{[1}(e_{2]}^{3}) - 2a_{[1}e_{2]}^{3} - (n^{13}e_{1}^{3} + n^{23}e_{2}^{3} + n^{33}e_{3}^{3}) , \\
\label{eqn:C1_Constraint_133}
0 = \tensor{(\mathcal{C}_{1})}{_1^3_3} &:= e_{1}(e_{3}^{3}) - a_{1}e_{3}^{3} + n^{23}e_{3}^{3}, \\
\label{eqn:C1_Constraint_233}
0 = \tensor{(\mathcal{C}_{1})}{_2^3_3} &:= e_{2}(e_{3}^{3}) - a_{2}e_{3}^{3} - n^{13}e_{3}^{3}, \\
\label{eqn:C2_Constraint_12}
0 = (\mathcal{C}_{2})_{12} &:= 2e_{[1}(\dot{u}_{2]}) + 2a_{[1}\dot{u}_{2]} + n^{13}\dot{u}_{1} + n^{23}\dot{u}_{2}, \\
\label{eqn:C3_Constraint_1}
0 = (\mathcal{C}_{3})_{1} &:= \dot{u}_{1} - \alpha^{-1}e_{1}(\alpha), \\
\label{eqn:C3_Constraint_2}
0 = (\mathcal{C}_{3})_{2} &:= \dot{u}_{2} - \alpha^{-1}e_{2}(\alpha), \\
\label{eqn:MomentumConstraint_1}
0 = (\mathcal{C}_{M})_{1} &:= e_{1}(\sigma_{11}) + e_{2}(\sigma_{12}) - 2e_{1}(H) - 3\sigma_{11}a_{1} - 3\sigma_{12}a_{2} + n_{23}(\sigma_{22} - \sigma_{33}) \nonumber \\
&+ n_{13}\sigma_{12} + n_{33}\sigma_{23} - T_{01}, \\
\label{eqn:MomentumConstraint_2}
0 = (\mathcal{C}_{M})_{2} &:= e_{1}(\sigma_{12}) + e_{2}(\sigma_{22}) - 2e_{2}(H) - 3\sigma_{12}a_{1} - 3\sigma_{22}a_{2} + n_{13}(\sigma_{33} - \sigma_{11}) \nonumber \\
&- n_{23}\sigma_{12} - n_{33}\sigma_{13} - T_{02}, \\
\label{eqn:MomentumConstraint_3}
0 = (\mathcal{C}_{M})_{3} &:= e_{1}(\sigma_{13}) + e_{2}(\sigma_{23}) - 3\sigma_{13}a_{1} - 3\sigma_{23}a_{2} - n_{13}\sigma_{23} + n_{23}\sigma_{13} - T_{03}, \\
\label{eqn:HamiltonianConstraint_U1}
0 = (\mathcal{C}_{H}) &:= 4e_{1}(a_{1}) + 4e_{2}(a_{2}) + 6H^{2} -6(a_{1}^{2} + a_{2}^{2}) - (2n_{13}^{2} + 2n_{23}^{2} + \frac{1}{2}n_{33}^{2}) \nonumber \\
&- \sigma_{AB}\sigma^{AB} - 2T_{00},
\end{align}
where $(\mathcal{C}_{M})$ and $(\mathcal{C})_{H}$ are the momentum and Hamiltonian constraints, respectively. \newline \par 

Since all of the commutator variables are independent of $x^{3}$, the $e_{1}$ and $e_{2}$ frame derivatives reduce to
\begin{align*}
    e_{1}(f) = e_{1}^{1}\del_{x^{1}}(f) + e_{1}^{2}\del_{x^{2}}(f), \quad e_{2}(f) = e_{2}^{1}\del_{x^{1}}(f) + e_{2}^{2}\del_{x^{2}}(f).
\end{align*}
In particular, this means the evolution equations \eqref{eqn:e13_evo}, \eqref{eqn:e23_evo}, and \eqref{eqn:e33_evo} for the frame components $e_{1}^{3}$, $e_{2}^{3}$, and $e_{3}^{3}$, respectively, decouple from the remaining evolution equations. 

\begin{rem}[Polarisation Condition]
For a polarised $U(1)$-symmetric spacetime we have that
\begin{align*}
    n_{33} = \sigma_{23} = \sigma_{13} = 0.
\end{align*}
From inspection of \eqref{eqn:e11_evo}-\eqref{eqn:sigma33_evo}, we see that this is an invariant set preserved by the evolution equations provided that the stress-energy tensor satisfies
\begin{align*}
    T_{13} = T_{23} = 0.
\end{align*}
Additionally, the constraint \eqref{eqn:MomentumConstraint_3} reduces to 
\begin{align*}
    T_{03} = 0.
\end{align*}
\end{rem}

\subsection{Topological Constraints}
Since all the variables in the system \eqref{eqn:e11_evo}-\eqref{eqn:HamiltonianConstraint_U1} are constant along the group orbits and independent of $x^{3}$ in symmetry adapted coordinates, the $U(1)$-symmetric field equations reduce to a 2+1 dimensional problem. In particular, the spacetime $M$ can be interpreted as a principal $U(1)$ bundle\footnote{See for example \cite{Wald:1984}*{\S 13.2}.} over a 2+1 dimensional Lorentzian base space $S$ with bundle map\footnote{Specifically, $\pi$ maps each integral curve of the Killing vector $\xi$ to a point in $S$.}
\begin{align*}
    \pi: M \rightarrow S.
\end{align*}
Notice, however, that many four-dimensional spacetimes can reduce to the same base space. For example, spacetimes with spatial hypersurfaces diffeomorphic to $\mathbb{S}^{2}\times \mathbb{S}^{1}$ and $\mathbb{S}^{3}$ both reduce to $2+1$ spacetimes with hypersurfaces diffeomorphic to $\mathbb{S}^{2}$. Moreover, since the evolution equations and constraints \eqref{eqn:e11_evo}-\eqref{eqn:HamiltonianConstraint_U1} are independent of $x^{3}$, there is nothing distinguishing these cases at the level of the equations. \newline \par

It turns out that the topology of the principal bundle is connected to the curvature of the base space through a constraint on the initial data, which is preserved by the evolution equations. It is exactly this constraint which distinguishes between $U(1)$ reductions of trivial and non-trivial bundles. Indeed, constraints of this form have been discussed extensively by Choquet-Bruhat and Moncrief in their work on $U(1)$-symmetric spacetimes, see for example \cites{Moncrief:1986,moncrief:1990,ChoquetBruhat:2003,ChoquetBruhatMoncrief:2001} and \cite{ChoquetBruhat:2009}*{Chapter 16 and Appendix VII}. In the remainder of this section, we will derive this topological constraint and show that it is propagated by the field equations. \newline \par 

To begin, following \cites{Miller:1980}, we define a connection one-form\footnote{Note that this agrees with the form $\eta_{A}$ which is ultimately obtained by Geroch in \cite{Geroch:1971}.} $\eta$ on $M$
\begin{align*}
    \eta = \frac{1}{\lambda}\xi^{\flat} = \frac{1}{\lambda^{\frac{1}{2}}}e_{3}^{\flat} = \frac{1}{\lambda^{\frac{1}{2}}}\theta^{3}_{\mu},
\end{align*}
where $\flat$ denotes the musical isomorphism \cite{Gourgoulhon:2012}*{\S 2.3.3}. In frame components, we then have
\begin{align}
\label{eqn:connection_form_frame_components}
    \eta_{a} = \frac{1}{\lambda^{\frac{1}{2}}}\delta^{3}_{a}.
\end{align}
The curvature form $F$ is obtained by taking the exterior derivative of $\eta$
\begin{align*}
    F := d\eta.
\end{align*}
In frame components, this becomes
\begin{align}
\label{eqn:curvature_form_frame_components}
    F_{ab} &= 2\nabla_{[a}(\lambda^{\frac{-1}{2}}\delta^{3}_{b]}) \nonumber \\
    &= e_{a}(\lambda^{\frac{-1}{2}}\delta^{3}_{b}) - \lambda^{\frac{-1}{2}}\tensor{\omega}{_a^3_b} - e_{b}(\lambda^{\frac{-1}{2}}\delta^{3}_{a}) + \lambda^{\frac{-1}{2}}\tensor{\omega}{_b^3_a} \nonumber \\
    &= \lambda^{\frac{-3}{2}}\delta^{3}_{[a}e_{b]}(\lambda) + \lambda^{\frac{-1}{2}}\tensor{c}{_b^3_a}.
\end{align}
From this identity, it is straightforward to verify that the curvature form is constant on the group orbits. Next, by expressing the Killing equation \eqref{eqn:Killing_eqn} in terms of the frame, we obtain
\begin{align*}
    0 &= \nabla_{a}(\lambda^{\frac{1}{2}}\delta^{3}_{b}) + \nabla_{a}(\lambda^{\frac{1}{2}}\delta^{3}_{b}) \nonumber \\
    &= \lambda^{\frac{-1}{2}}\delta^{3}_{(a}e_{b)}(\lambda) - \lambda^{\frac{1}{2}}(\tensor{\omega}{_a^3_b} + \tensor{\omega}{_b^3_a}).
\end{align*}
Setting $a=3$ in the above then yields
\begin{align}
\label{eqn:Killing_identity_curvature_form_derivation}
    \frac{1}{2}\lambda^{\frac{-1}{2}}e_{b}(\lambda) - \lambda^{\frac{1}{2}}(\tensor{\omega}{_3^3_b}) = 0,
\end{align}
where we have used the antisymmetry of the connection coefficients
\begin{align*}
    \omega_{abc} = -\omega_{acb}.
\end{align*}
Using \eqref{eqn:curvature_form_frame_components} we then find that $F_{3b}$ is given by
\begin{align*}
    F_{3b} = \frac{1}{2}\lambda^{\frac{-3}{2}}e_{b}(\lambda) - \lambda^{\frac{-1}{2}}\tensor{\omega}{_3^3_b} = 0
\end{align*}
where the last equality follows from \eqref{eqn:Killing_identity_curvature_form_derivation}. Hence the curvature form has the following non-trivial components
\begin{align*}
    F_{01} &= -\lambda^{\frac{-1}{2}}\tensor{c}{_0^3_1} = \lambda^{\frac{-1}{2}}(\sigma_{13} - \Omega_{2}) = 2\lambda^{\frac{-1}{2}}\sigma_{13}, \\
    F_{02} &= -\lambda^{\frac{-1}{2}}\tensor{c}{_0^3_2} = \lambda^{\frac{-1}{2}}(\sigma_{23} + \Omega_{1}) = 2\lambda^{\frac{-1}{2}}\sigma_{23}, \\
    F_{12} &= -\lambda^{\frac{-1}{2}}\tensor{c}{_1^3_2} = -\lambda^{\frac{-1}{2}}n_{33},
\end{align*}
from which it is clear that $F$ descends to a two-form on $S$. In symmetry adapted local coordinates \eqref{eqn:symmetry_adapted_frame_ansatz}, we then have
\begin{align*}
    F_{01} &= 2e_{3}^{3}\sigma_{13}, \\
    F_{02} &= 2e_{3}^{3}\sigma_{23}, \\
    F_{12} &= -e_{3}^{3}n_{33}.
\end{align*}
Following \cite{ChoquetBruhat:2009}*{Chapter 16, \S 4.2}, the curvature form must satisfy
\begin{align}
    \int_{\Sigma_{t}}F = \int_{\Sigma_{t}}F_{12} \sqrt{\gamma} \;dx^{1}\;dx^{2} = 2\pi \ell
\end{align}
where $\Sigma_{t}$ is a $t=$ constant hypersurface of the 2+1 dimensional base space $S$, $\ell$ is a topological invariant called the \textit{Chern number} of the principal bundle, and $\gamma$ is the determinant of the metric on the two-dimensional spatial surfaces. In particular $\ell=0$ for trivial bundles\footnote{That is bundles of the form $M = \mathbb{R} \times \Sigma \times \mathbb{S}^{1}$, where $\Sigma$ is a two dimensional space and we have identified $U(1)$ and $\mathbb{S}^{1}$.}, for example $\mathbb{T}^{3}$ or $\mathbb{S}^{2}\times \mathbb{S}^{1}$, while for non-trivial bundles we have $\ell>0$. \newline \par

In frame components, the integral becomes
\begin{align}
\label{eqn:Topological_Integral_Constraint}
     \int_{\Sigma_{t}} -e^{3}_{3}n_{33} \frac{1}{e_{1}^{1}e_{2}^{2} - e_{1}^{2}e_{2}^{1}} \;dx^{1}\;dx^{2} = 2\pi \ell.
\end{align}
where we have used the fact $\frac{1}{e_{1}^{1}e_{2}^{2} - e_{1}^{2}e_{2}^{1}} = \sqrt{\gamma}$. 

\begin{rem}[Topology of Polarised Spacetimes]
    Recall from \eqref{eqn:Polarised_U1_Frame_Simplifications_a}-\eqref{eqn:Polarised_U1_Frame_Simplifications_b}, that a polarised $U(1)$ spacetime must have $n_{33} = 0$. In particular, this means that \eqref{eqn:Topological_Integral_Constraint} is always zero and, hence, polarised spacetimes must have vanishing Chern number $\ell=0$. That is to say, only trivial principle bundles can be polarised. \newline \par
\end{rem}

\begin{rem}[Moncrief Form of Topological Constraint]
    From the frame decomposition \eqref{eqn:MoncriefBerger_U1_FrameDecomposition} of the Berger-Moncrief parameterisation of a $U(1)$-symmetric spacetime \eqref{eqn:BergerMoncrief_U1_Metric} in Appendix \ref{app:BergerMoncrief_FrameExample}, we see that
    \begin{align*}
        e^{3}_{3}n_{33} \frac{1}{e_{1}^{1}e_{2}^{2} - e_{1}^{2}e_{2}^{1}} = \del_{v}\beta_{1} - \del_{u}\beta_{2}.
    \end{align*}
    This is exactly the conjugate momentum of the twist potential (for a trivial bundle) which Moncrief uses to obtain the same topological constraint on the initial data, see \cite{Moncrief:1986}*{Pages 125-127}. In particular, this verifies our topological constraint \eqref{eqn:Topological_Integral_Constraint} coincides with Moncrief's form of the constraint, at least for trivial bundles.
\end{rem}

We want to show that this integral is conserved by the frame evolution equations. To prove this, we will use the following lemma
\begin{lem}
    The following identities hold
    \begin{align}
    \label{eqn:Integral_lemma_ident1}
        e_{2}(e_{3}^{3}\alpha \sigma_{13}) - e_{1}(e_{3}^{3}\alpha \sigma_{23}) + e_{3}^{3}\alpha \tensor{c}{_1^{\mathcal{c}}_2}\sigma_{\mathcal{c}3} &= -e_{1}(\sigma_{23})\alpha e_{3}^{3} + e_{2}(\sigma_{13})\alpha e_{3}^{3} + 2\alpha e_{3}^{3}n_{13}\sigma_{13} \nonumber \\
        &+ 2\alpha e_{3}^{3}n_{23}\sigma_{23} + \alpha e_{3}^{3} \dot{u}_{2} \sigma_{13} - \alpha \dot{u}_{1}e_{3}^{3}\sigma_{23}, \\
        \label{eqn:Integral_lemma_ident2}
        \del_{t}\Big((e_{1}^{1}e_{2}^{2} - e_{1}^{2}e_{2}^{1})^{-1}\Big) &= \frac{\alpha}{e_{1}^{2}e_{2}^{2} - e_{1}^{2}e_{2}^{1}}(2H - \sigma_{33}), \\ 
        \label{eqn:Integral_lemma_ident3}
        \del_{t}\left(\frac{e_{3}^{3}n_{33}}{e_{1}^{1}e_{2}^{2}-e_{1}^{2}e_{2}^{1}}\right) &= 2\sqrt{\gamma}\tilde{\nabla}_{\mathcal{a}}\big(e_{3}^{3}\alpha\sigma_{\mathcal{b}3}\eps^{\mathcal{b}\mathcal{a}}\big) 
    \end{align}
    where $\tilde{\nabla}_{\mathcal{a}}$ is the spatial covariant derivative on $S$,$\eps^{\mathcal{a}\mathcal{b}} = \eps^{0\mathcal{a}\mathcal{b}3}$ is the two dimensional volume form with respect to our orthonormal frame
    \begin{align*}
        \eps^{12} = -\eps^{21} = 1,
    \end{align*}
    and we have assumed the shift vector is zero, $\beta^{\Omega}=0$.
\end{lem}

\begin{proof}$\;$
\subsubsection*{Identity 1.}
    First, let us establish \eqref{eqn:Integral_lemma_ident1}. Expanding the left hand side and replacing the derivatives of $e_{3}^{3}$ and $\dot{u}_{\mathcal{a}}$ using \eqref{eqn:C1_Constraint_133}-\eqref{eqn:C1_Constraint_233} and \eqref{eqn:C3_Constraint_1}-\eqref{eqn:C3_Constraint_2} yields 
    \begin{align} 
    \label{eqn:int_ident_1a}
        e_{2}(e_{3}^{3}\alpha \sigma_{13}) - e_{1}(e_{3}^{3}\alpha \sigma_{23}) + e_{3}^{3}\alpha \tensor{c}{_1^{\mathcal{c}}_2}\sigma_{\mathcal{c}3} &= \big[a_{2}e_{3}^{3} + n_{13}\big]\alpha\sigma_{13} + \alpha \dot{u}_{2}e_{3}^{3}\sigma_{23} + e_{2}(\sigma_{13})\alpha e_{3}^{3} \nonumber \\
        &- \big[a_{1}e_{3}^{3} -n_{23}e_{3}^{3}\big]\alpha\sigma_{23} - \alpha \dot{u}_{1}\sigma_{23} - e_{1}(\sigma_{23})\alpha e_{3}^{3} \nonumber \\
        &+ e_{3}^{3}\alpha\tensor{c}{_1^1_2}\sigma_{13} + e_{3}^{3}\alpha\tensor{c}{_1^2_2}\sigma_{23}.
    \end{align}
    Next, using the expressions for the spatial frame commutator \eqref{eqn:mixed_commutator_decomposition}, we find
    \begin{align*}
        \tensor{c}{_1^1_2} = n_{13} - a_{2}, \quad \tensor{c}{_1^2_2} = n_{23} + a_{1}.
    \end{align*}
    Substituting the above into \eqref{eqn:int_ident_1a} and collecting terms then leads to the identity \eqref{eqn:Integral_lemma_ident1}.

\subsubsection*{Identity 2}
Expanding the left hand side of \eqref{eqn:Integral_lemma_ident2} and replacing the time derivatives using \eqref{eqn:e11_evo}-\eqref{eqn:e12_evo} and \eqref{eqn:e21_evo}-\eqref{eqn:e22_evo} yields the identity
\begin{align*}
    \del_{t}\Big((e_{1}^{1}e_{2}^{2} - e_{1}^{2}e_{2}^{1})^{-1}\Big) &= \frac{-1}{(e_{1}^{1}e_{2}^{2} - e_{1}^{2}e_{2}^{1})^{2}}\Big(\del_{t}(e_{1}^{1})e_{2}^{2} + \del_{t}(e_{2}^{2})e_{1}^{2} - e_{0}(e_{1}^{2})e_{2}^{1} - \del_{t}(e_{2}^{1})e_{1}^{2}\Big) \nonumber \\
    &= \frac{-\alpha}{(e_{1}^{1}e_{2}^{2} - e_{1}^{2}e_{2}^{1})^{2}}\Big(-(H+\sigma_{11})e_{1}^{1}e_{2}^{2} - (\sigma_{12}+\Omega_{3})e_{2}^{1}e_{2}^{2} \nonumber \\
    &-(\sigma_{12}-\Omega_{3})e_{1}^{2}e_{1}^{1} - (H+\sigma_{22})e_{2}^{2}e_{1}^{1} + (H + \sigma_{11})e_{1}^{2}e_{2}^{1} \nonumber \\
    &+ (\sigma_{12} + \Omega_{3})e_{2}^{2}e_{2}^{1} + (\sigma_{12} - \Omega_{3})e_{1}^{1}e_{1}^{2} + (H + \sigma_{22})e_{2}^{1}e_{1}^{2}\Big) \nonumber \\
    &= \frac{\alpha(2H +\sigma_{11} + \sigma_{22})}{e_{1}^{1}e_{2}^{2} - e_{1}^{2}e_{2}^{2}} \nonumber \\
    &= \frac{\alpha(2H -\sigma_{33})}{e_{1}^{1}e_{2}^{2} - e_{1}^{2}e_{2}^{2}},
\end{align*}
where we have used the fact $\sigma_{AB}$ is trace-free to obtain the last equality.

\subsubsection*{Identity 3}
First, we split the time derivatives on the left hand side of \eqref{eqn:Integral_lemma_ident3} to obtain
\begin{align*}
     \del_{t}\left(\frac{e_{3}^{3}n_{33}}{e_{1}^{1}e_{2}^{2}-e_{1}^{2}e_{2}^{1}}\right) &= e_{3}^{3}n_{33}\del_{t}\Big((e_{1}^{1}e_{2}^{2}-e_{1}^{2}e_{2}^{1})^{-1}\Big) + \frac{n_{33}}{e_{1}^{1}e_{2}^{2}-e_{1}^{2}e_{2}^{1}}\del_{t}(e_{3}^{3}) + \frac{e_{3}^{3}}{e_{1}^{1}e_{2}^{2}-e_{1}^{2}e_{2}^{1}}\del_{t}(n_{33}).
\end{align*}
Next, by substituting \eqref{eqn:e33_evo}, \eqref{eqn:n33_evo}, and \eqref{eqn:Integral_lemma_ident1}-\eqref{eqn:Integral_lemma_ident2} into the above, we obtain
\begin{align}
\label{eqn:int_ident_3a}
    \del_{t}\left(\frac{e_{3}^{3}n_{33}}{e_{1}^{1}e_{2}^{2}-e_{1}^{2}e_{2}^{1}}\right) &= \frac{1}{e_{1}^{1}e_{2}^{2} - e_{1}^{2}e_{2}^{2}}\Big[\alpha(2H -\sigma_{33})e_{3}^{3}n_{33} - \alpha e_{3}^{3} n_{33}(H + \sigma_{33}) +  2e_{2}(e_{3}^{3}\alpha \sigma_{13}) \nonumber \\
    &- 2e_{1}(e_{3}^{3}\alpha \sigma_{23}) + 2e_{3}^{3}\alpha \tensor{c}{_1^{\mathcal{c}}_2}\sigma_{\mathcal{c}3} + \alpha e_{3}^{3}n_{33}(2\sigma_{33} -H)\Big] \nonumber \\
    &= \frac{2}{e_{1}^{1}e_{2}^{2} - e_{1}^{2}e_{2}^{2}}\Big[e_{2}(e_{3}^{3}\alpha \sigma_{13}) - e_{1}(e_{3}^{3}\alpha \sigma_{23}) + e_{3}^{3}\alpha \tensor{c}{_1^{\mathcal{c}}_2}\sigma_{\mathcal{c}3}\Big].
\end{align}
Next, observe that
\begin{align*}
    \tilde{\nabla}_{\mathcal{a}}(e_{3}^{3}\alpha \eps^{\mathcal{b}\mathcal{a}}\sigma_{\mathcal{b}3}) &= -e_{1}(e_{3}^{3}\alpha\sigma_{23}) + e_{2}(e_{3}^{3}\alpha\sigma_{13}) + e_{3}^{3}\alpha\Big(\sigma_{13}\tensor{\omega}{_1^1_2} - \sigma_{23}\tensor{\omega}{_2^2_1}\Big) \nonumber \\
    &= -e_{1}(e_{3}^{3}\alpha\sigma_{23}) + e_{2}(e_{3}^{3}\alpha\sigma_{13}) + e_{3}^{3}\alpha\Big(\sigma_{\mathcal{a}3}\tensor{\omega}{_1^{\mathcal{a}}_2} - \sigma_{\mathcal{a}3}\tensor{\omega}{_2^{\mathcal{a}}_1}\Big) \nonumber \\
    &= -e_{1}(e_{3}^{3}\alpha\sigma_{23}) + e_{2}(e_{3}^{3}\alpha\sigma_{13}) + e_{3}^{3}\alpha\tensor{c}{_1^{\mathcal{a}}_2}\sigma_{\mathcal{a}3}
\end{align*}
where we have used the fact $\tensor{\omega}{_a^1_1} = \tensor{\omega}{_a^2_2} = 0$. Substituting the above into \eqref{eqn:int_ident_3a} then yields \eqref{eqn:Integral_lemma_ident3}, as required.
\end{proof}

Now, we can prove the following 
\begin{prop}
    The integral constraint \eqref{eqn:Topological_Integral_Constraint} is conserved by the frame evolution equations with $\beta^{\Omega}=0$.
\end{prop}
\begin{proof}
    Taking the time derivative of \eqref{eqn:Topological_Integral_Constraint} and using \eqref{eqn:Integral_lemma_ident3}, we find
    \begin{align*}
        \del_{t}\Big(\int_{\Sigma_{t}} -e^{3}_{3}n_{33} \frac{1}{e_{1}^{1}e_{2}^{2} - e_{1}^{2}e_{2}^{1}} \;dx^{1}\;dx^{2}\Big) &= \int_{\Sigma_{t}} -\del_{t}\Big(e^{3}_{3}n_{33} \frac{1}{e_{1}^{1}e_{2}^{2} - e_{1}^{2}e_{2}^{1}}\Big) \;dx^{1}\;dx^{2} \nonumber \\
        &= \int_{\Sigma_{t}} -2\sqrt{\gamma}\tilde{\nabla}_{\mathcal{a}}(e_{3}^{3}\alpha \sigma_{\mathcal{b}3}\eps^{\mathcal{b}\mathcal{a}}) \;dx^{1}\;dx^{2} \nonumber \\
        &= 0
    \end{align*}
    where the last equality follows from Stokes' theorem \cite{Wald:1984}*{Theorem B.2.1}.
\end{proof}
Thus, it is sufficient to specify initial data satisfying \eqref{eqn:Topological_Integral_Constraint}. 

\section{A Symmetric Hyperbolic Formulation of the Orthonormal Frame Equations}
\label{sec:Sym_Hyp_U1_Eqns}

Our next goal is to write a well-posed formulation of the equations \eqref{eqn:e11_evo}-\eqref{eqn:HamiltonianConstraint_U1}. First, by subtracting $(\mathcal{C}_{M})_{1}$ and  $(\mathcal{C}_{M})_{2}$ from \eqref{eqn:a1_evo} and \eqref{eqn:a2_evo}, respectively, we obtain 
\begin{align}
\label{eqn:a1_evo2}
    e_{0}(a_{1}) &= e_{1}(H) - \frac{1}{2}e_{1}(\sigma_{11}) - \frac{1}{2}e_{2}(\sigma_{12}) + \frac{1}{2}e_{2}(\Omega_{3}) - H(\dot{u}_{1} + a_{1}) + \sigma_{11}(\frac{1}{2}\dot{u}_{1} + 2a_{1}) \nonumber \\
    &+ \sigma_{12}(\frac{1}{2}\dot{u}_{2} + 2a_{2}) + \Omega_{3}(\frac{1}{2}\dot{u}_{2} - a_{2}) - n_{23}(\sigma_{22} - \sigma_{33}) - n_{13}\sigma_{12} - n_{33}\sigma_{23} + T_{01} , \\
\label{eqn:a2_evo2}
    e_{0}(a_{2}) &= -\frac{1}{2}e_{1}(\sigma_{12}) - \frac{1}{2}e_{2}(\sigma_{22}) + e_{2}(H) - \frac{1}{2}e_{1}(\Omega_{3}) - H(\dot{u}_{2} + a_{2}) + \sigma_{22}(\frac{1}{2}\dot{u}_{2} + 2a_{2}) \nonumber \\
    &+ \sigma_{12}(\frac{1}{2}\dot{u}_{1} + 2a_{1}) - \Omega_{3}(\frac{1}{2}\dot{u}_{1} - a_{1}) - n_{13}(\sigma_{33} - \sigma_{11}) + n_{23}\sigma_{12} + n_{33}\sigma_{13} + T_{02}.
\end{align}
Similarly, adding $\frac{1}{12}(\mathcal{C}_{H})$ to \eqref{eqn:H_evo} yields 
\begin{align}
\label{eqn:H_evo2}
e_{0}(H) &= \frac{1}{3}e_{1}(a_{1}) + \frac{1}{3}e_{2}(a_{2}) + \frac{1}{3}e_{1}(\dot{u}_{1}) + \frac{1}{3}e_{2}(\dot{u}_{2}) - \frac{1}{2}H^{2} + \frac{1}{3}\dot{u}_{1}(\dot{u}_{1} - 2a_{1}) + \frac{1}{3}\dot{u}_{2}(\dot{u}_{2} - 2a_{2}) \nonumber \\
&- \frac{5}{12}\sigma_{AB}\sigma^{AB} - \frac{1}{2}(a_{1}^{2} + a_{2}^{2}) - \frac{1}{12}(2n_{13}^{2} + 2n_{23}^{2} + \frac{1}{2}n_{33}^{2}) - \frac{1}{3}T_{00} - \frac{1}{6}\tensor{T}{_A^A}.
\end{align}

Next, in the vein of \cite{Elst_et_al:2001}, we introduce the new variables
\begin{equation}
    \begin{gathered}
    \label{eqn:decomposition_variables}
        \sigma_{+} = \frac{1}{2}(\sigma_{22} + \sigma_{33}) = \frac{-1}{2}\sigma_{11}, \quad \sigma_{-} = \frac{1}{2\sqrt{3}}(\sigma_{22} - \sigma_{33}),  \\
        \sigma_{\times} = \frac{1}{\sqrt{3}}\sigma_{23}, \quad
        \sigma_{2} = \frac{1}{\sqrt{3}}\sigma_{13}, \quad \sigma_{3} = \frac{1}{\sqrt{3}}\sigma_{12}, \\
        n_{+} = \frac{1}{2\sqrt{3}}n_{33}, \quad n_{\times} = \frac{1}{\sqrt{3}}n_{23}, \quad \hat{n} = \frac{1}{\sqrt{3}}n_{13}, \\
        \hat{a}_{1} = \frac{1}{\sqrt{3}}a_{1}, \quad \hat{a}_{2} = \frac{1}{\sqrt{3}}a_{2}.
    \end{gathered}
\end{equation}
In terms of these variables, a polarised spacetime corresponds to
    \begin{align*}
    n_{+} = \sigma_{\times} = \sigma_{2} = 0.
    \end{align*}
We also note the following useful identities,
\begin{equation*}
\begin{gathered}
    \sigma_{11} = -2\sigma_{+}, \quad \sigma_{22} = \sigma_{+} +\sqrt{3}\sigma_{-}, \quad \sigma_{33} = \sigma_{+} - \sqrt{3}\sigma_{-}, \\
    \sigma_{AB}\sigma^{AB} = 6(\sigma_{+}^{2} + \sigma_{-}^{2} + \sigma_{2}^{2} + \sigma_{3}^{2} + \sigma_{\times}^{2}).
\end{gathered}
\end{equation*}
Now, in terms of the variables \eqref{eqn:decomposition_variables}, the evolution equations become
\begin{align}
\label{eqn:sigma_plus_evo}
    e_{0}(\sigma_{+}) &= \frac{1}{\sqrt{3}}e_{1}(\hat{a}_{1}) - \frac{1}{2\sqrt{3}}e_{2}(\hat{a}_{2}) - \frac{\sqrt{3}}{2}e_{2}(\hat{n}) - \frac{1}{3}e_{1}(\dot{u}_{1}) + \frac{1}{6}e_{2}(\dot{u}_{2}) -3H\sigma_{+} \nonumber \\
    &- \frac{1}{3}(\dot{u}_{1}^{2} + \sqrt{3}\hat{a}_{1}\dot{u}_{1}) + \frac{1}{6}(\dot{u}_{2}^{2} + \sqrt{3}\hat{a}_{2}\dot{u}_{2}) - \frac{\sqrt{3}}{2}\hat{n}(\dot{u}_{2} - 2\sqrt{3}\hat{a}_{2}) + (\sqrt{3}\sigma_{3}\Omega_{3} + 3\sigma_{2}^{2}) \nonumber \\
    &- (2n_{\times}^{2} + 2n_{+}^{2} - \hat{n}^{2}) - \frac{1}{2}T_{11} + \frac{1}{6}\tensor{T}{_A^A} , \\
\label{eqn:sigma_minus_evo}
    e_{0}(\sigma_{-}) &= -e_{1}(n_{\times}) - \frac{1}{2}e_{2}(\hat{a}_{2} ) +\frac{1}{2}e_{2}(\hat{n}) + \frac{1}{2\sqrt{3}}e_{2}(\dot{u}_{2}) -3H\sigma_{-} + \frac{1}{2\sqrt{3}}(\dot{u}_{2}^{2} + \sqrt{3}\hat{a}_{2}\dot{u}_{2})  \nonumber \\
    &- n_{\times}(\dot{u}_{1} -2\sqrt{3}\hat{a}_{1})+ \frac{1}{2}\hat{n}(\dot{u}_{2} - 2\sqrt{3}\hat{a}_{2}) + (\sigma_{3}\Omega_{3} -2\sqrt{3}\sigma_{\times}^{2} - \sqrt{3}\sigma_{2}^{2}) \nonumber \\
    &+ \sqrt{3}(\hat{n}^{2} + 2n_{+}^{2}) + \frac{1}{2\sqrt{3}}(T_{22} - T_{33}), \\
\label{eqn:n_times_evo}
    e_{0}(n_{\times}) &= -e_{1}(\sigma_{-}) + \frac{1}{2}e_{2}(\sigma_{3}) + \frac{1}{2\sqrt{3}}e_{2}(\Omega_{3}) + n_{\times}(2\sigma_{+} - H) + \hat{n}(\sqrt{3}\sigma_{3} + \Omega_{3}) - \dot{u}_{1}\sigma_{-} \nonumber \\
    &+ \frac{1}{2}\dot{u}_{2}(\sigma_{3} + \frac{1}{\sqrt{3}}\Omega_{3}), \\
\label{eqn:sigma_times_evo}
    e_{0}(\sigma_{\times}) &= -e_{1}(n_{+}) -3H\sigma_{\times} + n_{+}(2\sqrt{3}\hat{a}_{1} - \dot{u}_{1}) + 2\sqrt{3}\sigma_{\times}\sigma_{-} + \sigma_{2}(\Omega_{3} + \sqrt{3}\sigma_{3}) - 2\sqrt{3}n_{\times}n_{+} \nonumber \\
    &+ \frac{1}{\sqrt{3}}T_{23}, \\
\label{eqn:n_plus_evo}
    e_{0}(n_{+}) &= -e_{1}(\sigma_{\times}) + e_{2}(\sigma_{2}) + n_{+}(2\sigma_{+} - 2\sqrt{3}\sigma_{-} - H) + 2\sqrt{3}\sigma_{2}\hat{n} + 2\sqrt{3}n_{\times}\sigma_{\times} - \dot{u}_{1}\sigma_{\times} \nonumber \\
    &+ \dot{u}_{2}\sigma_{2}, \\
\label{eqn:n_hat_evo}
    e_{0}(\hat{n}) &= \frac{1}{2\sqrt{3}}e_{1}(\Omega_{3}) - \frac{1}{2}e_{1}(\sigma_{3}) - \frac{\sqrt{3}}{2}e_{2}(\sigma_{+}) + \frac{1}{2}e_{2}(\sigma_{-}) - \hat{n}(\sigma_{+} + \sqrt{3}\sigma_{-} + H) \nonumber \\
    &- n_{\times}(\Omega_{3} - \sqrt{3}\sigma_{3}) + \frac{1}{2\sqrt{3}}\dot{u}_{1}(\Omega_{3} - \sqrt{3}\sigma_{3}) + \frac{1}{2\sqrt{3}}\dot{u}_{2}(-3\sigma_{+} + \sqrt{3}\sigma_{-}), \\
\label{eqn:sigma_3_evo}
    e_{0}(\sigma_{3}) &= \frac{-1}{2}e_{1}(\hat{n}) - \frac{1}{2}e_{1}(\hat{a}_{2}) - \frac{1}{2}e_{2}(\hat{a}_{1})  + \frac{1}{2}e_{2}(n_{\times}) + \frac{1}{2\sqrt{3}}e_{1}(\dot{u}_{2}) + \frac{1}{2\sqrt{3}}e_{2}(\dot{u}_{1})  - 3H\sigma_{3} \nonumber \\
    &+ \frac{1}{2}n_{\times}(\dot{u}_{2} - 2\sqrt{3}\hat{a}_{2}) - \Omega_{3}(\sigma_{-} + \sqrt{3}\sigma_{+}) -2\sqrt{3}\sigma_{2}\sigma_{\times} - \frac{1}{2}\hat{n}(\dot{u}_{1} - 2\sqrt{3}\hat{a}_{1}) \nonumber \\
    &- 2\sqrt{3}\hat{n}n_{\times} + \frac{1}{2\sqrt{3}}\dot{u}_{1}(\dot{u}_{2} + \sqrt{3}\hat{a}_{2}) + \frac{1}{2\sqrt{3}}\dot{u}_{2}(\dot{u}_{1} + \sqrt{3}\hat{a}_{1}) + \frac{1}{\sqrt{3}}T_{12}, \\
\label{eqn:sigma_2_evo}
    e_{0}(\sigma_{2}) &= e_{2}(n_{+}) - 3H\sigma_{2} + n_{+}(\dot{u}_{2} - 2\sqrt{3}\hat{a}_{2}) + \sigma_{2}(\sqrt{3}\sigma_{-} - 3\sigma_{+}) + \sigma_{\times}(\sqrt{3}\sigma_{3} - \Omega_{3}) \nonumber \\
    &- 2\sqrt{3}\hat{n}n_{+} + \frac{1}{\sqrt{3}}T_{13}, \\
\label{eqn:ahat_1_evo}
    e_{0}(\hat{a}_{1}) &= \frac{1}{\sqrt{3}}e_{1}(H) + \frac{1}{\sqrt{3}}e_{1}(\sigma_{+}) - \frac{1}{2}e_{2}(\sigma_{3}) + \frac{1}{2\sqrt{3}}e_{2}(\Omega_{3}) - H(\frac{1}{\sqrt{3}}\dot{u}_{1} + \hat{a}_{1}) \nonumber \\
    &- 2\sigma_{+}(\frac{1}{2\sqrt{3}}\dot{u}_{1} + 2\hat{a}_{1}) + \sigma_{3}(\frac{1}{2}\dot{u}_{2} + 2\sqrt{3}\hat{a}_{2}) + \Omega_{3}(\frac{1}{2\sqrt{3}}\dot{u}_{2} - \hat{a}_{2}) - 2\sqrt{3}n_{\times}\sigma_{-} \nonumber \\
    &- \sqrt{3}\hat{n}\sigma_{3} -2\sqrt{3}n_{+}\sigma_{\times} + \frac{1}{\sqrt{3}}T_{01}, \\
\label{eqn:ahat_2_evo}
    e_{0}(\hat{a}_{2}) &= -\frac{1}{2}e_{1}(\sigma_{3}) - \frac{1}{2\sqrt{3}}e_{2}(\sigma_{+}) - \frac{1}{2}e_{2}(\sigma_{-}) + \frac{1}{\sqrt{3}}e_{2}(H) - \frac{1}{2\sqrt{3}}e_{1}(\Omega_{3}) - H(\frac{1}{\sqrt{3}}\dot{u}_{2} + \hat{a}_{2}) \nonumber \\
    &+ (\sigma_{+} +\sqrt{3}\sigma_{-})(\frac{1}{2\sqrt{3}}\dot{u}_{2} + 2\hat{a}_{2}) + \sigma_{3}(\frac{1}{2}\dot{u}_{1} + 2\sqrt{3}\hat{a}_{1}) - \Omega_{3}(\frac{1}{2\sqrt{3}}\dot{u}_{1} - \hat{a}_{1}) \nonumber \\ 
    &- \hat{n}(3\sigma_{+} - \sqrt{3}\sigma_{-}) + \sqrt{3}n_{\times}\sigma_{3} + 2\sqrt{3}n_{+}\sigma_{2} + \frac{1}{\sqrt{3}}T_{02} , \\
\label{eqn:H_U1_evo}
    e_{0}(H) &= \frac{1}{\sqrt{3}}e_{1}(\hat{a}_{1}) + \frac{1}{\sqrt{3}}e_{2}(\hat{a}_{2}) + \frac{1}{3}e_{1}(\dot{u}_{1}) + \frac{1}{3}e_{2}(\dot{u}_{2}) - \frac{1}{2}H^{2} + \frac{1}{3}\dot{u}_{1}(\dot{u}_{1} - 2\sqrt{3}\hat{a}_{1}) \nonumber \\
    &+ \frac{1}{3}\dot{u}_{2}(\dot{u}_{2} - 2\sqrt{3}\hat{a}_{2}) - \frac{5}{2}(\sigma_{+}^{2} + \sigma_{-}^{2} + \sigma_{2}^{2} + \sigma_{3}^{2} + \sigma_{\times}^{2}) - \frac{3}{2}(\hat{a}_{1}^{2} + \hat{a}_{2}^{2}) - \frac{1}{2}(\hat{n}^{2} + n_{\times}^{2} + n_{+}^{2}) \nonumber \\
    &- \frac{1}{3}T_{00} - \frac{1}{6}\tensor{T}{_A^A}, \\
\label{eqn:e11_U1_evo}
    e_{0}(e^{1}_{1}) &= -(H - 2\sigma_{+})e^{1}_{1} - (\sqrt{3}\sigma_{3} + \Omega_{3})e_{2}^{1}, \\
\label{eqn:e12_U1_evo}
    e_{0}(e_{1}^{2}) &= -(H - 2\sigma_{+})e_{1}^{2} - (\sqrt{3}\sigma_{3}  + \Omega_{3})e_{2}^{2}, \\
\label{eqn:e13_U1_evo} 
    e_{0}(e_{1}^{3}) &= -(H - 2\sigma_{+})e_{1}^{3} - (\sqrt{3}\sigma_{3} + \Omega_{3})e_{2}^{3} - 2\sqrt{3}\sigma_{2}e_{3}^{3}, \\
\label{eqn:e21_U1_evo}
    e_{0}(e_{2}^{1}) &= -(\sqrt{3}\sigma_{3}  - \Omega_{3})e_{1}^{1} - (H + \sigma_{+} + \sqrt{3}\sigma_{-})e_{2}^{1}, \\
\label{eqn:e22_U1_evo}
    e_{0}(e_{2}^{2}) &= -(\sqrt{3}\sigma_{3}  - \Omega_{3})e_{1}^{2} - (H + \sigma_{+} + \sqrt{3}\sigma_{-})e_{2}^{2}, \\
\label{eqn:e23_U1_evo} 
    e_{0}(e_{2}^{3}) &= -(\sqrt{3}\sigma_{3} - \Omega_{3})e_{1}^{3} - (H + \sigma_{+} + \sqrt{3}\sigma_{-})e_{2}^{3} - 2\sqrt{3}\sigma_{\times}e_{3}^{3}, \\
\label{eqn:e33_U1_evo} 
    e_{0}(e_{3}^{3}) &= -(H + \sigma_{+} - \sqrt{3}\sigma_{-})e_{3}^{3}.
\end{align}

Similarly, the constraints reduce to
\begin{align}
\label{eqn:C1_112_U1}
0 =\tensor{(\mathcal{C}_{1})}{_1^1_2} &:= 2e_{[1}(e_{2]}^{1}) - 2\sqrt{3}\hat{a}_{[1}e_{2]}^{1} - \sqrt{3}(\hat{n}e_{1}^{1} + n_{\times}e_{2}^{1}) , \\
\label{eqn:C1_122_U1}
0 =\tensor{(\mathcal{C}_{1})}{_1^2_2} &:= 2e_{[1}(e_{2]}^{2}) - 2\sqrt{3}\hat{a}_{[1}e_{2]}^{2} - \sqrt{3}(\hat{n}e_{1}^{2} + n_{\times}e_{2}^{2}) , \\
0 =\tensor{(\mathcal{C}_{1})}{_1^3_2} &:= 2e_{[1}(e_{2]}^{3}) - 2\sqrt{3}\hat{a}_{[1}e_{2]}^{3} - (\sqrt{3}\hat{n}e_{1}^{3} + \sqrt{3}n_{\times}e_{2}^{3} + 2\sqrt{3}n_{+}e_{3}^{3}) , \\
0 = \tensor{(\mathcal{C}_{1})}{_1^3_3} &:= e_{1}(e_{3}^{3}) - \sqrt{3}\hat{a}_{1}e_{3}^{3} + \sqrt{3}n_{\times}e_{3}^{3}, \\
0 = \tensor{(\mathcal{C}_{1})}{_2^3_3} &:= e_{2}(e_{3}^{3}) - \sqrt{3}\hat{a}_{2}e_{3}^{3} - \sqrt{3}\hat{n}e_{3}^{3}, \\
0 = (\mathcal{C}_{2})_{12} &:= 2e_{[1}(\dot{u}_{2]}) + 2\sqrt{3}\hat{a}_{[1}\dot{u}_{2]} + \sqrt{3}\hat{n}\dot{u}_{1} + \sqrt{3}n_{\times}\dot{u}_{2}, \\
0 = (\mathcal{C}_{3})_{1} &:= \dot{u}_{1} - \alpha^{-1}e_{1}(\alpha), \\
0 = (\mathcal{C}_{3})_{2} &:= \dot{u}_{2} - \alpha^{-1}e_{2}(\alpha), \\
\label{eqn:C4_U1}
0 = (\mathcal{C}_{4}) &:= \int_{\Sigma_{t}} \frac{-2\sqrt{3}e_{3}^{3}n_{+}}{e_{1}^{1}e_{2}^{2} - e_{1}^{2}e_{2}^{1}} \;dx^{1}\;dx^{2} -2\pi \ell, \\
\label{eqn:CM1_U1}
0 = (\mathcal{C}_{M})_{1} &:= -2e_{1}(\sigma_{+}) + \sqrt{3}e_{2}(\sigma_{3}) - 2e_{1}(H) + 6\sqrt{3}\sigma_{+}\hat{a}_{1} - 9\sigma_{3}\hat{a}_{2} + 6n_{\times}\sigma_{-} \nonumber \\
&+ 3\hat{n}\sigma_{3} + 6n_{+}\sigma_{\times} - T_{01}, \\
\label{eqn:CM2_U1}
0 = (\mathcal{C}_{M})_{2} &:= \sqrt{3}e_{1}(\sigma_{3}) + e_{2}(\sigma_{+}) + \sqrt{3}e_{2}(\sigma_{-}) - 2e_{2}(H) - 9\sigma_{3}\hat{a}_{1} \nonumber \\
&- 3\sqrt{3}(\sigma_{+}+\sqrt{3}\sigma_{-})\hat{a}_{2} + \sqrt{3}\hat{n}(3\sigma_{+} - \sqrt{3}\sigma_{-}) - 3n_{\times}\sigma_{3} - 6n_{+}\sigma_{2} - T_{02}, \\
\label{eqn:CM3_U1}
0 = (\mathcal{C}_{M})_{3} &:= \sqrt{3}e_{1}(\sigma_{2}) + \sqrt{3}e_{2}(\sigma_{\times}) - 9\sigma_{2}\hat{a}_{1} - 9\sigma_{\times}\hat{a}_{2} - 3\hat{n}\sigma_{\times} + 3n_{\times}\sigma_{2} - T_{03}, \\
\label{eqn:CH_U1}
0 = (\mathcal{C}_{H}) &:= 4\sqrt{3}e_{1}(\hat{a}_{1}) + 4\sqrt{3}e_{2}(\hat{a}_{2}) + 6H^{2} -18(\hat{a}_{1}^{2} + \hat{a}_{2}^{2}) - 6(\hat{n}^{2} + n_{\times}^{2} + n_{+}^{2}) \nonumber \\
&- 6(\sigma_{+}^{2} + \sigma_{-}^{2} + \sigma_{2}^{2} + \sigma_{3}^{2} + \sigma_{\times}^{2}) - 2T_{00}.
\end{align}
If we assume that the gauge functions $\alpha$, $\beta^{\Omega}$, $\dot{u}_{\mathcal{a}}$, and $\Omega_{3}$ are prescribed functions of the coordinates $(t,x^{\mathcal{a}})$, then the system \eqref{eqn:sigma_plus_evo}-\eqref{eqn:e22_U1_evo} is immediately symmetric hyperbolic\footnote{Since $e_{1}$ and $e_{2}$ are linear combinations of the coordinate derivatives $\del_{x^{1}}$ and $\del_{x^{2}}$, symmetry with respect to the frame derivatives immediately implies symmetry with respect to the coordinate derivatives.}.

\subsection{Matter Equations}
The conservation of stress-energy is given by
\begin{align*}
    \nabla_{a}T^{ab} = 0.
\end{align*}
Expanding this equation yields
\begin{align*}
    e_{0}(T^{00}) + e_{A}(T^{A0}) = -\tensor{\omega}{_a^a_c}T^{c0} - \tensor{\omega}{_a^0_c}T^{ac}, \\
    e_{0}(T^{0B}) + e_{A}(T^{AB}) = -\tensor{\omega}{_a^a_c}T^{cB} - \tensor{\omega}{_a^B_c}T^{ac}. 
\end{align*}
By expressing the above in terms of the commutator variables we obtain the following evolution equations\footnote{See \cite{RöhrUggla:2005}*{Page 6} for the decomposition of the connection coefficients into the commutator variables.}
\begin{align}
\label{eqn:T_00_Frame_Evo}
    e_{0}(T^{00}) + e_{A}(T^{A0}) &= -3HT^{00} + 2(a_{A}-\dot{u}_{A})T^{A0} -\sigma_{AB}T^{AB} - H\tensor{T}{_A^A}, \\
\label{eqn:T_0B_Frame_Evo}
    e_{0}(T^{0B}) + e_{A}(T^{AB}) &= -4HT^{0B} + 3a_{C}T^{CB} - \dot{u}^{B}T^{00} - \tensor{\sigma}{_A^B}T^{A0} - a^{B}\tensor{T}{_A^A} \nonumber \\
    &-\dot{u}_{C}T^{CB}- \tensor{\eps}{^B_{CE}}(\tensor{n}{^{E}_{A}}T^{AC} + \Omega^{E}T^{0C}).
\end{align}
By substituting \eqref{eqn:U1_Frame_Simplifications}, \eqref{eqn:U(1)_Symmetric_Matter_Condition}, and \eqref{eqn:decomposition_variables} into \eqref{eqn:T_00_Frame_Evo}-\eqref{eqn:T_0B_Frame_Evo} we then obtain the $U(1)$-symmetric evolution equations for the matter
\begin{align}
\label{eqn:T_00_U1_Frame_Evo}
    e_{0}(T^{00}) + e_{\mathcal{a}}(T^{0\mathcal{a}}) &= -3HT^{00} +2(\sqrt{3}\hat{a}_{\mathcal{a}} - \dot{u}_{\mathcal{a}})T^{0\mathcal{a}} - H\tensor{T}{_A^A} + 2\sigma_{+}T^{11} - (\sigma_{+} + \sqrt{3}\sigma_{-})T^{22} \nonumber \\
    &- (\sigma_{+} -\sqrt{3}\sigma_{-})T^{33} - 2\sqrt{3}\sigma_{3}T^{12} - 2\sqrt{3}\sigma_{2}T^{13} - 2\sqrt{3}\sigma_{\times}T^{23}, \\
\label{eqn:T_01_U1_Frame_Evo}
    e_{0}(T^{01}) + e_{\mathcal{a}}(T^{\mathcal{a}1}) &= -(4H -2\sigma_{+})T^{01} -\dot{u}_{1}T^{00} + (3\sqrt{3}\hat{a}_{\mathcal{a}} - \dot{u}_{\mathcal{a}})T^{\mathcal{a}1} - \sqrt{3}\hat{a}_{1}\tensor{T}{_A^A} \nonumber \\
    &- T^{02}(\Omega_{3} + \sqrt{3}\sigma_{3}) - 2\sqrt{3}\sigma_{2}T^{03} - \sqrt{3}\hat{n}T^{12} - 2\sqrt{3}n_{+}T^{23} \nonumber \\
    &- \sqrt{3}n_{\times}(T^{22} - T^{33}) , \\
\label{eqn:T_02_U1_Frame_Evo}
    e_{0}(T^{02}) + e_{\mathcal{a}}(T^{\mathcal{a}2}) &= -(4H + \sigma_{+} + \sqrt{3}\sigma_{-})T^{02} -\dot{u}_{2}T^{00} + (3\sqrt{3}\hat{a}_{\mathcal{a}} - \dot{u}_{\mathcal{a}})T^{\mathcal{a}2} - \sqrt{3}\hat{a}_{2}\tensor{T}{_A^A} \nonumber \\
    &+ (\Omega_{3} - \sqrt{3}\sigma_{3})T^{01}  -2\sqrt{3}\sigma_{\times}T^{03} + \sqrt{3}n_{\times}T^{12} + 2\sqrt{3}n_{+}T^{13} \nonumber \\
    &+ \sqrt{3}\hat{n}(T^{11} - T^{33}), \\
\label{eqn:T_03_U1_Frame_Evo}
    e_{0}(T^{03}) + e_{\mathcal{a}}(T^{\mathcal{a}3}) &= -(4H + \sigma_{+} -\sqrt{3}\sigma_{-})T^{03} + (3\sqrt{3}\hat{a}_{\mathcal{a}} - \dot{u}_{\mathcal{a}})T^{\mathcal{a}3} - \sqrt{3}n_{\times}T^{13} + \sqrt{3}\hat{n}T^{23}.
\end{align}

\begin{rem}[Polarised Matter]
    For non-vacuum polarised $U(1)$-symmetric spacetimes, we require
    \begin{align*}
        T^{03} = T^{13} = T^{23} &= 0. 
    \end{align*}
    An inspection of \eqref{eqn:T_00_U1_Frame_Evo}-\eqref{eqn:T_03_U1_Frame_Evo} demonstrates this a consistent reduction which is preserved by the matter evolution equations.
\end{rem}

\subsection{Gauge Conditions}
The system \eqref{eqn:sigma_plus_evo}-\eqref{eqn:e22_U1_evo} is symmetric hyperbolic when the gauge variables are prescribed functions of the spacetime coordinates. However, since this specific system has not been used in the literature before, it is useful to find a wider class of gauge conditions that result in a well-posed system. In this section, we will discuss some possible gauge choices.  \newline \par

\subsubsection{Spatial Frame Gauge}
As mentioned previously, the function $\Omega_{A}$ represents the angular velocity of the spatial frame $e_{A}$ relative to a frame which is Fermi propagated along the integral curves of $e_{0}$. In particular $\Omega_{A} = 0$, when we choose a Fermi propagated frame. By using a group-invariant frame, we are left with only one free component, $\Omega_{3}$. The most natural gauge choice appears to be
\begin{align*}
    \Omega_{3} = 0
\end{align*}
which ensures \eqref{eqn:sigma_plus_evo}-\eqref{eqn:H_U1_evo} remains symmetric hyperbolic. Another choice is to set
\begin{align*}
    \Omega_{3} = \pm\sqrt{3}\sigma_{3}.
\end{align*}
Indeed this is the gauge condition for the spatial frame in the Berger-Moncrief metric parameterisation given in Appendix \ref{app:BergerMoncrief_FrameExample}. 
While this simplifies the evolution equations it does not result in a obviously well-posed system.

\subsubsection{Generalised Harmonic Slicing with Zero Shift}
The generalised harmonic slicing condition is given by
\begin{align}
\label{eqn:Harmonic_Slicing}
    \Box t = f
\end{align}
where, for now, $f$ is an arbitrary function. The condition \eqref{eqn:Harmonic_Slicing} is equivalent to
\begin{align*}
    g^{ab}e_{a}\big(e_{b}(t)\big) - g^{ab}\tensor{\omega}{_a^k_b}e_{k}(t) = f
\end{align*}
Using the fact $e_{0} = \alpha^{-1}\del_{t}$ when $\beta^{\Omega}=0$ and the decomposition of the connection coefficients in \cite{RöhrUggla:2005}*{Page 6}, we obtain 
\begin{align}
\label{eqn:Harmonic_alpha_evo}
    e_{0}(\alpha) = 3H\alpha + \alpha^{2}f.
\end{align}
Next, by applying the mixed commutator \eqref{eqn:mixed_commutator_decomposition} to $\alpha$, using the gauge constraint \eqref{eqn:C3_Constraint_3+1}, and the equation for $\alpha$ \eqref{eqn:Harmonic_alpha_evo}, we obtain an evolution equation for $\dot{u}_{A}$
\begin{align}
\label{eqn:udotA_harmonic_evo}
    e_{0}(\dot{u}_{A}) = 3e_{A}(H) + \alpha e_{A}(f) + 2H\dot{u}_{A} + 2\alpha f\dot{u}_{A} - \big[\tensor{\sigma}{_A^B} +\tensor{\eps}{_A^B_C}\Omega^{C}]\dot{u}_{B}
\end{align}
Using \eqref{eqn:U1_Frame_Simplifications} and the definition of the symmetry adapted variables \eqref{eqn:decomposition_variables}, the evolution equation \eqref{eqn:udotA_harmonic_evo} reduces to
\begin{align}
\label{eqn:udot1_harmonic_evo}
    e_{0}(\dot{u}_{1}) &= 3e_{1}(H) + \alpha e_{1}(f) + 2\dot{u}_{1}(H + \alpha f + \sigma_{+}) - \dot{u}_{2}(\sqrt{3}\sigma_{3} + \Omega_{3}), \\
\label{eqn:udot2_harmonic_evo}
    e_{0}(\dot{u}_{2}) &= 3e_{2}(H) + \alpha e_{2}(f) +\dot{u}_{2}(2H +2\alpha f -\sigma_{+} - \sqrt{3}\sigma_{-}) -\dot{u}_{1}(\sqrt{3}\sigma_{3} - \Omega_{3})
\end{align}
Clearly, incorporating \eqref{eqn:udot1_harmonic_evo}-\eqref{eqn:udot2_harmonic_evo} into the evolution system \eqref{eqn:sigma_plus_evo}-\eqref{eqn:H_U1_evo} will spoil the symmetric hyperbolicity of the system. However, by adding appropriate multiples of components of the momentum constraint \eqref{eqn:CM1_U1}-\eqref{eqn:CM2_U1} we can recover a symmetric hyperbolic system. In particular, the combinations
\begin{align*}
    \frac{2}{9}e_{0}(\dot{u}_{1}) + \frac{1}{6}(\mathcal{C}_{M})_{1}, \quad \frac{2}{9}e_{0}(\dot{u}_{2}) + \frac{1}{6}(\mathcal{C}_{M})_{2},
\end{align*}
yield the following evolution equations
\begin{align}
\label{eqn:udot1_harmonic_evo_constraintmod}
\frac{2}{9}e_{0}(\dot{u}_{1}) &= \frac{1}{3}e_{1}(H) - \frac{1}{3}e_{1}(\sigma_{+}) + \frac{2}{9}\alpha e_{1}(f) + \frac{1}{2\sqrt{3}}e_{2}(\sigma_{3}) + \frac{4}{9}\dot{u}_{1}(H + \alpha f +\sigma_{+}) \nonumber \\ 
&- \frac{2}{9}\dot{u}_{2}(\sqrt{3} \sigma_{3} + \Omega_{3}) 
+ \sqrt{3} \sigma_{+}\hat{a}_{1} - \frac{3}{2}\sigma_{3}\hat{a}_{2} + n_{\times}\sigma_{-} + \frac{1}{2}\hat{n}\sigma_{3} + n_{+}\sigma_{\times} - \frac{1}{6}T_{01}, \\
\label{eqn:udot2_harmonic_evo_constraintmod}
\frac{2}{9}e_{0}(\dot{u}_{2}) &= \frac{1}{2\sqrt{3}}e_{1}(\sigma_{3}) + \frac{1}{3}e_{2}(H) + \frac{1}{6}e_{2}(\sigma_{+}) + \frac{1}{2\sqrt{3}}e_{2}(\sigma_{-}) + \frac{2}{9}\alpha e_{2}(f) - \frac{2}{9}\dot{u}_{1}(\sqrt{3}\sigma_{3} - \Omega_{3}) \nonumber \\
&+ \frac{2}{9}\dot{u}_{2}(2H + 2\alpha f -\sigma_{+} - \sqrt{3}\sigma_{-})- \frac{3}{2}\sigma_{3}\hat{a}_{1} - \frac{\sqrt{3}}{2}\hat{a}_{2}(\sigma_{+} + \sqrt{3}\sigma_{-}) + \frac{1}{2\sqrt{3}}\hat{n}(3\sigma_{+} - \sqrt{3}\sigma_{-}) \nonumber \\
&- \frac{1}{2}n_{\times}\sigma_{3} - n_{+}\sigma_{2} - \frac{1}{6}T_{02}.
\end{align}
It is straightforward to verify that the system consisting of \eqref{eqn:sigma_plus_evo}-\eqref{eqn:e22_U1_evo}, \eqref{eqn:Harmonic_alpha_evo}, and \eqref{eqn:udot1_harmonic_evo_constraintmod}-\eqref{eqn:udot2_harmonic_evo_constraintmod} is symmetric hyperbolic provided $f$ and $\Omega_{3}$ are prescribed functions of the spacetime coordinates. 

\subsubsection{CMC Slicing}
The constant mean curvature slicing condition is given by
\begin{align*}
    3H = \tau(t)
\end{align*}
where $\tau$ is a monotonic function of $t$. Substituting this into the evolution equation for $H$ \eqref{eqn:H_evo} yields
\begin{align*}
    e_{\mathcal{a}}(\dot{u}^{\mathcal{a}}) -\frac{\tau^{2}}{3} + \dot{u}_{\mathcal{a}}(\dot{u}^{\mathcal{a}} -2\sqrt{3}\hat{a}^{\mathcal{a}}) -6(\sigma_{+}^{2} + \sigma_{-}^{2} + \sigma_{2}^{2} + \sigma_{3}^{2} + \sigma_{\times}^{2}) - \frac{1}{2}(T_{00} + \tensor{T}{_A^A}) - e_{0}(\tau) = 0.
\end{align*}
By replacing $\dot{u}_{\mathcal{a}}$ in the above using the gauge constraint \eqref{eqn:C3_Constraint_3+1}, we then obtain an elliptic equation for the lapse
\begin{align}
\label{eqn:CMC_Lapse_eqn}
    e_{\mathcal{a}}\big(e^{\mathcal{a}}(\alpha)\big) -2\sqrt{3}\hat{a}^{\mathcal{a}}e_{\mathcal{a}}(\alpha) -\alpha\Big[ 6(\sigma_{+}^{2} + \sigma_{-}^{2} + \sigma_{2}^{2} + \sigma_{3}^{2} + \sigma_{\times}^{2}) + \frac{1}{2}(T_{00} + \tensor{T}{_A^A}) + \frac{\tau^{2}}{3}] - \del_{t}(\tau) = 0.
\end{align}
In this case, the field equations become a mixed elliptic-hyperbolic system consisting of the symmetric hyperbolic system\footnote{Provided that $\Omega_{3}$ is appropriately chosen.} \eqref{eqn:sigma_plus_evo}-\eqref{eqn:e22_U1_evo} and the elliptic equation for the lapse \eqref{eqn:CMC_Lapse_eqn}. 

\section{Relativistic Fluids}
\label{sec:U(1)_Euler_Equations}
Now, let us specialise to the case of a perfect fluid matter source,
\begin{align*}
    T^{\mu\nu} = (\rho + p)u^{\mu}u^{\nu} + pg^{\mu\nu},
\end{align*}
where $\rho$ is the energy density of the fluid, $p$ is the pressure, and $u^{\mu}$ is the fluid four-velocity. As discussed in the introduction, we assume the density and pressure are related by a linear equation of state
\begin{align*}
    p = K\rho, \quad K \in [0,1]
\end{align*}
where $K = c_{s}^{2}$ is the square of the fluid sound speed. In terms of our orthonormal frame, the 3+1 decomposition of the fluid stress-energy tensor is given by
\begin{align}
\label{eqn:T00_Fluid_3+1}
    T^{00} &= \big(\Gamma^{2}(K+1)-K\big)\rho, \\
\label{eqn:T0A_Fluid_3+1}
    T^{0A} &= \Gamma^{2}(K+1)\rho\nu^{A} = \frac{(1+K)T^{00}}{1+K|\nu|^{2}}\nu^{A}, \\
\label{eqn:TAB_Fluid_3+1}
    T^{AB} &= \Gamma^{2}(K+1)\rho\nu^{A}\nu^{B} + K\rho \eta^{AB} = \frac{(1+K)T^{00}}{1+K|\nu|^{2}}\nu^{A}\nu^{B} + \frac{K(1-|\nu|^{2})T^{00}}{1 + K|\nu|^{2}}\eta^{AB}.
\end{align}
The spatial fluid velocity $\nu^{A}$ and Lorentz factor $\Gamma$ are defined by
\begin{align*}
    u^{a} = \Gamma(n^{a} + \nu^{a}), \quad |\nu|^{2} = \delta_{AB}\nu^{A}\nu^{B}, \quad \Gamma^{2} = \frac{1}{1-|\nu|^{2}}, \quad |\nu| < 1,  
\end{align*}
where $n = e_{0}$ is the normal vector of the foliation and $n^{a}\nu_{a} = \nu_{0} = 0$. Previous numerical experience \cite{Marshall:2026} has shown that evolving the primitive variables $(\rho,\nu^{A})$ directly instead of using the `balance law' form of the matter equations \eqref{eqn:T_00_U1_Frame_Evo}-\eqref{eqn:T_03_U1_Frame_Evo} can result in better numerical stability and accuracy. Thus, we will also derive evolution equations for the primitive variables. First, using \eqref{eqn:T0A_Fluid_3+1}, we observe that
\begin{align*}
    e_{0}(T^{0A}\tensor{T}{^0_A}) = \frac{2(1+K)^{2}|\nu|^{2}T^{00}}{(1+K|\nu|^{2})^{2}}e_{0}(T^{00}) + \frac{(1+K)^{2}(1-K|\nu|^{2})(T^{00})^{2}}{(1+K|\nu|^{2})^{3}}e_{0}(|\nu|^{2})
\end{align*}
Re-arranging the above leads to an evolution equation for $|\nu|^{2}$\footnote{Note there is a typo in equation (2.114) of \cite{Marshall:2026}; a factor of $(1+K|\nu|^{2})$ is missing from the numerator on the right hand side. However, this is factor is present in equation (2.116) of that paper and all future equations, so the final system of equations for the primitive variables is still correct.},
\begin{align}
\label{eqn:nu_norm_evo}
    e_{0}(|\nu|^{2}) = \frac{(1+K|\nu|^{2})^{2}}{(1+K)T^{00}(1-K|\nu|^{2})}\Big(2\nu_{A}e_{0}(T^{0A}) - \frac{2|\nu|^{2}(1+K)}{(1+K|\nu|^{2})}e_{0}(T^{00}) \Big).
\end{align}
Next, using \eqref{eqn:T0A_Fluid_3+1}, we find 
\begin{align*}
    e_{0}(\nu^{A}) = \frac{1+K|\nu|^{2}}{(1+K)T^{00}}e_{0}(T^{0A}) - \frac{\nu^{A}}{T^{00}}e_{0}(T^{00}) + \frac{K\nu^{A}}{1+K|\nu|^{2}}e_{0}(|\nu|^{2}).
\end{align*}
Substituting \eqref{eqn:nu_norm_evo} into the above then yields
\begin{align}
\label{eqn:nu_A_evo}
    e_{0}(\nu^{A}) = \frac{1+K|\nu|^{2}}{(1+K)(1-K|\nu|^{2})T^{00}}\Bigg[\Big((1-K|\nu|^{2})\delta^{A}_{C} + 2K\nu^{A}\nu_{C}\Big)e_{0}(T^{0C}) - (K+1)\nu^{A}e_{0}(T^{00})\Bigg].
\end{align}
Now, by systematically replacing the stress-energy terms in \eqref{eqn:T_00_U1_Frame_Evo} and \eqref{eqn:nu_A_evo} and defining the \textit{modified energy density},
\begin{align}
    \label{eqn:zeta_variable_defn}
    \zeta := \log(T^{00}),
\end{align}
we obtain the following evolution equations,
\begin{align}
\label{eqn:Log_T00_evo}
&e_{0}(\zeta) = -\frac{(1+K)\nu_{1}}{1+K|\nu|^{2}}e_{1}(\zeta) - \frac{(1+K)(1+K|\nu|^{2} - 2K\nu_{1}^{2})}{(1+K|\nu|^{2})^{2}}e_{1}(\nu_{1}) \nonumber \\
&+ \frac{2K(1+K)\nu_{1}\nu_{2}}{(1+K|\nu|^{2})^{2}}e_{1}(\nu_{2}) + \frac{2K(1+K)\nu_{1}\nu_{3}}{(1+K|\nu|^{2})^{2}}e_{1}(\nu_{3}) \nonumber \\
&- \frac{(1+K)\nu_{2}}{1+K|\nu|^{2}}e_{2}(\zeta) + \frac{2K(1+K)\nu_{1}\nu_{2}}{(1+K|\nu|^{2})^{2}}e_{2}(\nu_{1}) \nonumber \\
&+ \frac{(1+K)\big(-1+K|\nu|^{2}-2K(\nu_{1}^{2}+\nu_{3}^{2})\big)}{(1+K|\nu|^{2})^{2}}e_{2}(\nu_{2}) + \frac{2K(1+K)\nu_{2}\nu_{3}}{(1+K|\nu|^{2})^{2}}e_{2}(\nu_{3}) \nonumber \\
&+ S_{0}, \\
\label{eqn:nu1_evo}
&e_{0}\big(\nu^{1}\big) = \frac{K(1-|\nu|^{2})(1-K|\nu|^{2} + (K-1)\nu_{1}^{2})}{(1+K)(-1+K|\nu|^{2})}e_{1}(\zeta) + \frac{\nu_{1}(1-3K + K(1+K)|\nu|^{2} - 2(K-1)K\nu_{1}^{2})}{-1+K^{2}|\nu|^{4}}e_{1}(\nu_{1}) \nonumber \\
&+ \frac{2K(-1+K|\nu|^{2} - (K-1)\nu_{1}^{2})\nu_{2}}{-1+K^{2}|\nu|^{4}}e_{1}(\nu_{2}) + \frac{2K(-1+K|\nu|^{2} - (K-1)\nu_{1}^{2})\nu_{3}}{-1+K^{2}|\nu|^{4}}e_{1}(\nu_{3}) \nonumber \\
&+ \frac{(K-1)K(1-|\nu|^{2})\nu_{1}\nu_{2}}{(1+K)(-1+K|\nu|^{2})}e_{2}(\zeta) - \frac{(-1+K^{2}|\nu|^{4} + 2(-1+K)K\nu_{1}^{2})\nu_{2}}{-1+K^{2}|\nu|^{4}}e_{2}(\nu_{1}) \nonumber \\
&+ \frac{K\nu_{1}\big((-1+|\nu|^{2})(1+K|\nu|^{2})-2(K-1)\nu_{2}^{2}\big)}{-1+K^{2}|\nu|^{4}}e_{2}(\nu_{2}) - \frac{2(K-1)K\nu_{1}\nu_{2}\nu_{3}}{-1+K^{2}|\nu|^{4}}e_{2}(\nu_{3}) \nonumber \\
&+ S_{1}, \\
\label{eqn:nu2_evo}
&e_{0}\big(\nu^{2}\big) = \frac{(K-1)K(1-|\nu|^{2})\nu_{1}\nu_{2}}{(1+K)(-1+K|\nu|^{2})}e_{1}(\zeta) + \frac{K\big((-1+|\nu|^{2})(1+K|\nu|^{2})-2(K-1)\nu_{1}^{2}\big)\nu_{2}}{-1+K^{2}|\nu|^{4}}e_{1}(\nu_{1}) \nonumber \\
&- \frac{\nu_{1}(-1+K^{2}|\nu|^{4}+2(K-1)K\nu_{2}^{2})}{-1+K^{2}|\nu|^{4}}e_{1}(\nu_{2}) - \frac{2(K-1)K\nu_{1}\nu_{2}\nu_{3}}{-1+K^{2}|\nu|^{4}}e_{1}(\nu_{3}) \nonumber \\
&+ \frac{K(-1+|\nu|^{2})(-1+K|\nu|^{2} - (K-1)\nu_{2}^{2})}{(1+K)(-1+K|\nu|^{2})}e_{2}(\zeta) + \frac{2K\nu_{1}(-1+K|\nu|^{2}-(K-1)\nu_{2}^{2})}{-1+K^{2}|\nu|^{4}}e_{2}(\nu_{1}) \nonumber \\
&+ \frac{\nu_{2}(1-3K+K(1+K)|\nu|^{2} - 2(K-1)K\nu_{2}^{2})}{-1+K^{2}|\nu|^{4}}e_{2}(\nu_{2}) + \frac{2K(-1+K|\nu|^{2} - (K-1)\nu_{2}^{2})\nu_{3}}{-1+K^{2}|\nu|^{4}}e_{2}(\nu_{3}) \nonumber \\
&+ S_{2}, \\
\label{eqn:nu3_evo}
&e_{0}\big(\nu^{3}\big) = \frac{(K-1)K(1-|\nu|^{2})\nu_{1}\nu_{3}}{(1+K)(-1+K|\nu|^{2})}e_{1}(\zeta) + \frac{\big(-1-3(K-1)|\nu|^{2} + K|\nu|^{4} + 2(K-1)(\nu_{2}^{2} + \nu_{3}^{2})\big)K\nu_{3}}{-1+K^{2}|\nu|^{4}}e_{1}(\nu_{1}) \nonumber \\
&- \frac{2(K-1)K\nu_{1}\nu_{2}\nu_{3}}{-1+K^{2}|\nu|^{4}}e_{1}(\nu_{2}) - \frac{\nu_{1}(-1+K^{2}|\nu|^{4} + 2(K-1)K\nu_{3}^{2})}{-1+K^{2}|\nu|^{4}}e_{1}(\nu_{3}) \nonumber \\
&+\frac{(K-1)K(1-|\nu|^{2})\nu_{2}\nu_{3}}{(1+K)(-1+K|\nu|^{2})}e_{2}(\zeta) - \frac{2(K-1)K\nu_{1}\nu_{2}\nu_{3}}{-1+K^{2}|\nu|^{4}}e_{2}(\nu_{1}) \nonumber \\
&+ \frac{K\big((-1+|\nu|^{2})(1+K|\nu|^{2}) - 2(K-1)\nu_{2}^{2})\nu_{3}}{-1+K^{2}|\nu|^{4}}e_{2}(\nu_{2}) - \frac{\nu_{2}(-1+K^{2}|\nu|^{4} + 2(K-1)K\nu_{3}^{2})}{-1+K^{2}|\nu|^{4}}e_{2}(\nu_{3}) \nonumber \\
&+ S_{3}, 
\end{align}
where the source terms $S_{a}$ are given by
\begin{align*}
    S_{0} &:= \frac{2\sqrt{3}(1+K)\hat{a}_{1}\nu_{1}}{1+K|\nu|^{2}} + \frac{2\sqrt{3}(1+K)\hat{a}_{2}\nu_{2}}{1+K|\nu|^{2}} + \frac{(1+K)\big(H(3+|\nu|^{2}) + 2\dot{u}_{1}\nu_{1} + 2\dot{u}_{2}\nu_{2}\big)}{1+K|\nu|^{2}} \nonumber \\
    &- \frac{\sqrt{3}(1+K)(|\nu|^{2} - \nu_{1}^{2} - 2\nu_{3}^{2})\sigma_{-}}{(1+K|\nu|^{2})} - \frac{(1+K)(|\nu|^{2} - 3\nu_{1}^{2})\sigma_{+}}{1+K|\nu|^{2}} - \frac{2\sqrt{3}(1+K)\nu_{1}\nu_{3}\sigma_{2}}{1+K|\nu|^{2}} \nonumber \\
    &- \frac{2\sqrt{3}(1+K)\nu_{1}\nu_{2}\sigma_{3}}{1+K|\nu|^{2}} - \frac{2\sqrt{3}(1+K)\nu_{2}\nu_{3}\sigma_{\times}}{1+K|\nu|^{2}} , \\
    S_{1} &:= \frac{(-1+3K)H(|\nu|^{2}-1)\nu_{1}}{-1+K|\nu|^{2}} + \dot{u}_{1}(-1+\nu_{1}^{2}) + \sqrt{3}\hat{a}_{1}\Big(-|\nu|^{2} - \frac{(1-2K+K|\nu|^{2})\nu_{1}^{2}}{-1+K|\nu|^{2}}\Big) - \sqrt{3}\hat{n}\nu_{1}\nu_{2} \nonumber \\
    &- \frac{\sqrt{3}\hat{a}_{2}(1-2K+K|\nu|^{2})\nu_{1}\nu_{2}}{-1+K|\nu|^{2}} + \frac{\sqrt{3}(-1+K)\nu_{1}(|\nu|^{2} - \nu_{1}^{2}-2\nu_{3}^{2})\sigma_{-}}{-1+K|\nu|^{2}} - \frac{\nu_{1}\big(2+(1-3K)|\nu|^{2} + 3(K-1)\nu_{1}^{2}\big)\sigma_{+}}{-1+K|\nu|^{2}} \nonumber \\
    &+ \sqrt{3}\big(-|\nu|^{2} + \nu_{1}^{2} + 2\nu_{3}^{2}\big)n_{\times} - \frac{2\sqrt{3}\big(-1+K|\nu|^{2} - (K-1)\nu_{1}^{2}\big)\nu_{3}\sigma_{2}}{-1+K|\nu|^{2}} + \frac{\sqrt{3}\big(1-K|\nu|^{2} + 2(K-1)\nu_{1}^{2}\big)\nu_{2}\sigma_{3}}{-1+K|\nu|^{2}} \nonumber \\
    &+ \frac{2\sqrt{3}(K-1)\nu_{1}\nu_{2}\nu_{3}\sigma_{\times}}{-1+K|\nu|^{2}} + \nu_{2}(\dot{u}_{2}\nu_{1} - 2\sqrt{3}\nu_{3}n_{+} - \Omega_{3}), \\
    S_{2} &:= \frac{(3K-1)H(-1+|\nu|^{2})\nu_{2}}{-1+K|\nu|^{2}} - \frac{\sqrt{3}\hat{a}_{1}(1-2K+K|\nu|^{2})\nu_{1}\nu_{2}}{-1+K|\nu|^{2}} + \dot{u}_{2}(-1+ \nu_{2}^{2}) + \sqrt{3}\hat{n}(-|\nu|^{2} + 2\nu_{1}^{2} + \nu_{2}^{2}) \nonumber \\
    &+ \sqrt{3}\hat{a}_{2}\Big(-|\nu|^{2} - \frac{(1-2K+K|\nu|^{2})\nu_{2}^{2}}{-1+K|\nu|^{2}}\Big) + \frac{\sqrt{3}\nu_{2}\big(1+(1-2K)|\nu|^{2} + (K-1)\nu_{1}^{2} + 2(K-1)\nu_{2}^{2}\big)\sigma_{-}}{-1+K|\nu|^{2}} \nonumber \\
    &- \frac{\big(-1+|\nu|^{2}+3(K-1)\nu_{1}^{2}\big)\nu_{2}\sigma_{+}}{-1+K|\nu|^{2}} + \sqrt{3}\nu_{1}\nu_{2}n_{\times} + \frac{2\sqrt{3}(K-1)\nu_{1}\nu_{2}\nu_{3}\sigma_{2}}{-1+K|\nu|^{2}} + \frac{\sqrt{3}\nu_{1}(1-K|\nu|^{2} + 2(K-1)\nu_{2}^{2})\sigma_{3}}{-1+K|\nu|^{2}} \nonumber \\
    &- \frac{2\sqrt{3}\big(-1 + K|\nu|^{2} - (K-1)\nu_{2}^{2})\nu_{3}\sigma_{\times}}{-1+K|\nu|^{2}} + \nu_{1}(\dot{u}_{1}\nu_{2} + 2\sqrt{3}\nu_{3}n_{+} + \Omega_{3}) , \\
    S_{3} &:= \frac{-\sqrt{3}\hat{a}_{1}(1-2K+K|\nu|^{2})\nu_{1}\nu_{3}}{-1+K|\nu|^{2}} + \sqrt{3}\hat{n}\nu_{2}\nu_{3} - \frac{\sqrt{3}\hat{a}_{2}(1-2K+K|\nu|^{2})\nu_{2}\nu_{3}}{-1+K|\nu|^{2}} \nonumber \\
    &+ \frac{\Big((-1+3K)H(-1+|\nu|^{2}) + \dot{u}_{1}(-1+K|\nu|^{2})\nu_{1} + \dot{u}_{2}(-1+K|\nu|^{2})\nu_{2}\Big)\nu_{3}}{-1 + K|\nu|^{2}} \nonumber \\
    &- \frac{\sqrt{3}\nu_{3}(1 + (1-2K)|\nu|^{2} + (K-1)\nu_{1}^{2} + 2(K-1)\nu_{3}^{2})\sigma_{-}}{-1+K|\nu|^{2}} - \frac{\big(-1 + |\nu|^{2} + 3(K-1)\nu_{1}^{2}\big)\nu_{3}\sigma_{+}}{-1+K|\nu|^{2}} - \sqrt{3}\nu_{1}\nu_{3}n_{\times} \nonumber \\
    &+ \frac{2\sqrt{3}(K-1)\nu_{1}\nu_{3}^{2}\sigma_{2}}{-1+K|\nu|^{2}} + \frac{2\sqrt{3}(K-1)\nu_{1}\nu_{2}\nu_{3}\sigma_{3}}{-1+K|\nu|^{2}} + \frac{2\sqrt{3}(K-1)\nu_{2}\nu_{3}^{2}\sigma_{\times}}{-1 + K|\nu|^{2}}.
\end{align*}

\subsection{Characteristic Structure of the Euler Equations}
The Euler equations \eqref{eqn:Log_T00_evo}-\eqref{eqn:nu3_evo} can be schematically written as a quasilinear system of the form
\begin{align*}
    \del_{t}(U) + \alpha(B^{1}e_{1}^{1} + B^{2}e_{2}^{1})
    \del_{x^{1}}(U) + \alpha(B^{1}e_{1}^{1} + B^{2}e_{2}^{1})\del_{x^{2}}(U) = \alpha S,
\end{align*}
where $U := (\zeta, \nu^{A})^{\text{T}}$. The characteristic speeds along each coordinate direction are given by the eigenvalues of the matrices in front of the derivatives, where we have assumed zero shift ($\beta^{\Omega}=0$),
\begin{align*}
    \lambda_{0,x^{1}} &= \alpha(e_{1}^{1}\nu_{1} + e_{2}^{1}\nu_{2}) , \\
    \lambda_{\pm, x^{1}} &= \frac{\alpha}{-1+K|\nu|^{2}}\Bigg((-1+K)(e_{1}^{1}\nu_{1} + e_{2}^{1}\nu_{2}) \pm \bigg[K(|\nu|^{2}-1)\Big((e_{1}^{1})^{2}(-1+K|\nu|^{2} - (K-1)\nu_{1}^{2}) \nonumber \\
    &+ (e_{2}^{1})^{2}(-1+K|\nu|^{2} - (K-1)\nu_{2}^{2}) - 2e_{1}^{1}e_{2}^{1}(K-1)\nu_{1}\nu_{2}\Big)\bigg]^{\frac{1}{2}}\Bigg), \\
    \lambda_{0,x^{2}} &= \alpha(e_{1}^{2}\nu_{1} + e_{2}^{2}\nu_{2}), \\
    \lambda_{\pm, x^{2}} &= \frac{\alpha}{-1+K|\nu|^{2}}\Bigg((-1+K)(e_{1}^{2}\nu_{1} + e_{2}^{2}\nu_{2}) \pm \bigg[K(|\nu|^{2}-1)\Big((e_{1}^{2})^{2}(-1+K|\nu|^{2} - (K-1)\nu_{1}^{2}) \nonumber \\
    &+ (e_{2}^{2})^{2}(-1+K|\nu|^{2} - (K-1)\nu_{2}^{2}) - 2e_{1}^{2}e_{2}^{2}(K-1)\nu_{1}\nu_{2}\Big)\bigg]^{\frac{1}{2}}\Bigg).
\end{align*}

\section{Numerical Implementation}
\label{sec:Numerical_Implementation}
We numerically evolve the equations \eqref{eqn:sigma_plus_evo}-\eqref{eqn:e33_U1_evo} and \eqref{eqn:Log_T00_evo}-\eqref{eqn:nu3_evo} on spacelike hypersurfaces diffeomorphic to\footnote{Note that this implies $\ell = 0$ in the constraint \eqref{eqn:C4_U1}.} $\mathbb{T}^{3}$,
\begin{align*}
    \{t\} \times \Sigma_{t} \cong \mathbb{T}^{3}.
\end{align*}
and use the following simple gauge conditions
\begin{align*}
    \alpha = 1, \quad \beta^{\Omega} = 0, \quad \dot{u}_{\mathcal{a}} = 0, \quad \Omega_{3} = 0.
\end{align*}
Our computational domain is given by $(x^{1},x^{2}) \in [0,2\pi)\times [0,2\pi)$ with periodic boundary conditions which is discretised into a uniform cell centred grid with $N^{2}$ total cells\footnote{That is to say, $\Delta x^{1} = \Delta x^{2}$.}. 
Spatial derivatives are discretised using centred, fourth-order finite difference stencils and we integrate in time using a fourth-order Runge-Kutta method. The time step is taken to be 
\begin{align*}
    \Delta t = \min \left\{ C \Delta x, C\frac{\Delta x}{|\lambda_{\max}|} \right\}
\end{align*}
where $C$ is the CFL constant and $|\lambda_{\max}|$ is the largest value of the characteristic speeds in the domain. The resulting scheme is formally fourth-order accurate in time and space for smooth solutions. 

\begin{rem}
    The spatial fluid velocity must satisfy the condition $|\nu|<1$ to ensure the Euler equations are well-posed. However, numerical error can lead to this constraint being violated as the fluid approaches extreme tilt, which eventually results in the code crashing. To resolve this issue, we enforce the fluid constraint using a limiter at each Runge-Kutta timestep, see \cite{Marshall:2026}*{\S 4.2} for details.
\end{rem}

\subsection{Initial Data} 
\label{sec:InitialData}
Ultimately, we are interested in studying the onset of the Rendall instability from small perturbations of spatially homogeneous and orthogonal ($|\nu|=0$) initial data with positive cosmological constant $\Lambda >0$ and super-radiative equations of state $\frac{1}{3}<K<1$. To this end, we study initial data sets which are characterised by a single parameter $0 < \eps < 1$ such that $\eps = 0$ corresponds to spatially homogeneous initial data. That is to say, $\eps$ parameterises the size of the inhomogeneous perturbation with respect to appropriate Sobolev norms. \newline \par

We specify our initial data using a variant of the conformal method \cites{York:1972,Lichnerowicz:1944}. In particular, our approach closely follows previous work by Garfinkle and collaborators \cites{CurtisGarfinkle:2005,Garfinkle:2004,Garfinkle:2007,Ijjas_et_al:2021,GarfinklePretorius:2020,Garfinkle_et_al:2023,Ijjas_et_al:2024}. To begin, we assume our initial spatial slice is conformally flat
\begin{align*}
    h_{\Omega\Gamma}|_{t_{0}} = \psi^{4}\delta_{\Omega\Gamma}
\end{align*}
We then fix the spatial frame components, $n_{AB}$, $\hat{a}_{\mathcal{a}}$ as follows\footnote{By setting $n_{33}=0$, we trivially satisfy the topological constraint \eqref{eqn:C4_U1}.}
\begin{align}
\label{eqn:spatial_frame_initial}
    e_{A}^{\Omega} &= \psi^{-2}\delta_{A}^{\Omega}, \\
\label{eqn:n_AB_initial}
    n_{AB} &= 0, \\
\label{eqn:a_hat_initial}
    \hat{a}_{\mathcal{a}} &= \frac{-2}{\sqrt{3}}\psi^{-3}\del_{\mathcal{a}}(\psi),
\end{align}
where the last equation follows from inserting \eqref{eqn:spatial_frame_initial}-\eqref{eqn:n_AB_initial} into the constraints \eqref{eqn:C1_112_U1}-\eqref{eqn:C1_122_U1}. Additionally, we assume that
\begin{equation*}
\begin{gathered}
    \sigma_{+} = \psi^{-6}Z_{+}, \quad \sigma_{-} = \psi^{-6}Z_{-}, \quad \sigma_{2} = \psi^{-6}Z_{2}, \quad \sigma_{\times} = \sigma_{3} = 0, \\
    \nu^{A} = 0, \quad H = 1, \quad \rho = 1 + \eps\sin(x^{1})\sin(x^{2}), \quad \Lambda = 1.
\end{gathered}
\end{equation*}
Under these assumptions, the Momentum and Hamiltonian constraints \eqref{eqn:CM1_U1}-\eqref{eqn:CH_U1} reduce to ODEs for $Z_{+}$, $Z_{-}$, and $Z_{2}$ and an elliptic equation for $\psi$
\begin{gather}
\label{eqn:constraints_Z}
\del_{x^{1}}(Z_{+}) = 0, \quad \del_{x^{2}}(Z_{+} +\sqrt{3}Z_{-}) = 0, \quad \del_{x^{1}}(Z_{2}) = 0, \\
\label{eqn:elliptic_psi}
\del_{x^{1}}^{2}(\psi) + \del_{x^{2}}^{2}(\psi) - \frac{3}{4}\psi^{5}H^{2} + \frac{1}{4}\psi^{5}(\Lambda + T_{00}^{\text{fl}}) + \frac{3}{4}\psi^{-7}(Z_{+}^{2} + Z_{-}^{2} + Z_{2}^{2}) = 0,
\end{gather}
where we have decomposed the stress-energy tensor into the cosmological constant and fluid parts
\begin{align*}
    T_{ab} &= -\Lambda g_{ab} + T_{ab}^{\text{fl}}, \\
    T_{ab}^{\text{fl}} &= (\rho + p)u_{a}u_{b} + pg_{ab}.
\end{align*}
We choose the following simple solution for \eqref{eqn:constraints_Z}
\begin{align*}
    Z_{+} = 1, \quad Z_{-} = 1 + \eps\sin(x^{1}), \quad Z_{2} = 1 + \eps\sin(x^{2}).
\end{align*}
The elliptic equation \eqref{eqn:elliptic_psi} is then solved numerically using Newton-Successive over-relaxation and fourth order finite difference stencils, see Appendix \ref{app:Numerical_Constraint_Solver} for details. 

\begin{rem}
    Since the initial condition $\nu^{3}=0$ is preserved by the evolution equations \eqref{eqn:Log_T00_evo}-\eqref{eqn:nu3_evo}, our numerical solutions always have $\nu^{3}=0$. 
\end{rem}

\subsection{Convergence Tests}
We have run convergence tests using resolutions of $N^{2} \in \{50^{2},100^{2},200^{2},400^{2}\}$. The numerical error is estimated by taking the difference of the appropriately restricted $400^{2}$ evolution with the lower resolution runs. Importantly, since we are using a cell centred grid, the grid points at different resolutions do not align. Hence, we must use (at least) fourth order interpolation to restrict the solution at the high resolution grid points to the coarser grid points. For $N_{\text{fine}} = 2N_{\text{coarse}}$, the fourth order restriction operator is given by\footnote{This operator is obtained by sequentially using fourth order Lagrange interpolation along each coordinate direction.}
\begin{align*}
    f_{k,l}^{\text{coarse}} &= \frac{-1}{16}\Big(\frac{-1}{16}f_{i-1,j-1} + \frac{9}{16}f_{i,j-1} + \frac{9}{16}f_{i+1,j-1} - \frac{1}{16}f_{i+2,j-1} \Big) \nonumber \\
    &+ \frac{9}{16}\Big(\frac{-1}{16}f_{i-1,j} + \frac{9}{16}f_{i,j} + \frac{9}{16}f_{i+1,j} - \frac{1}{16}f_{i+2,j} \Big) \nonumber \\
    &+ \frac{9}{16}\Big(\frac{-1}{16}f_{i-1,j+1} + \frac{9}{16}f_{i,j+1} + \frac{9}{16}f_{i+1,j+1} - \frac{1}{16}f_{i+2,j+1} \Big) \nonumber \\
    &- \frac{1}{16}\Big(\frac{-1}{16}f_{i-1,j+2} + \frac{9}{16}f_{i,j+2} + \frac{9}{16}f_{i+1,j+2} - \frac{1}{16}f_{i+2,j+2} \Big),
\end{align*}
where $(x_{i}^{1},x_{j}^{2}) = (x_{k}^{1}-\frac{\Delta x^{\text{coarse}}}{4}, x_{l}^{2}-\frac{\Delta x^{\text{coarse}}}{4})$. Apart from the topological constraint \eqref{eqn:C4_U1}, which is preserved to machine precision cf. Figure \ref{fig:C4_Conservation}, all variables and constraints display close to the expected fourth order convergence with slight degradation at later times during the strongest cascade regime, see Tables \ref{table:zeta_convergence_error_t10}-\ref{table:zeta_convergence_errort_t40} and Figure \ref{fig:Convergence_HamiltonianConstraint}.

\begin{table}[htbp!]
\centering
\begin{tabular}{|l|l|l|l|l|l|l|}
\hline
\multicolumn{1}{|l|}{$N$} &  \multicolumn{1}{l|}{$\|u_{N^{2}} - u_{400^{2}}\|_{L^{2}}$} & \multicolumn{1}{l|}{Order} & \multicolumn{1}{l|}{$\|u_{N^{2}} - u_{400^{2}}\|_{L^{1}}$} & \multicolumn{1}{l|}{Order} & \multicolumn{1}{l|}{$\|u_{N^{2}} - u_{400^{2}}\|_{L^{\infty}}$} & \multicolumn{1}{l|}{Order} \\ \hline
50& $3.14 \times 10^{-4}$ & - &  $1.81\times10^{-4}$ & - &  $1.77\times10^{-4}$ & - \\ 
100 & $2.19 \times 10^{-5}$ & $3.84$  & $1.28 \times 10^{-4}$ & $3.82$ & $1.17 \times 10^{-5}$ & $3.92$\\  
200 & $1.23 \times 10^{-6}$ & $4.15$ & $7.20 \times 10^{-6}$& $4.15$ & $6.70 \times 10^{-7}$& $4.13$\\ 
\hline                                
\end{tabular}
\caption{Numerical error and convergence order for $\zeta$ at $t = 10.0$, $\eps = \frac{1}{2}$, $K = 0.4$.}
\label{table:zeta_convergence_error_t10}
\end{table}

\begin{table}[htbp!]
\centering
\begin{tabular}{|l|l|l|l|l|l|l|}
\hline
\multicolumn{1}{|l|}{$N$} &  \multicolumn{1}{l|}{$\|u_{N^{2}} - u_{400^{2}}\|_{L^{2}}$} & \multicolumn{1}{l|}{Order} & \multicolumn{1}{l|}{$\|u_{N^{2}} - u_{400^{2}}\|_{L^{1}}$} & \multicolumn{1}{l|}{Order} & \multicolumn{1}{l|}{$\|u_{N^{2}} - u_{400^{2}}\|_{L^{\infty}}$} & \multicolumn{1}{l|}{Order} \\ \hline
50& $3.74 \times 10^{-2}$ & - &  $2.34\times10^{-2}$ & - &  $1.95\times10^{-1}$ & - \\ 
100 & $2.32 \times 10^{-3}$ & $4.01$  & $1.61 \times 10^{-3}$ & $3.86$ & $8.78 \times 10^{-3}$ & $4.47$\\  
200 & $1.80 \times 10^{-4}$ & $3.69$ & $1.13 \times 10^{-4}$& $3.83$ & $8.55 \times 10^{-4}$& $3.36$\\ 
\hline                                
\end{tabular}
\caption{Numerical error and convergence order for $\zeta$ at $t = 39.9$, $\eps = \frac{1}{2}$, $K = 0.4$.}
\label{table:zeta_convergence_errort_t40}
\end{table}

\begin{figure}[htbp]
\centering
\subfigure[Subfigure 1 list of figures text][]{
\includegraphics[width=0.45\textwidth]{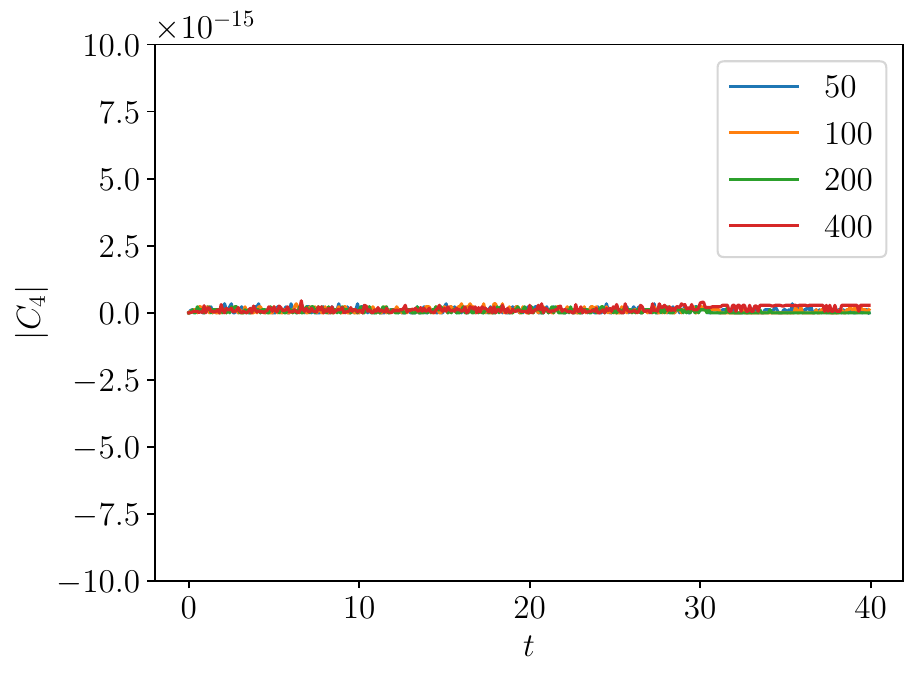}
\label{fig:C4_Conservation}}
\subfigure[Subfigure 2 list of figures text][]{
\includegraphics[width=0.45\textwidth]{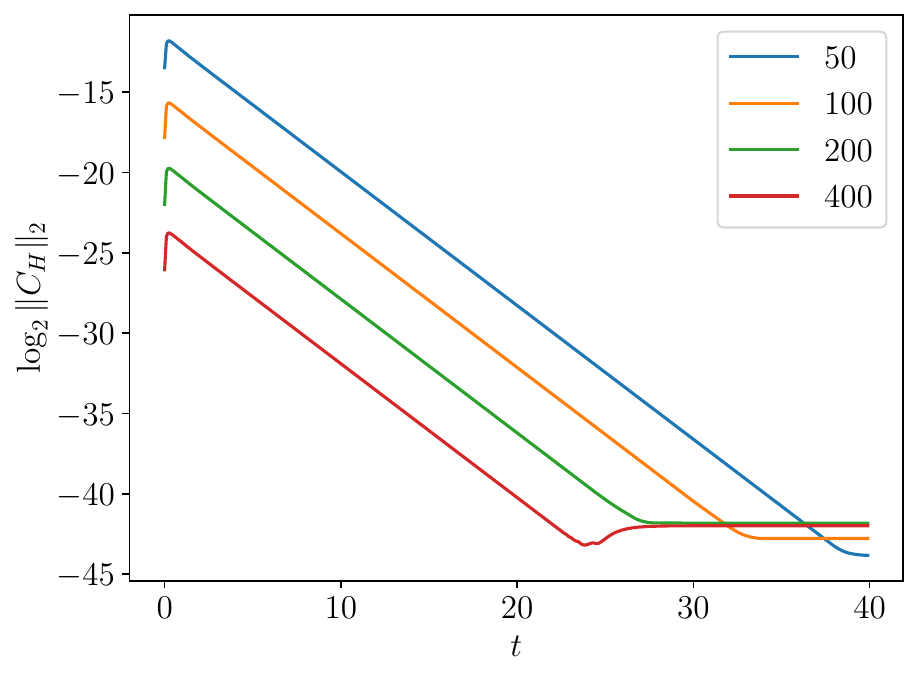}
\label{fig:Convergence_HamiltonianConstraint}}
\caption{(a) Conservation of the topological constraint \eqref{eqn:C4_U1}. (b) $L^{2}$ norm of the Hamiltonian constraint violation. The plateau at late times is due to the floating point precision limit of the finite difference stencils being reached. $\eps = \frac{1}{2}$, $K = 0.4$}
\end{figure}

\section{Numerical Results}
\label{sec:Numerical_Results}
The heuristic explanation for the Rendall instability in the introduction relied on the assumption that the Einstein-Euler equations are well approximated, pointwise, by spatially homogeneous solutions of the field equations towards future timelike infinity. In particular, we expect that our solutions should behave like spatially homogeneous fluids on a fixed de Sitter background. For our variables, this approximation corresponds to the conditions
\begin{align}
\label{eqn:de_Sitter_conditions}
   \{e_{A}^{\Omega}, \hat{a}_{\mathcal{a}}, n_{+}, n_{\times}, \hat{n}, \sigma_{+}, \sigma_{-}, \sigma_{\times}, \sigma_{2}, \sigma_{3}, T^{\text{fl}}_{ab} \} \rightarrow 0, \quad H  \rightarrow \sqrt{\frac{\Lambda}{3}}.
\end{align}
as $t \nearrow \infty$. In particular, $e_{A}^{\Omega} \rightarrow 0$ implies that solutions to the Einstein-Euler equations develop ODE-dominated asymptotics, provided that the derivatives of the gravitational and fluid variables remain suitably small. In all of our simulations, we observe that \eqref{eqn:de_Sitter_conditions} is rapidly satisfied, cf. Figure \ref{fig:decay_plots}, indicating that we reaching the regime where the Rendall instability can be expected to occur. Indeed, in Figure \ref{fig:RendallInstability_Plots}, we show the onset of the Rendall instability in the norm of the density gradient\footnote{Similar behaviour can also be seen for the norm of $\del_{x^{\mathcal{a}}}\zeta$.},
\begin{align*}
    \left|\frac{\del_{x^{\Omega}}\rho}{\rho}\right| = \sqrt{\frac{1}{\rho^{2}}\Big((\del_{x^{1}}\rho)^{2} + (\del_{x^{2}}\rho)^{2}\Big)}.
\end{align*}
In particular, Figure \ref{fig:RendallInstability_Plots} highlights that the spikes in the density gradient coincide with the transition between orthogonal ($|\nu|=0$) and (near) extremely tilted ($|\nu|=1$) states in the spatial fluid velocity, which is consistent with both the heuristic argument presented in the introduction and previous works in one spatial dimension \cites{MarshallOliynyk:2022,Oliynyk:2024,BMO:2023,BMO:2024,ColeyLim:2012,ColeyLim:2013,ColeyLim:2015}. 

\begin{figure}
    \centering
    \includegraphics[width=0.45\textwidth]{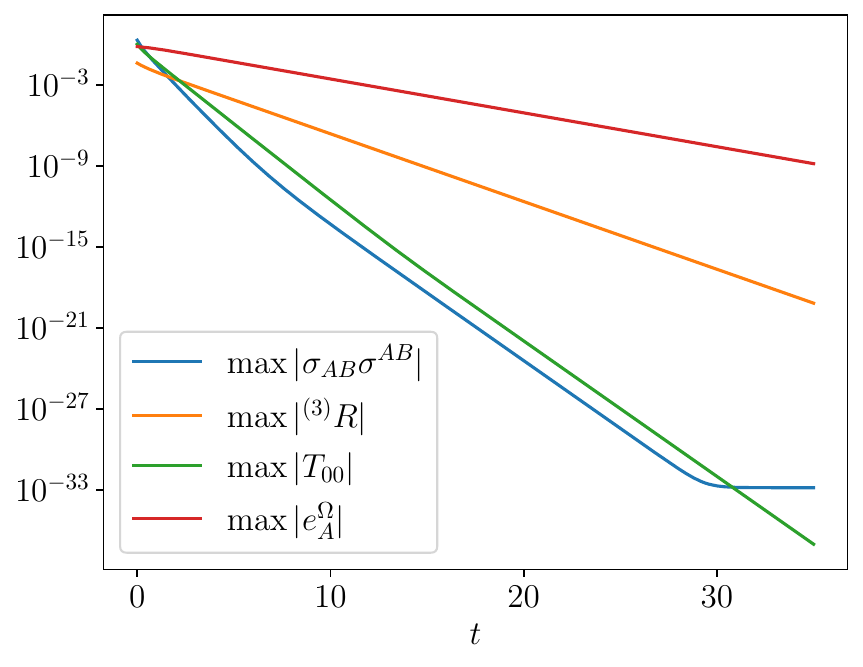}
    \caption{Decay of the spatial Ricci scalar ${}^{(3)}R$, $\sigma_{AB}\sigma^{AB}$, $T_{00}$, and the frame components $e_{A}^{\Omega}$. The maximum at each timestep is taken over all grid points. $N = 800$, $K = 0.5$, $\eps = 0.05$.}
    \label{fig:decay_plots}
\end{figure}

\begin{figure}[htbp]
    \centering
    \subfigure[$t=0$]{
        \begin{minipage}[b]{0.21\textwidth}
            \includegraphics[width=\textwidth]{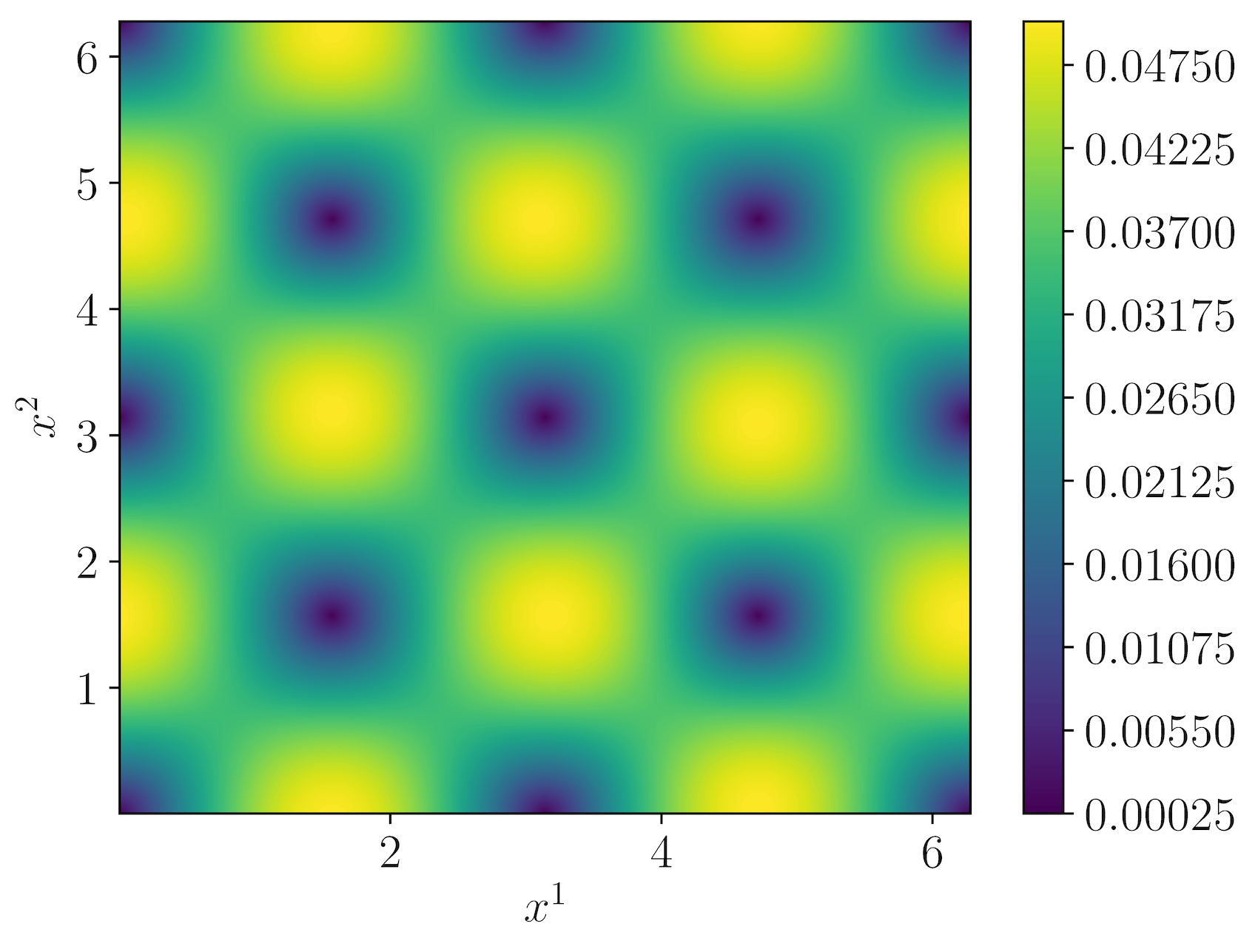}\\[0.5em]
            \includegraphics[width=\textwidth]{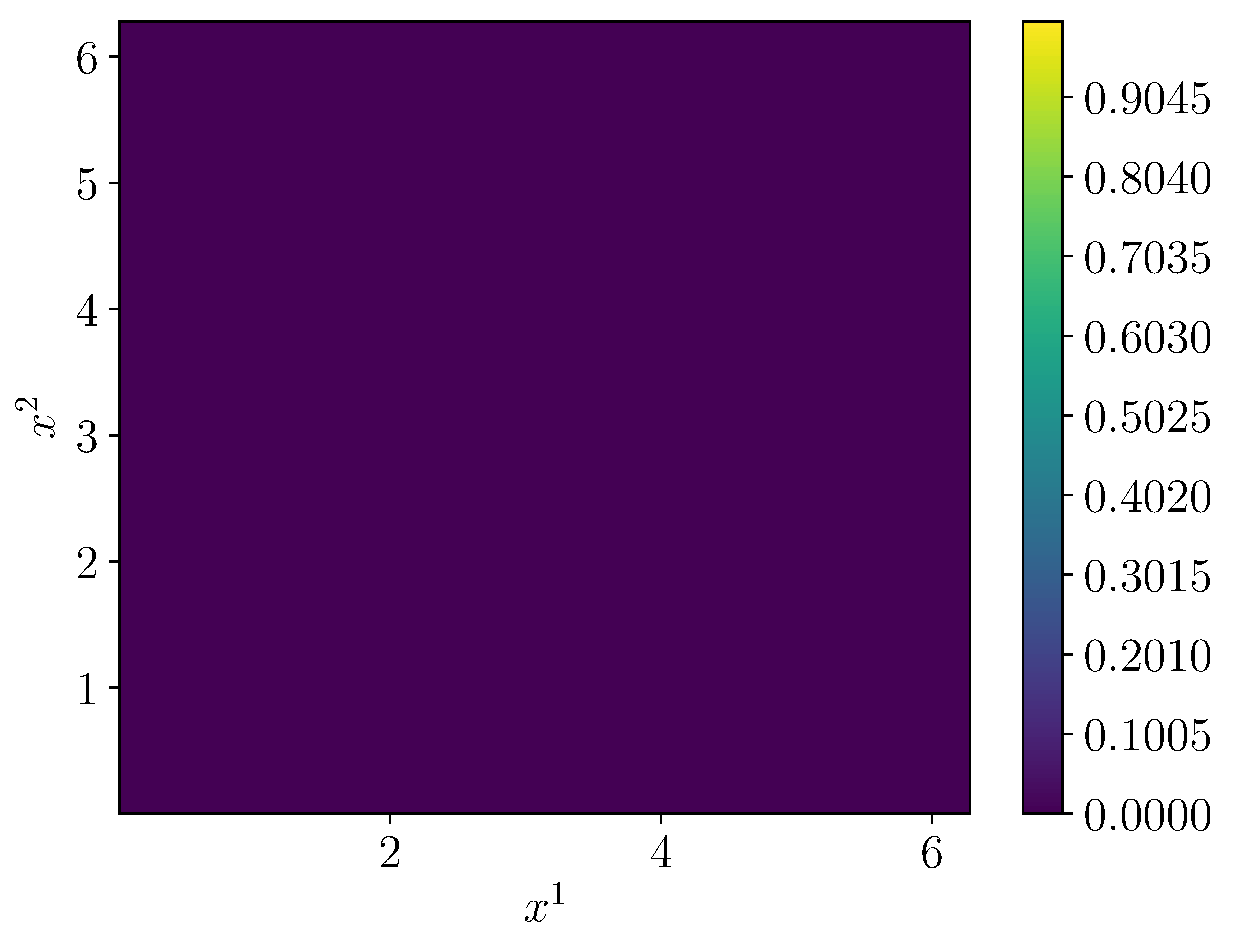}
        \end{minipage}
    }\hspace{0.01\textwidth}
    \subfigure[$t = 18$]{
        \begin{minipage}[b]{0.21\textwidth}
            \includegraphics[width=\textwidth]{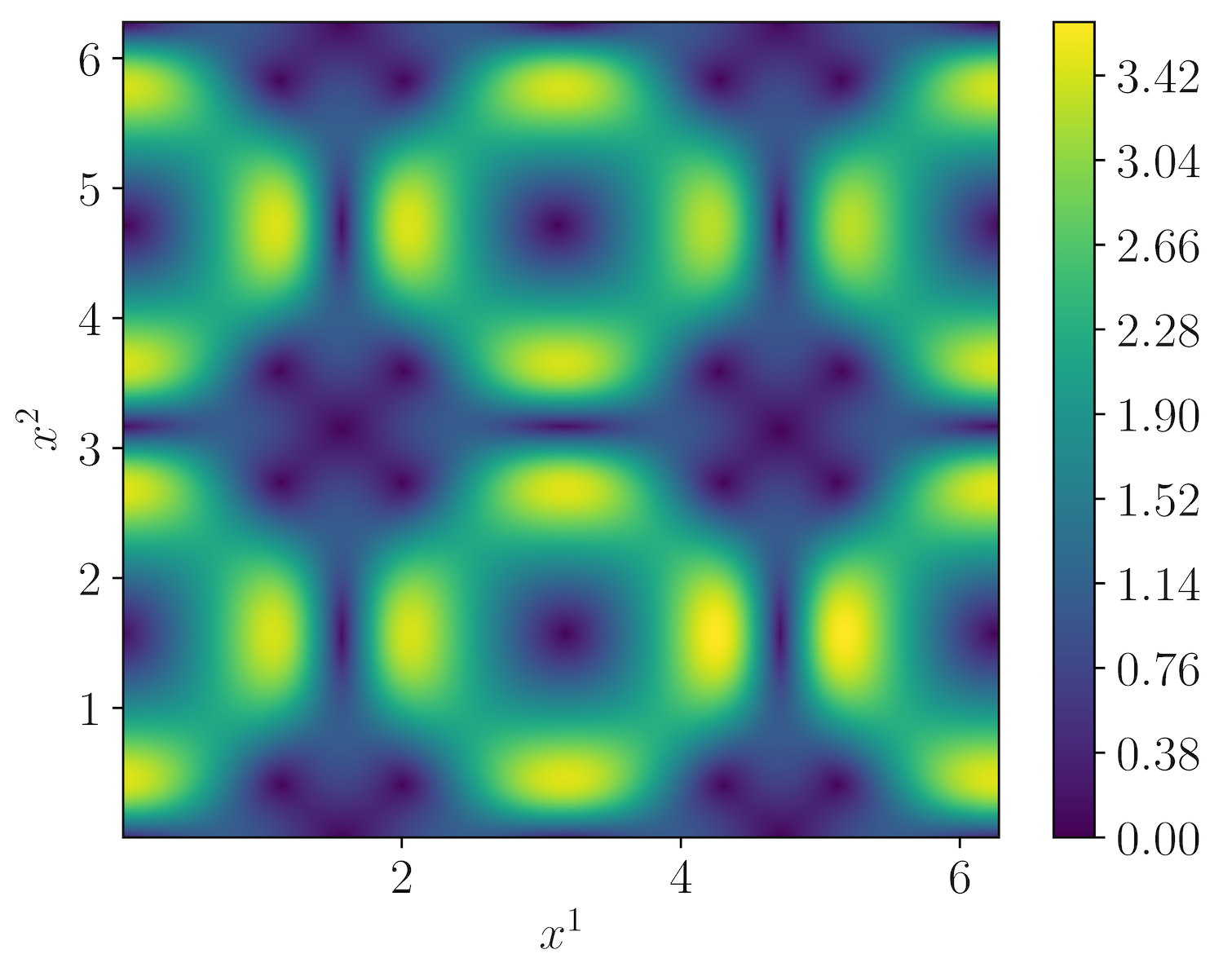}\\[0.5em]
            \includegraphics[width=\textwidth]{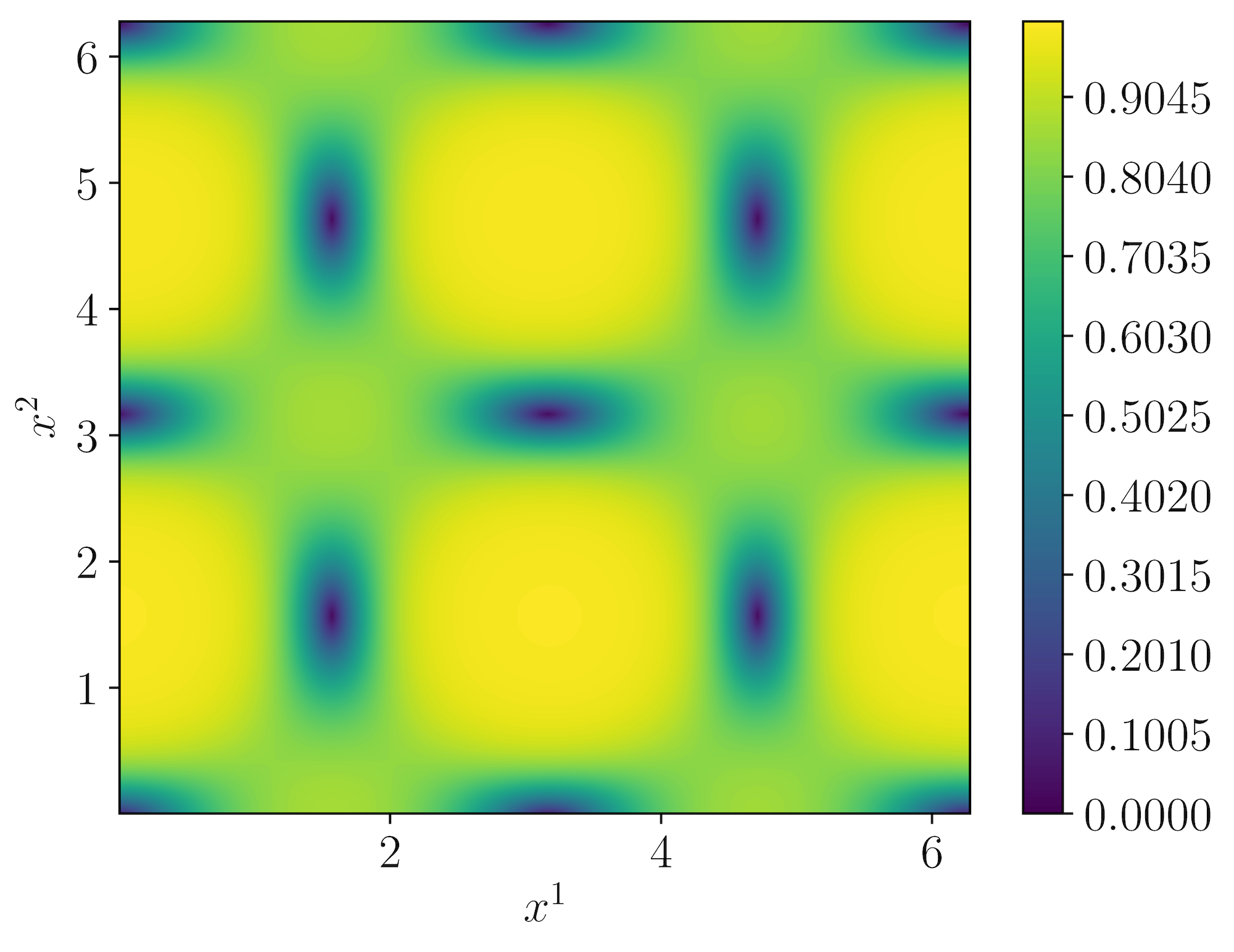}
        \end{minipage}
    }\hspace{0.01\textwidth}
    \subfigure[$t = 25$]{
        \begin{minipage}[b]{0.21\textwidth}
            \includegraphics[width=\textwidth]{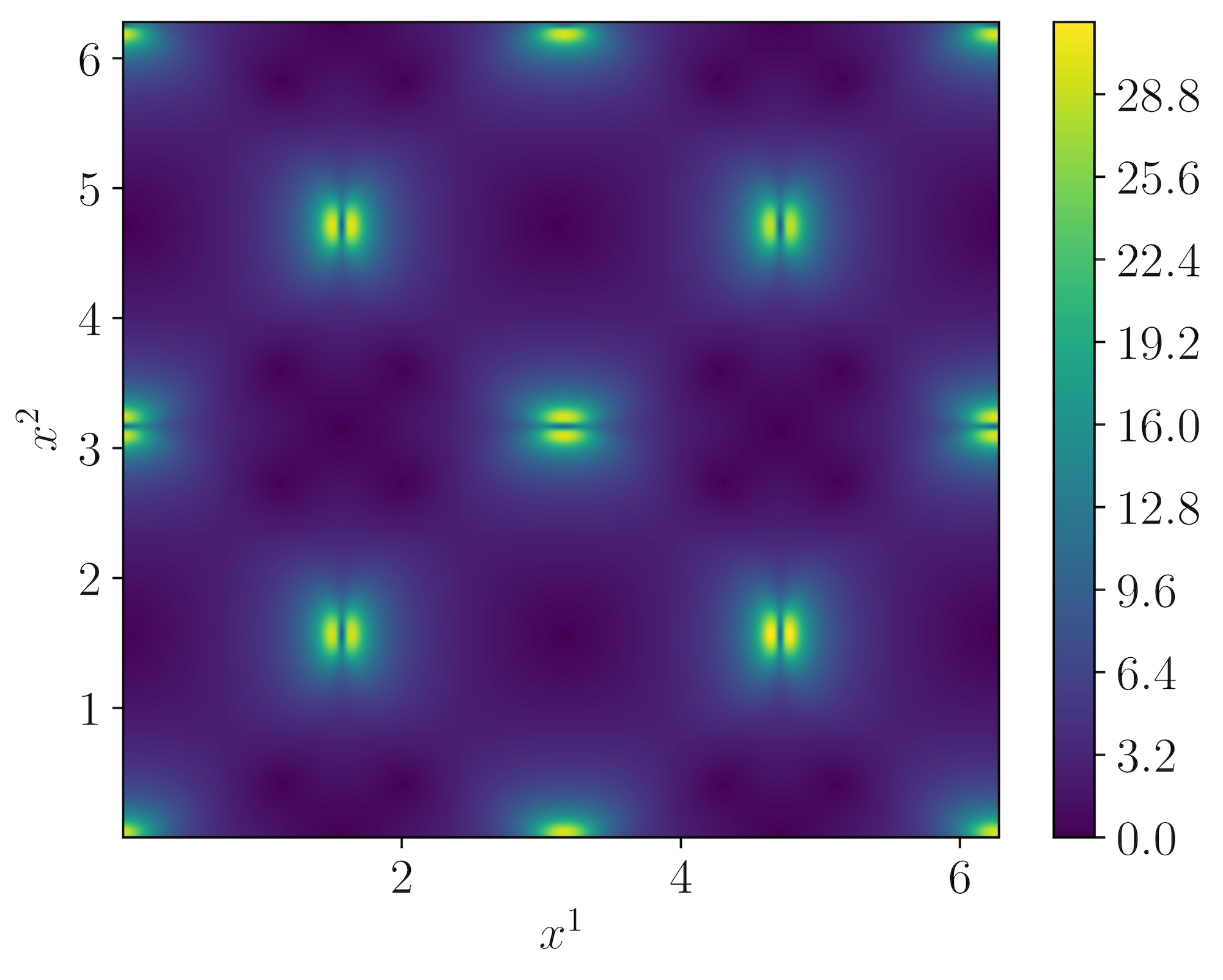}\\[0.5em]
            \includegraphics[width=\textwidth]{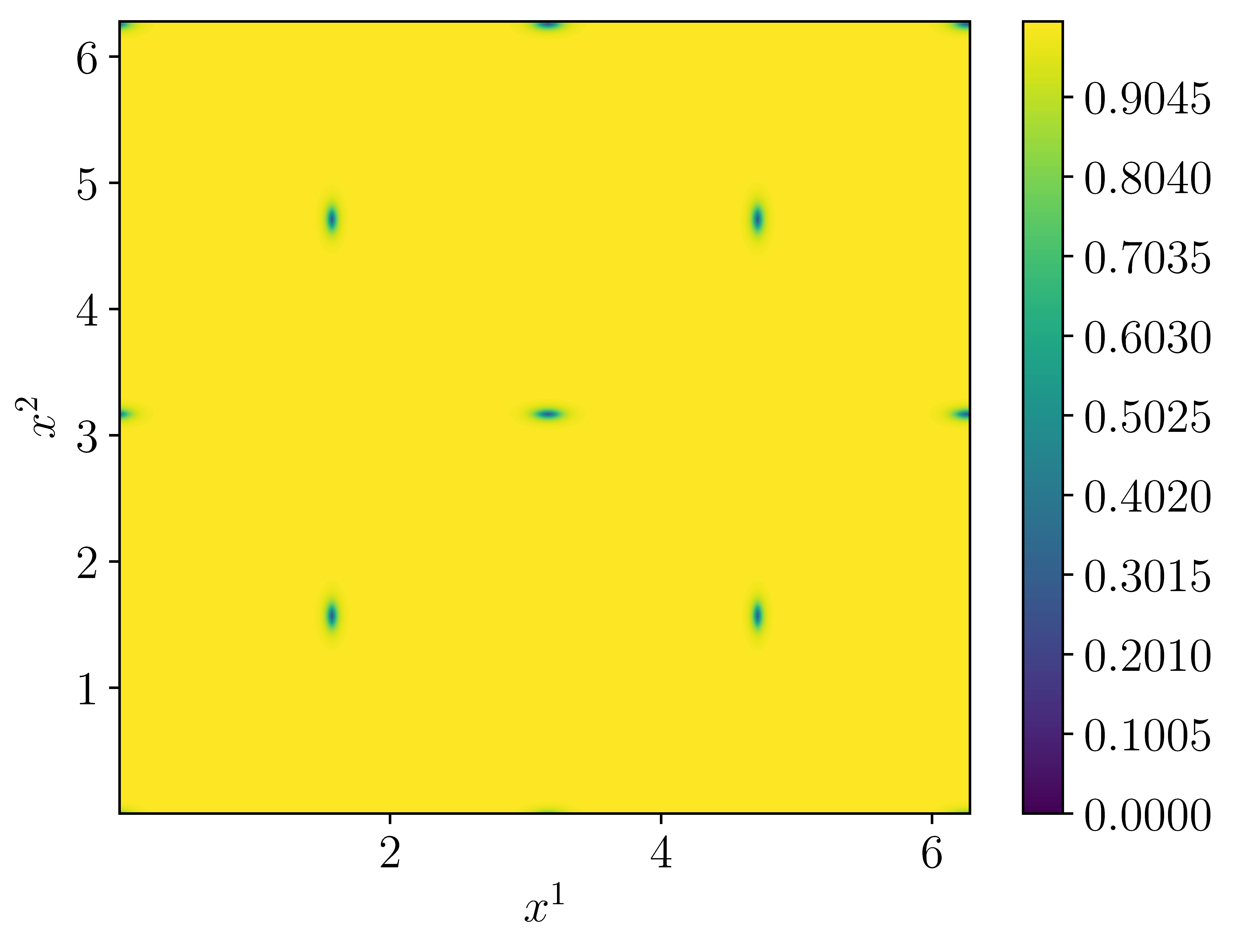}
        \end{minipage}
    }\hspace{0.01\textwidth}
    \subfigure[$t = 28.4$]{
        \begin{minipage}[b]{0.21\textwidth}
            \includegraphics[width=\textwidth]{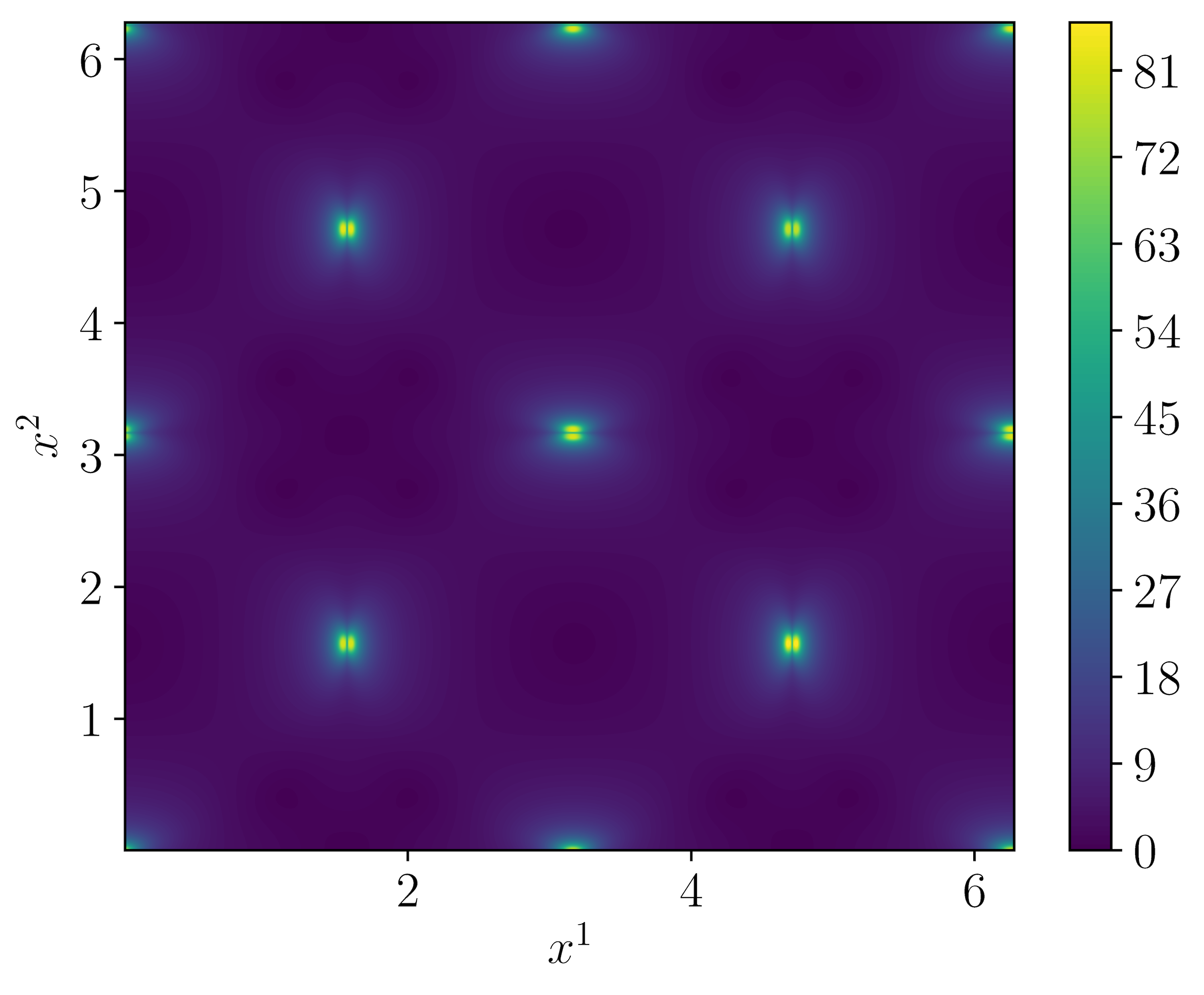}\\[0.5em]
            \includegraphics[width=\textwidth]{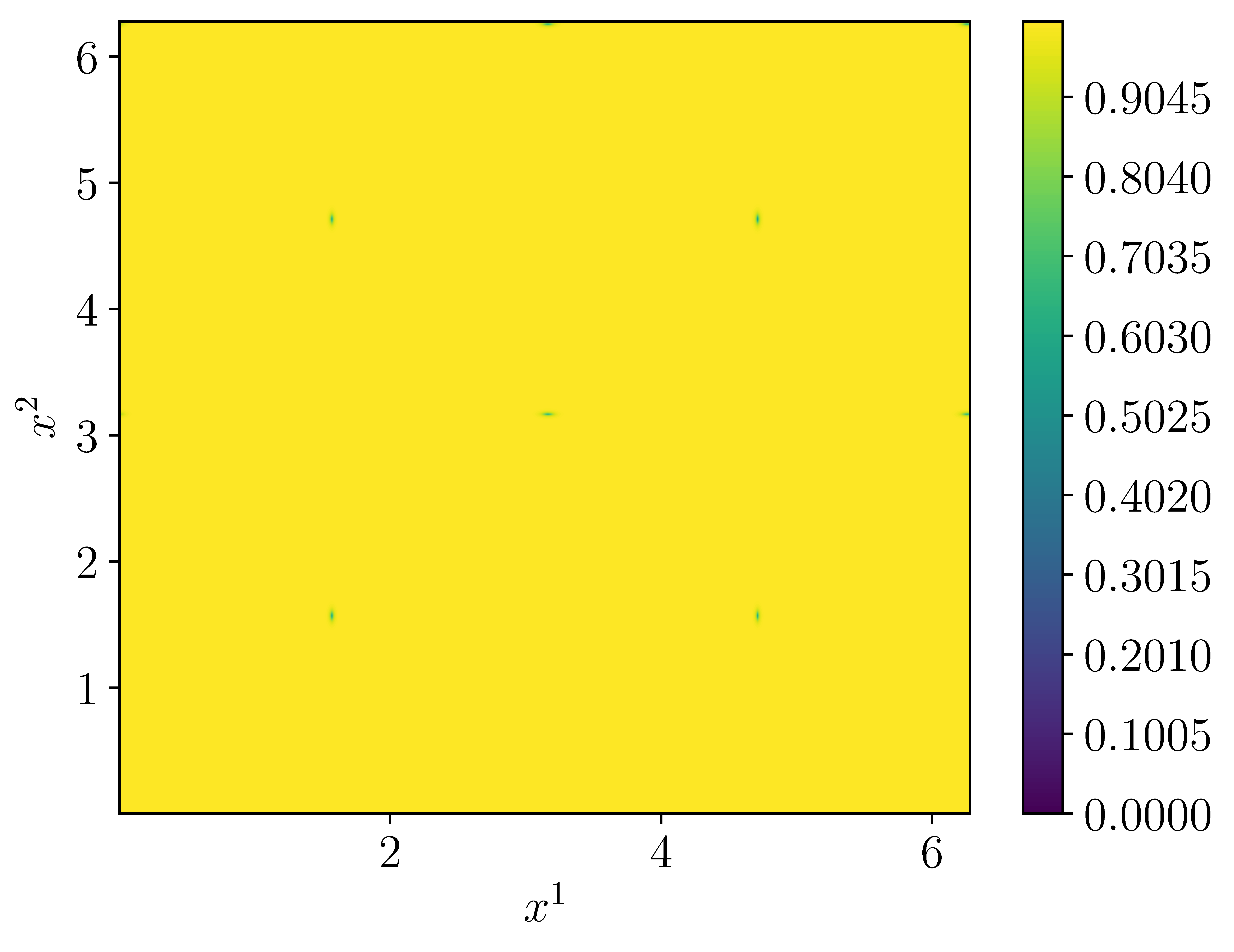}
        \end{minipage}
    }
    \caption{The norms of the density gradient  $\left|\frac{\del_{x^{\Omega}}\rho}{\rho}\right|$ (top row) and spatial fluid velocity $|\nu|$ (bottom row) at various times. The spikes in the density gradient at late times coincide with the transition between orthogonal and extremely tilted states in the fluid velocity, consistent with Rendall's original conjecture. $N = 800$, $K = 0.5$, $\eps = 0.05$.}
    \label{fig:RendallInstability_Plots}
\end{figure}

\subsection{Ultraviolet Cascade and Norm Inflation}
As future timelike infinity is approached, the Rendall instability drives the fluid density to concentrate on progressively smaller scales. In this section, we demonstrate that this concentration leads to a forward ultraviolet cascade in the energy density of the fluid. Moreover, we provide strong evidence that this rapid transfer of energy from low to high frequency modes leads to the inflation of Sobolev norms of the modified energy density $\zeta$, consistent with the previously described conjecture of Oliynyk \cite{Oliynyk:2024}. \newline \par

\subsubsection{Ultraviolet Cascade}

To analyse how energy is transferred from lower to higher modes, we compute the discrete Fourier transform of the modified energy density, $\zeta$. For a computational domain with $N^{2}$ grid points, this is given by
\begin{align*}
    \widehat{\zeta}_{m,n} = \sum_{l,j=1}^{N}e^{-2\pi i \frac{(jm + ln)}{N}} \zeta_{j,l}
\end{align*}
where $j,l \in [1,\cdots, N]$ are the grid indices corresponding to the $x^{1}$ and $x^{2}$ coordinates, respectively, $\zeta_{j,l}$ is the value of $\zeta$ at the grid point $(x_{j}^{1},x_{l}^{2})$, and $m,n \in [0,N-1]$ are integers. For even $N$, the wavenumber $k_{m,n}$ at index $(m,n)$ is given by\footnote{This matches the convention for frequencies in the \textsc{NumPy} \cite{numpy:2020} implementation of the discrete Fourier transform.} 
\begin{align*}
    k_{m,n} := (k_{1},k_{2}), \quad k_{1} = \begin{cases}
        m & \text{if} \quad 0 \leq m < \frac{N}{2} \\
        m-N & \text{if} \quad \frac{N}{2} \leq m < N 
    \end{cases}, \quad k_{2} = \begin{cases}
        n & \text{if} \quad 0 \leq n < \frac{N}{2} \\
        n-N & \text{if} \quad \frac{N}{2} \leq n < N
        \end{cases}.
\end{align*}
In the vein of Carrasco et al. \cite{Carrasco_et_al:2012}, we then define the `power spectrum' of $\zeta$ as a function of the norm of the wavenumber $|k_{m,n}|$ by
\begin{align*}
    P(\tilde{k}) := \sum_{(m,n)\in I}|\widehat{\zeta}_{m,n}|^{2}, \quad I = \Big\{(m,n) \in [0,N-1]^{2}|\; \tilde{k} \leq |k_{m,n}| < 
    \tilde{k}+1, \; \tilde{k} \in \mathbb{Z}_{\geq 0} \Big\},
\end{align*}
where $|k_{m,n}| = \sqrt{k_{1}^{2} + k_{2}^{2}}$. \newline \par

Initially, the power spectrum is dominated by the lower modes. However, as the Rendall instability develops, we observe a forward ultraviolet cascade where energy is rapidly transferred to higher modes and the spectrum flattens out. This is shown in Figure \ref{fig:spectrum_cascade}. It is important to emphasise that, for any fixed resolution, we can only follow the cascade up to the Nyquist frequency, at which point the spectrum becomes saturated and the cascade artificially stops. \newline \par 

\begin{figure}
    \centering
    \includegraphics[width=0.5\textwidth]{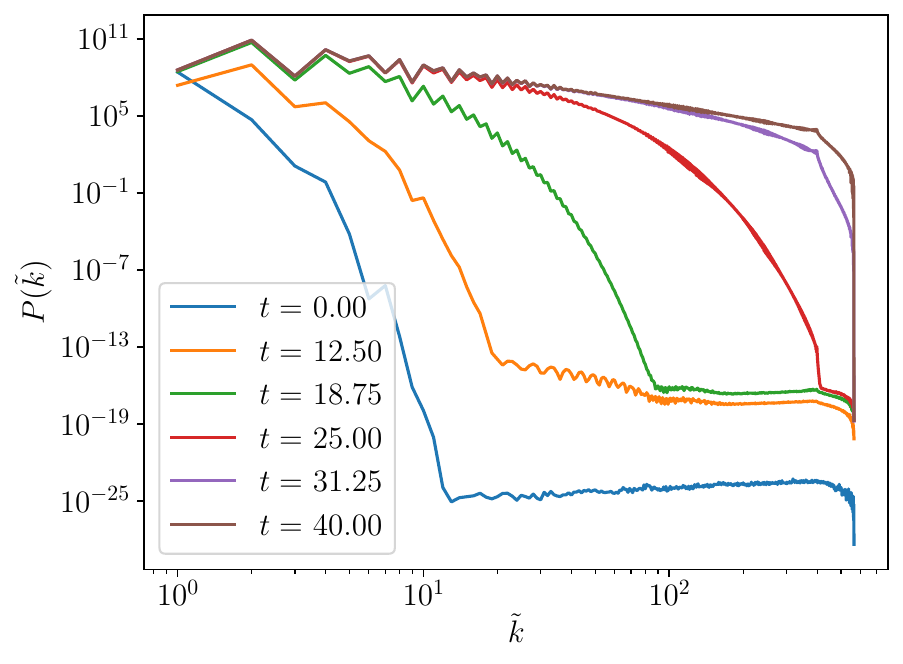}
    \caption{Ultraviolet cascade in the power spectrum of $\zeta$. $N=800$, $K=0.5$, $\eps = 0.05$.}
    \label{fig:spectrum_cascade}
\end{figure}

As a test on the reliability of our results, we have monitored both the spectrum and the $\bar{H}^{s}$ norm \eqref{eqn:Hs_bar_defn} for different resolutions. If the instability was due to limited numerical resolution, we would expect that increasing $N$ should delay the energy cascade and norm blow up. However, as shown in Figures \ref{fig:spectrum_resolution_comparison}-\ref{fig:Hs_bar_resolution_comparison}, the ultraviolet cascade and norm blow up coincide for all $N$ until, for a given resolution, the spectrum is saturated, at which point the cascade ends and the norms plateau. This strongly suggests that the ultraviolet cascade and norm inflation are both genuine features of the continuum solution rather than numerical artefacts. In particular, the numerical resolution only limits the length of time we can accurately track these features rather than changing the dynamics of the solution itself.

\begin{figure}[htbp]
\centering
\subfigure[Subfigure 1 list of figures text][$t=0.0$]{
\includegraphics[width=0.4\textwidth]{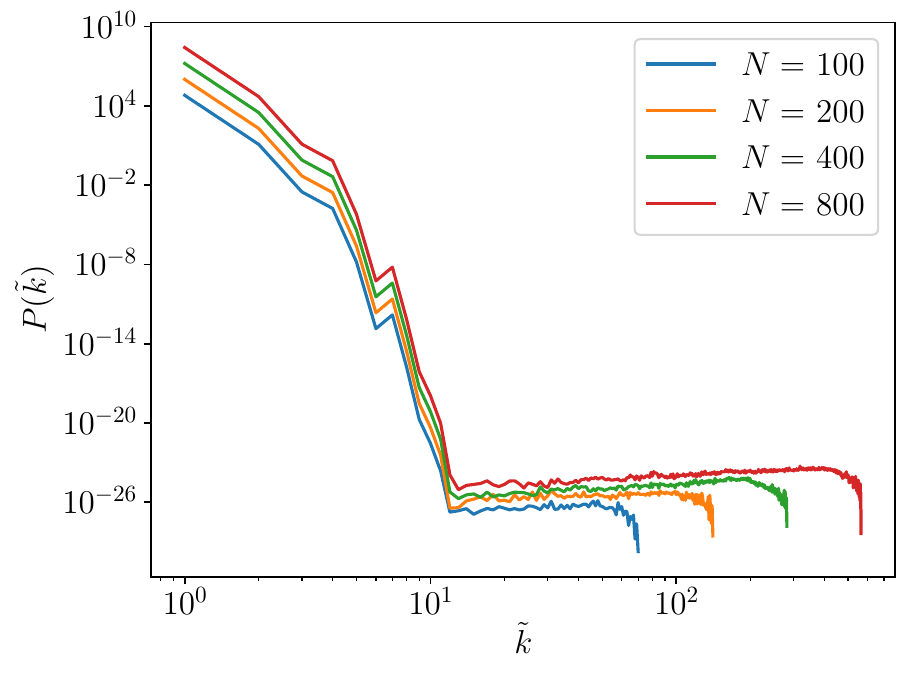}
}
\subfigure[Subfigure 2 list of figures text][$t=18.75$]{
\includegraphics[width=0.4\textwidth]{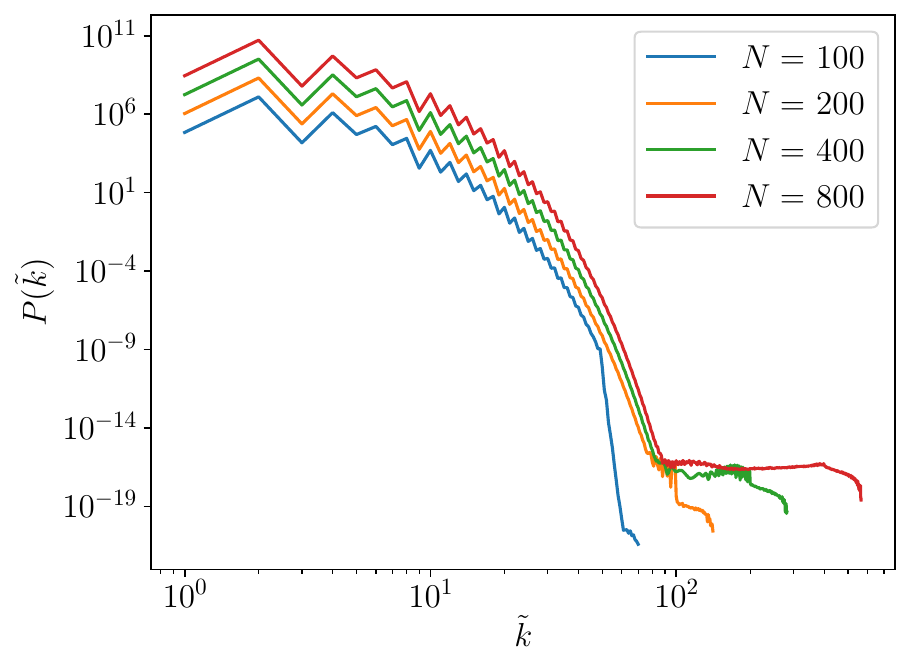}}
\subfigure[Subfigure 2 list of figures text][$t=25$]{
\includegraphics[width=0.4\textwidth]{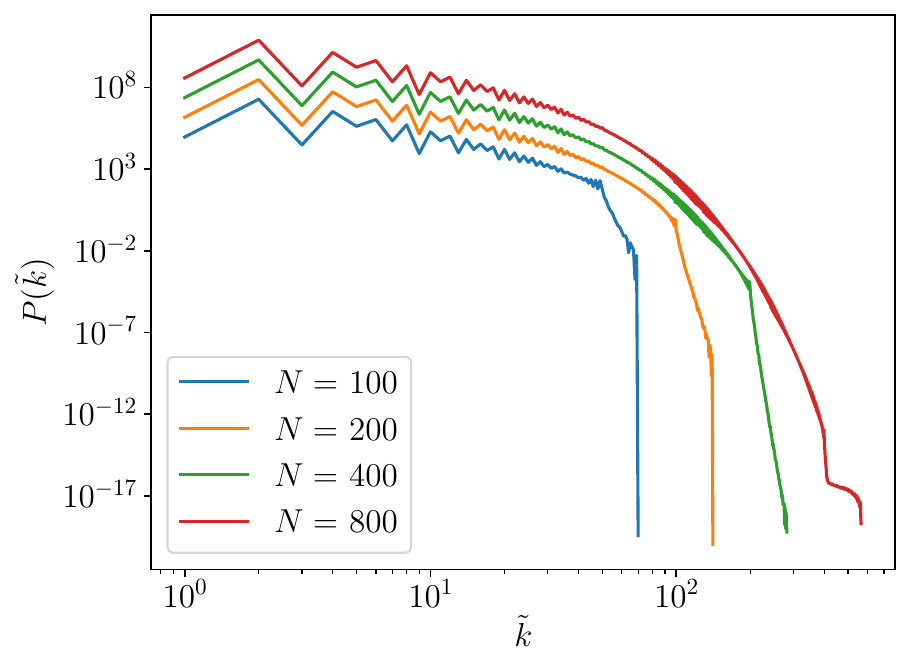}
}
\subfigure[Subfigure 2 list of figures text][$t=40$]{
\includegraphics[width=0.4\textwidth]{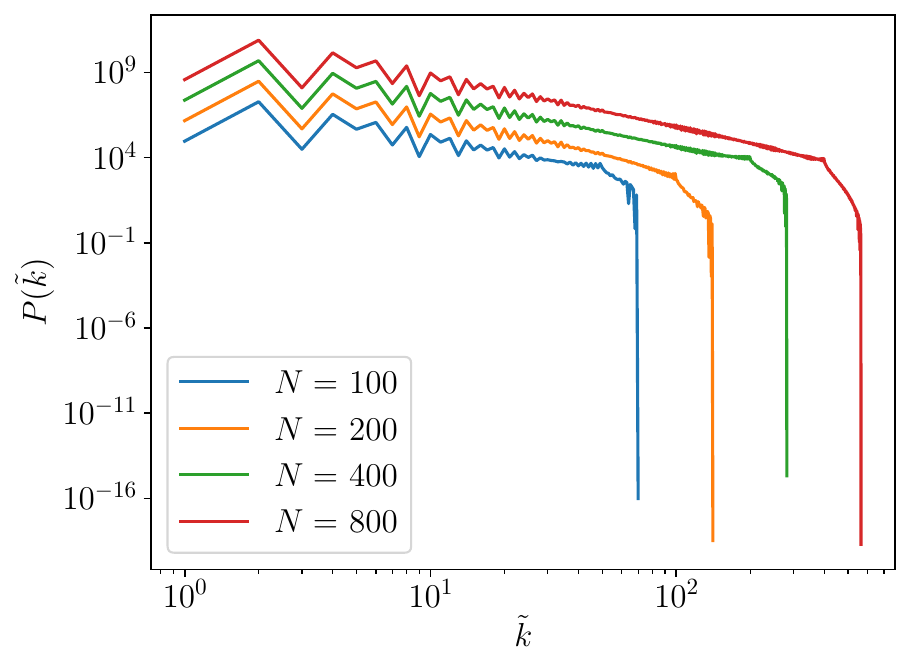}}
\caption{Comparison of the power spectrum $P(\tilde{k})$ for different resolutions. $K=0.5$, $\eps = 0.05$.}
\label{fig:spectrum_resolution_comparison}
\end{figure}

\begin{figure}
    \centering
    \includegraphics[width=0.45\textwidth]{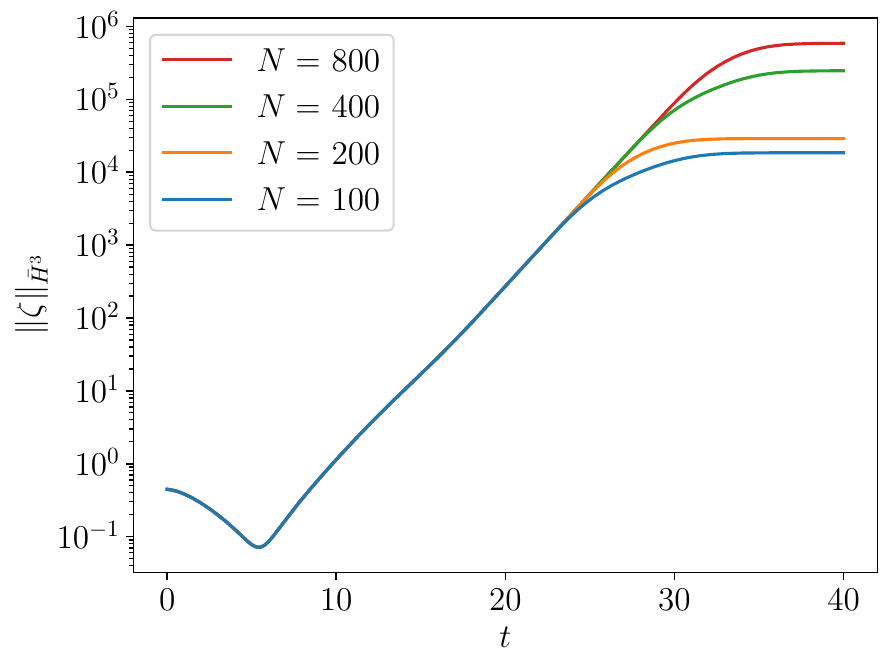}
    \caption{Blow up of the norm $\|\zeta\|_{\bar{H}^{3}}$ for different resolutions. Observe that norms overlap for all resolutions until the power spectrum is saturated which causes the norm to plateau. $K=0.5$, $\eps = 0.05$.}
    \label{fig:Hs_bar_resolution_comparison}
\end{figure}

\subsubsection{Norm Inflation}
Our next goal is to investigate norm inflation for perturbations of the modified density
\begin{align*}
   \zeta_{\text{pert}} := \zeta - \zeta_{\text{bg}} 
\end{align*} 
where $\zeta_{\text{bg}}$ is the spatially homogeneous background solution obtained from the initial data in Section \ref{sec:InitialData} with $\eps=0$. Rather than numerically computing and subtracting the background solution, we study norm inflation in the perturbed solution by considering the zero-average Sobolev norms
\begin{align}
\label{eqn:Hs_bar_defn}
    \|f\|^{2}_{\bar{H}^{s}} := \sum_{k \in \mathbb{Z}^{2}\backslash\{0\}}|k|^{2s}|\widehat{f}(k)|^{2},
\end{align}
where $\widehat{f}$ is the Fourier transform of $f$. Since the background solution is constant in space for fixed $t$, it only contributes to $|k|=0$ mode. Thus, computing $\|\zeta\|_{\bar{H}^{s}} = \|\zeta_{\text{pert}}\|_{\bar{H}^{s}}$ allows us to isolate the growth of the perturbation from the dynamics of the spatially homogeneous background. Moreover, norm inflation for $\|\zeta\|_{\bar{H}^{s}}$ necessarily implies norm inflation for the full Sobolev norm $\|\zeta_{\text{pert}}\|_{H^{s}}$. \newline \par

In Figure \ref{fig:H3_norm_blow_up_K07} we show the behaviour of $\frac{\|\zeta\|_{\bar{H}^{3}}(t)}{\|\zeta\|_{\bar{H}^{3}}(0)}$ for different values of $\eps$ and $K=0.7$. In these plots, we see larger norm blow up as the initial perturbation is decreased which is consistent with an $\bar{H}^{3}$ instability. Similarly, we plot the behaviour of $\|\zeta\|_{\bar{H}^{3}}$ for fixed $\eps$ and varying $K$ in Figure \ref{fig:H3_BlowUp_VariedK} from which we see that the blow up exponentially slows as $K \searrow \frac{1}{3}$. \newline \par

\begin{figure}
    \centering
    \includegraphics[width=0.45\textwidth]{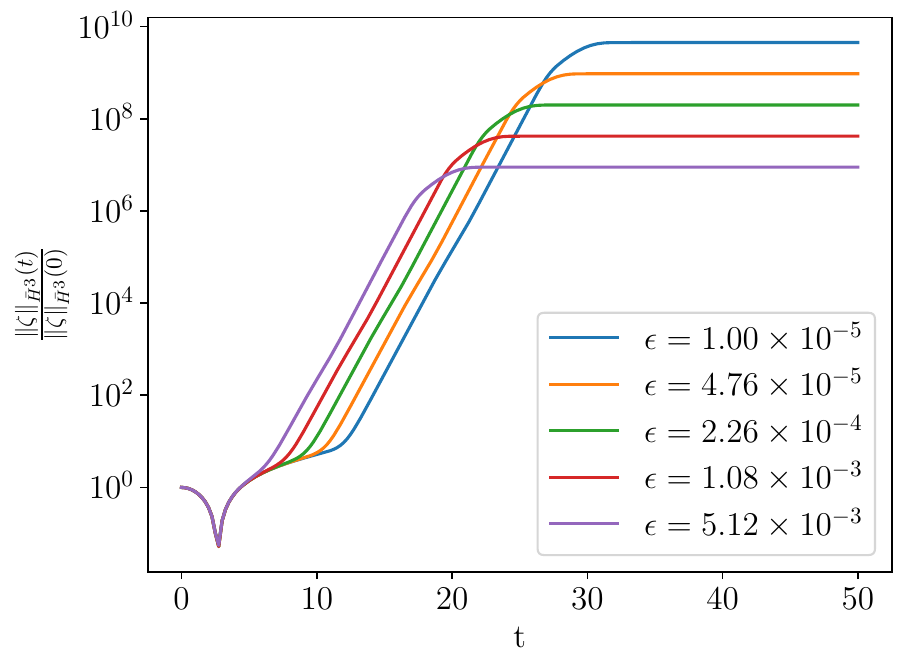}
    \caption{Blow up of the normalised $\bar{H}^{3}$ norm $\frac{\|\zeta\|_{\bar{H}^{3}}(t)}{\|\zeta\|_{\bar{H}^{3}}(0)}$ for different values of the perturbation parameter $\eps$. $N=400$, $K=0.7$}
    \label{fig:H3_norm_blow_up_K07}
\end{figure}

\begin{figure}
    \centering
    \includegraphics[width=0.45\textwidth]{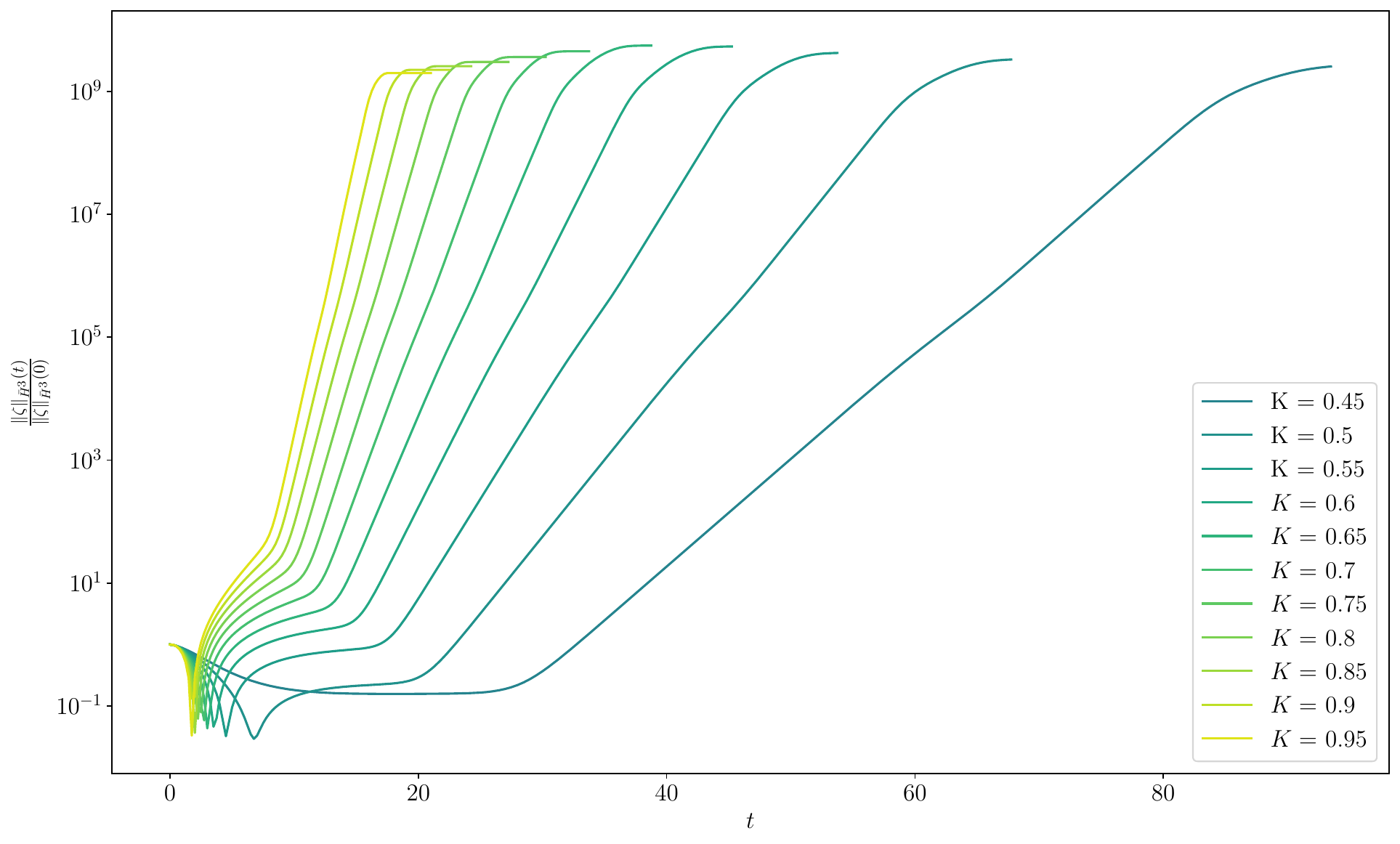}
    \caption{Blow up of the normalised $\bar{H}^{3}$ norm $\frac{\|\zeta\|_{\bar{H}^{3}}(t)}{\|\zeta\|_{\bar{H}^{3}}(0)}$ for different values of the sound speed parameter $K$. Observe that as $K \searrow \frac{1}{3}$, the blow up takes exponentially longer to occur. All solutions are plotted until the norm plateaus. $N=400$, $\eps=10^{-5}$.}
    \label{fig:H3_BlowUp_VariedK}
\end{figure}

In our simulations, we say $\bar{H}^{s}$ norm inflation occurs if, for fixed $K$ and $\eps$, there exists a critical time $t^{*}_{s} > 0$ such that
\begin{align}
\label{eqn:Hbar_s_norm_inflation_condition}
    \|\zeta(t^{*}_{s})\|_{\bar{H}^{s}} \geq \|\zeta(0)\|_{\bar{H}^{s}}^{-1}.
\end{align}
To investigate if $\bar{H}^{s}$ norm inflation occurs for arbitrarily small perturbations, we have monitored the asymptotic behaviour of $t^{*}_{s}$ as $\eps \searrow 0$. Specifically, we run a sequence of simulations with logarithmically spaced $\eps \in [10^{-5},10^{-2}]$ for $K \in \{0.4,0.45,0.5,\cdots , 0.95\}$. As a compromise between accuracy and computational efficiency, we used $N=400$ for all simulations. For all values of $K$ tested, we find evidence that $t^{*}_{s}$ scales like
\begin{align}
\label{eqn:t_crit_scaling}
    t^{*}_{s} = -a\log(\eps - \eps^{*}) + b
\end{align}
where $a >0$, $\eps^{*} \geq 0$, and $b$ are constants and $s>1$. If $\eps^{*} = 0$, then $t^{*}_{s} \nearrow \infty$ as $\eps \searrow 0$ and we expect that $\bar{H}^{s}$ norm inflation occurs for arbitrarily small initial data. On the other hand, if $\eps^{*} >0$, then $t^{*}_{s} \nearrow \infty$ as $\eps \searrow \eps^{*}$ and we expect that $\bar{H}^{s}$  norm inflation does not occur for initial data with $\eps < \eps^{*}$. Examples of this scaling for $t_{3}^{*}$ and $t_{2}^{*}$ are shown in Figure \ref{fig:t_crit_scaling}. \newline \par

\begin{figure}[htbp!]
\centering
\subfigure[Subfigure 1 list of figures text][]{
\includegraphics[width=0.42\textwidth]{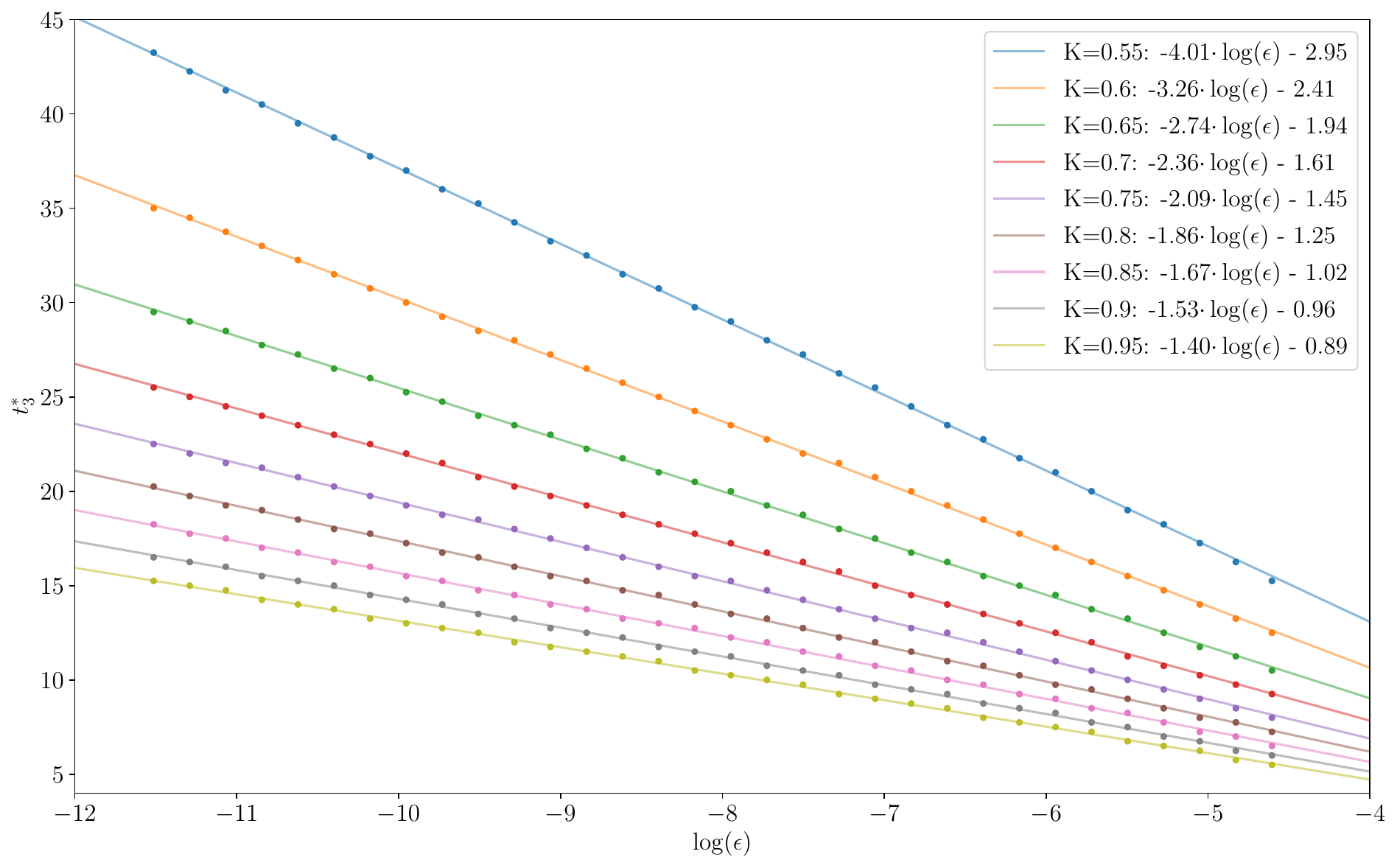}
}
\subfigure[Subfigure 2 list of figures text][]{
\includegraphics[width=0.42\textwidth]{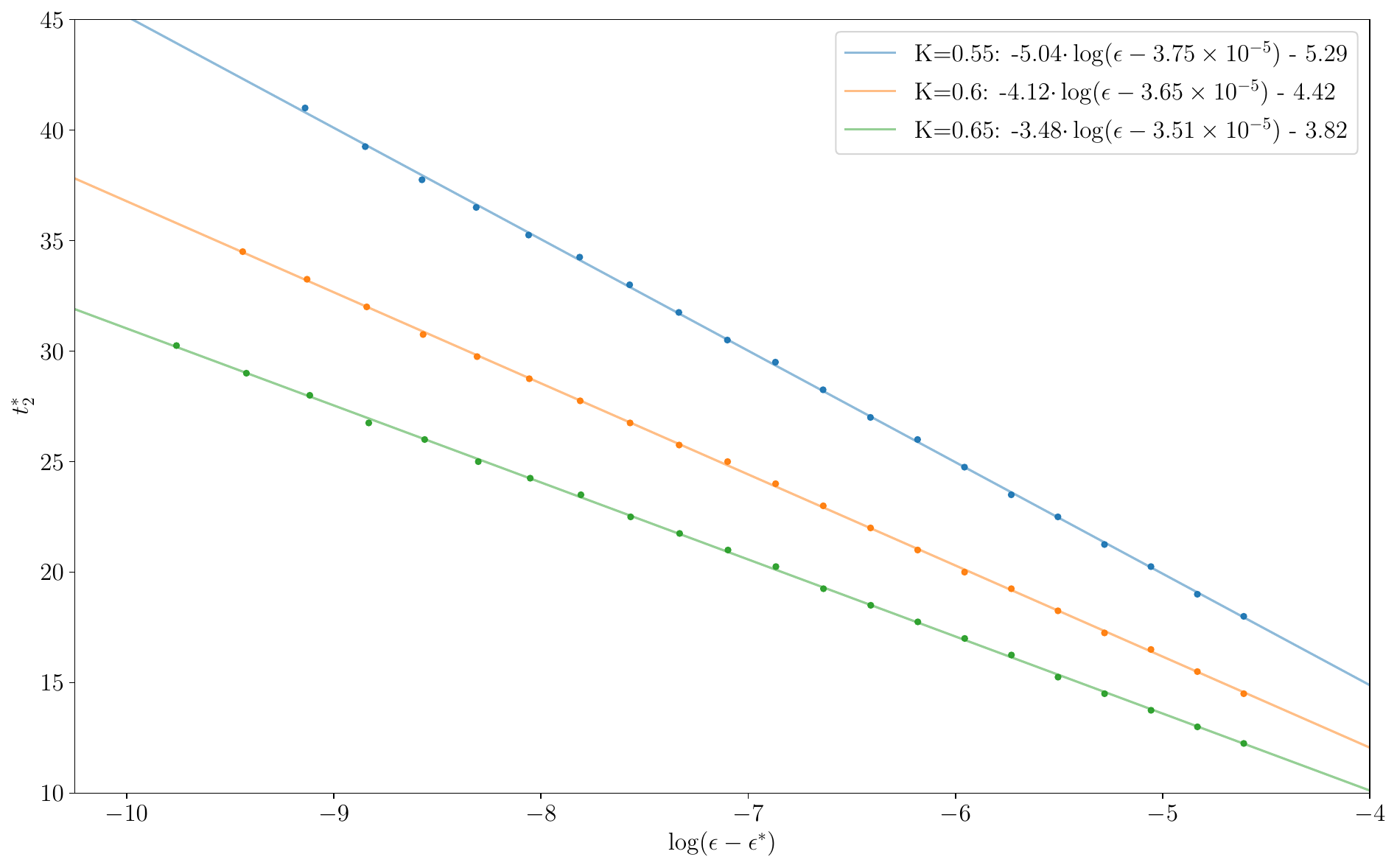}
}
\caption[caption]{Scaling of (a) $t^{*}_{3}$ and (b) $t^{*}_{2}$ as a function of $\eps$ for different values of $K$. The lines of best fit for (a) and (b) are given by \eqref{eqn:t_crit_scaling} with $\eps^{*} = 0$ and $\eps^{*} >0$, respectively. This suggests that norm inflation will occur for arbitrarily\protect\footnotemark small initial data for $\bar{H}^{3}$ but not $\bar{H}^{2}$ over these $K$ values. $N=400$.}
\label{fig:t_crit_scaling}
\end{figure}

\footnotetext{As discussed in the body of the text, our data \textit{cannot} rule out the possibility that $\eps^{*}$ is non-zero but much smaller than $10^{-5}$. Hence, claims of norm inflation for \textit{arbitrarily} small initial data must be treated cautiously at this point.}

For the $\bar{H}^{s}$ norm with $3 \leq s \leq 5$, our asymptotic analysis suggests that small data norm inflation occurs for all values of $K$ tested. On the other hand, it appears that a transition from stability to norm inflation occurs for the $\bar{H}^{2}$ norm somewhere between $0.8 < K < 0.85$. This is somewhat unexpected since Oliynyk's conjecture for the Euler equations on a fixed de Sitter background \cite{Oliynyk:2024} suggest there should be a transition between stability and norm inflation for all $2 \leq s \leq 5$ over the values of $K$ tested.  \newline \par

One possible reason for this discrepancy is that, due to the small sample size for $t^{*}_{s}$ and limited resolution, our numerical results are simply not accurate enough to clearly map out the transition from stability to norm inflation in the $(K,s)$ parameter space. Similarly, we have only studied initial data with $\eps \geq 10^{-5}$. If, for fixed $(K,s)$, the value of $\eps^{*}$ is non-zero but significantly smaller than $10^{-5}$, it is unlikely that we would be able to distinguish between the scaling \eqref{eqn:t_crit_scaling} with $\eps^{*} = 0$ and $\eps^{*} >0$ from our results. Ultimately, we expect it will be necessary to test a much larger range of $\eps$ and use some form of adaptive mesh refinement in order to provide a truly robust classification of the $(K,s)$ parameter space. Nonetheless, our results strongly suggest that norm inflation is a genuine feature of this system, at least for some values of $K$, $s$, and $\eps$. \newline \par

\begin{rem}
    It is also worth noting that Oliynyk does not work with the variable $\zeta$ in his analysis. Instead, he uses a density variable $\texttt{z}$ which roughly corresponds to the logarithm of the particle current $\big(\texttt{z} \approx \log(\Gamma \rho^{\frac{1}{1+K}})\big)$. However, we have found that the $\bar{H}^{s}$ norm of $\log(\Gamma \rho^{\frac{1}{1+K}})$ displays essentially the same behaviour as $\zeta$, so it appears unlikely this is the cause of any discrepancy. 
\end{rem}

\section{Discussion}
\label{sec:Discussion}
We have numerically studied inhomogeneous, $U(1)$-symmetric  perturbations of spatially homogeneous and orthogonal solutions to the Einstein-Euler equations with positive cosmological constant and super-radiative equations of state $K \in (\frac{1}{3},1)$. As part of our numerical implementation, we have derived a new symmetric hyperbolic form of the orthonormal frame formalism which is specifically adapted to $U(1)$-symmetry. We expect this formulation to be useful in both future numerical and analytic studies. \newline \par

Our results provide the first confirmation of the Rendall instability in the fluid outside of one spatial dimension, extending previous work in the Gowdy and $\mathbb{T}^{2}$-symmetric settings \cites{BMO:2023,BMO:2024,ColeyLim:2012,ColeyLim:2013,ColeyLim:2015,Marshall:2026}. Moreover, we find that this instability leads to a forward ultraviolet cascade in the energy density of the fluid matter and norm inflation in $\|\zeta\|_{H^{s}}$, consistent with an earlier conjecture of Oliynyk \cite{Oliynyk:2024} for the $\mathbb{T}^{2}$-symmetric Euler equations on fixed de Sitter backgrounds. However, possibly due to limited numerical resolution, we have been unable to obtain identify a clear transition between stability and norm inflation in the $(K,s)$ parameter space as predicted in \cite{Oliynyk:2024}. We plan on using an adaptive mesh refinement scheme in future work to further investigate this conjectured transition. \newline \par

There are several interesting avenues for future analytical and numerical work. An obvious first step would be to generalise Oliynyk's result \cite{Oliynyk:2024} by rigorously proving that the Rendall instability occurs for the Einstein-Euler system. Beyond this, one might hope to determine whether the conjectured transition between $H^{s}$ stability and norm inflation is a genuine feature of these solutions. It would also be interesting to determine if similar instabilities can develop for alternative equations of state and matter models.

\bibliography{refs.bib}

@article{BergerMoncrief:1993,
	author = {Berger, B.K. and Moncrief, V.},
	doi = {10.1103/PhysRevD.48.4676},
	journal = {Physical Review D},
	number = {10},
	pages = {4676--4687},
	title = {{N}umerical {I}nvestigation of {C}osmological {S}ingularities},
	volume = {48},
	year = {1993}}

@article{BMO:2023,
	author = {F. Beyer and E. Marshall and T.A. Oliynyk},
	doi = {10.1103/PhysRevD.107.104030},
	journal = {Phys. Rev. D},
	pages = {104030},
	publisher = {American Physical Society},
	title = {{F}uture {I}nstability of {FLRW} {F}luid {S}olutions for {L}inear {E}quations of {S}tate {$p=K\rho$} with {$1/3<K<1$}},
	url = {https://link.aps.org/doi/10.1103/PhysRevD.107.104030},
	volume = {107},
	year = {2023}}

@article{BOZ:2025,
      title={{L}ocalized {P}ast {S}tability of the {S}ubcritical {K}asner-{S}calar {F}ield {S}pacetimes}, 
      author={F. Beyer and T.A. Oliynyk and W. Zheng},
      year={2025},
      eprint={2502.09210},
      archivePrefix={arXiv},
      primaryClass={gr-qc},
      url={https://arxiv.org/abs/2502.09210}, 
}

@book{ChoquetBruhat:2009,
	author = {Y. Choquet-Bruhat},
	publisher = {Oxford University Press},
	title = {General {R}elativity and the {E}instein {E}quations},
	year = {2009}}

@article{ColeyLim:2015,
	author = {Coley, A.A. and Lim, W.C.},
	doi = {10.1088/0264-9381/33/1/015009},
	issn = {0264-9381, 1361-6382},
	journal = {Classical and Quantum Gravity},
	number = {1},
	pages = {015009},
	title = {{S}pikes and {M}atter {I}nhomogeneities in {M}assless {S}calar {F}ield {M}odels},
	volume = {33},
	year = {2016}}

@article{CurtisGarfinkle:2005,
	author = {Curtis, J. and Garfinkle, D.},
	doi = {10.1103/PhysRevD.72.064003},
	issn = {1550-7998, 1550-2368},
	journal = {Physical Review D},
	number = {6},
	pages = {064003},
	title = {Numerical {S}imulations of {S}tiff {F}luid {G}ravitational {S}ingularities},
	volume = {72},
	year = {2005}}

@article{FOW:2021,
	author = {D. Fajman and T.A. Oliynyk and Z. Wyatt},
	journal = {Commun. Math. Phys.},
	pages = {401-426},
	title = {{S}tabilizing {R}elativistic {F}luids on {S}pacetimes with {N}on-{A}ccelerated {E}xpansion},
	volume = {383},
	year = {2021}}

@unpublished{Fajman_et_al:2021b,
	author = {D.~Fajman and M.~Ofner and Z.~Wyatt},
	note = {preprint [arXiv:2107.00457]},
	title = {Slowly expanding stable dust spacetimes},
	year = {2021}}

@article{Garfinkle:2004,
	author = {Garfinkle, D.},
	doi = {10.1103/PhysRevLett.93.161101},
	journal = {Physical Review Letters},
	number = {16},
	pages = {6},
	title = {Numerical {{Simulations}} of {{Generic Singularities}}},
	volume = {93},
	year = {2004}}

@article{Garfinkle:2007,
	author = {Garfinkle, D.},
	doi = {10.1088/0264-9381/24/12/S19},
	issn = {0264-9381, 1361-6382},
	journal = {Classical and Quantum Gravity},
	number = {12},
	pages = {S295-S306},
	title = {Numerical {S}imulations of {G}eneral {G}ravitational {S}ingularities},
	volume = {24},
	year = {2007}}

@article{Geroch:1971,
  title={{A} {M}ethod for {G}enerating {S}olutions of {E}instein's {E}quations},
  author={Geroch, Robert},
  journal={Journal of Mathematical Physics},
  volume={12},
  number={6},
  pages={918--924},
  year={1971},
  publisher={AIP Publishing}
}

@article{HadzicSpeck:2015,
	author = {M. Had\v{z}i\'{c} and J. Speck},
	doi = {10.1142/S0219891615500046},
	eprint = {http://www.worldscientific.com/doi/pdf/10.1142/S0219891615500046},
	journal = {J. Hyper. Differential Equations},
	pages = {87-188},
	title = {{T}he {G}lobal {F}uture {S}tability of the {FLRW} {S}olutions to the {D}ust-{E}instein {S}ystem with a {P}ositive {C}osmological {C}onstant},
	url = {http://www.worldscientific.com/doi/abs/10.1142/S0219891615500046},
	volume = {12},
	year = {2015}}

@article{Ijjas_et_al:2021,
	author = {A. Ijjas and F. Pretorius and P.J. Steinhardt and D. Garfinkle},
	doi = {10.1088/1475-7516/2021/12/030},
	journal = {JCAP},
	month = {dec},
	number = {12},
	pages = {030},
	publisher = {{IOP} Publishing},
	title = {Dynamical {A}ttractors in {C}ontracting {S}pacetimes {D}ominated by {K}inetically {C}oupled {S}calar {F}ields},
	url = {https://doi.org/10.1088/1475-7516/2021/12/030},
	volume = {2021},
	year = 2021}

@article{LeFlochWei:2021,
	author = {P.G. LeFloch and C. Wei},
	journal = {Annales de l'Institut Henri Poincar\'e C, Analyse non lin\'eaire},
	pages = {757-814},
	title = {The {N}on-{L}inear {S}tability of {S}elf-{G}ravitating {I}rrotational {C}haplygin {F}luids in a {FLRW} {G}eometry},
	volume = {38},
	year = {2021}}

@article{LiuWei:2021,
	author = {C.~Liu and C.~Wei},
	doi = {10.1007/s00023-020-00987-1},
	journal = {Ann. Henri Poincar{\'e}},
	pages = {715-779},
	title = {{F}uture {S}tability of the {FLRW} {S}pacetime for a {L}arge {C}lass of {P}erfect {F}luids},
	volume = {22},
	year = {2021}}

@article{LubbeKroon:2013,
	author = {C. L\"{u}bbe and J.A. Valiente Kroon},
	doi = {http://dx.doi.org/10.1016/j.aop.2012.10.011},
	issn = {0003-4916},
	journal = {Annals of Physics},
	pages = {1-25},
	title = {A {C}onformal {A}pproach for the {A}nalysis of the {N}on-{L}inear {S}tability of {R}adiation {C}osmologies},
	url = {http://www.sciencedirect.com/science/article/pii/S0003491612001728},
	volume = {328},
	year = {2013}}

@article{MarshallOliynyk:2022,
	author = {E. Marshall and T.A. Oliynyk},
	journal = {Lett. Math. Phys.},
	pages = {102},
	title = {{O}n the {S}tability of {R}elativistic {P}erfect {F}luids with {L}inear {E}quations of {S}tate $ p= {K}\rho $ where $1/3< {K}< 1$},
	volume = {113},
	year = {2023}}

@article{Oliynyk:CMP_2016,
	author = {T.A. Oliynyk},
	journal = {Commun. Math. Phys.},
	pages = {293-312; see the preprint [arXiv:1505.00857] for a corrected version},
	title = {{F}uture {S}tability of the {FLRW} {F}luid {S}olutions in the {P}resence of a {P}ositive {C}osmological {C}onstant},
	volume = {346},
	year = {2016}}

@article{Rendall:2004,
	author = {Rendall, A.D.},
	doi = {10.1007/s00023-004-0189-1},
	journal = {Ann. Henri Poincar{\'e}},
	number = {6},
	pages = {1041-1064},
	title = {{A}symptotics of {S}olutions of the {{Einstein}} {E}quations with {P}ositive {C}osmological {C}onstant},
	volume = {5},
	year = {2004}}

@article{RodnianskiSpeck:2013,
	author = {I.~Rodnianski and J.~Speck},
	journal = {J. Eur. Math. Soc.},
	pages = {2369-2462},
	title = {{T}he {S}tability of the {I}rrotational {E}uler-{E}instein {S}ystem with a {P}ositive {C}osmological {C}onstant},
	volume = {15},
	year = {2013}}

@article{Ringstrom:2009,
	adsurl = {https://ui.adsabs.harvard.edu/abs/2009CMaPh.290..155R},
	author = {Ringstr{\"o}m, H.},
	doi = {10.1007/s00220-009-0812-6},
	journal = {Commun. Math. Phys.},
	number = {1},
	pages = {155-218},
	title = {{Power Law Inflation}},
	volume = {290},
	year = {2009}}

@article{Speck:2012,
	author = {J. Speck},
	journal = {Selecta Mathematica},
	pages = {633-715},
	title = {The {N}on-{L}inear {F}uture-{S}tability of the {FLRW} {F}amily of {S}olutions to the {E}uler-{E}instein {S}ystem with a {P}ositive {C}osmological {C}onstant},
	volume = {18},
	year = {2012}}

@article{Speck:2013,
	author = {J. Speck},
	journal = {Arch. Rat. Mech.},
	pages = {535-579},
	title = {The {S}tabilizing {E}ffect of {S}pacetime {E}xpansion on {R}elativistic {F}luids with {S}harp {R}esults for the {R}adiation {E}quation of {S}tate},
	volume = {210},
	year = {2013}}

@book{Wald:1984,
	author = {R.M. Wald},
	publisher = {University of Chicago Press},
	title = {General {R}elativity},
	year = {1984}}

@article{Wei:2018,
	author = {C. Wei},
	doi = {https://doi.org/10.1016/j.jde.2018.05.007},
	issn = {0022-0396},
	journal = {Journal of Differential Equations},
	pages = {3441 - 3463},
	title = {Stabilizing {E}ffect of the {P}ower {L}aw {I}nflation on {I}sentropic {R}elativistic {F}luids},
	url = {http://www.sciencedirect.com/science/article/pii/S002203961830278X},
	volume = {265},
	year = {2018}}

@article{BergerGarfinkle:1998,
	author = {Berger, B.K. and Garfinkle, D.},
	doi = {10.1103/PhysRevD.57.4767},
	journal = {Physical Review D},
	note = {DOI:~\href{https://doi.org/10.1103/PhysRevD.57.4767}{10.1103/PhysRevD.57.4767}},
	number = {8},
	pages = {4767--4777},
	title = {Phenomenology of the {{Gowdy}} Universe on {$T^3\times R$}},
	volume = {57},
	year = {1998}}

@article{UEWE:2003,
   title={Past {A}ttractor in {I}nhomogeneous {C}osmology},
   volume={68},
   ISSN={1089-4918},
   url={http://dx.doi.org/10.1103/PhysRevD.68.103502},
   DOI={10.1103/physrevd.68.103502},
   number={10},
   journal={Physical Review D},
   publisher={American Physical Society (APS)},
   author={Uggla, C. and van Elst, H. and Wainwright, J. and Ellis, G.F.R.},
   year={2003},
   month=nov }

@article{RöhrUggla:2005,
doi = {10.1088/0264-9381/22/17/026},
url = {https://dx.doi.org/10.1088/0264-9381/22/17/026},
year = {2005},
publisher = {},
volume = {22},
number = {17},
pages = {3775},
author = {N. Röhr and C. Uggla},
title = {Conformal {R}egularization of {E}instein's {F}ield {E}quations},
journal = {Classical and Quantum Gravity}
}

@book{Gourgoulhon:2012,
  title={{$3+1$} {F}ormalism in {G}eneral {R}elativity: {B}ases of {N}umerical {R}elativity},
  author={Gourgoulhon, {\'E}.},
  isbn={9783642245251},
  lccn={2011942212},
  series={Lecture Notes in Physics},
  year={2012},
  publisher={Springer Berlin Heidelberg}
}

@article{Sandin:2009,
   title={{T}ilted {T}wo-{F}luid {B}ianchi {T}ype {I} {M}odels},
   volume={41},
   ISSN={1572-9532},
   url={http://dx.doi.org/10.1007/s10714-009-0799-5},
   DOI={10.1007/s10714-009-0799-5},
   number={11},
   journal={General Relativity and Gravitation},
   publisher={Springer Science and Business Media LLC},
   author={Sandin, P.},
   year={2009},
   month=apr, pages={2707–2724} }

@article{SandinUggla:2008,
   title={{B}ianchi {T}ype {I} {M}odels with {T}wo {T}ilted {F}luids},
   volume={25},
   ISSN={1361-6382},
   url={http://dx.doi.org/10.1088/0264-9381/25/22/225013},
   DOI={10.1088/0264-9381/25/22/225013},
   number={22},
   journal={Classical and Quantum Gravity},
   publisher={IOP Publishing},
   author={Sandin, P. and Uggla, C.},
   year={2008},
   month=oct, pages={225013} }

@article{HewittWainwright:1990,
doi = {10.1088/0264-9381/7/12/011},
url = {https://dx.doi.org/10.1088/0264-9381/7/12/011},
year = {1990},
publisher = {},
volume = {7},
number = {12},
pages = {2295},
author = {C.G. Hewitt and  J. Wainwright},
title = {{O}rthogonally {T}ransitive {G2} {C}osmologies},
journal = {Classical and Quantum Gravity},
}

@article{Wainwright:1979,
  title={A {C}lassification {S}cheme for {N}on-{R}otating {I}nhomogeneous {C}osmologies},
  author={Wainwright, J.},
  journal={Journal of Physics A: Mathematical and General},
  volume={12},
  number={11},
  pages={2015},
  year={1979},
  publisher={IOP Publishing}
}

@article{Elst_et_al:2001,
   title={{D}ynamical {S}ystems {A}pproach to {$G_{2}$} {C}osmology},
   volume={19},
   ISSN={1361-6382},
   url={http://dx.doi.org/10.1088/0264-9381/19/1/304},
   DOI={10.1088/0264-9381/19/1/304},
   number={1},
   journal={Classical and Quantum Gravity},
   publisher={IOP Publishing},
   author={van Elst, H. and Uggla, C. and Wainwright, J.},
   year={2001},
   month=dec, pages={51–82} }

@book{EllisWainwright:1997,
  title={{D}ynamical {S}ystems in {C}osmology},
  author={Wainwright, J. and Ellis, G.F.R.},
  isbn={9780521554572},
  lccn={96028592},
  year={1997},
  publisher={Cambridge University Press}
}

@article{Andersson_et_al:2005,
  title = {{A}symptotic {S}ilence of {G}eneric {C}osmological {S}ingularities},
  author = {Andersson, L. and van Elst, H. and Lim, W.C. and Uggla, C.},
  journal = {Phys. Rev. Lett.},
  volume = {94},
  issue = {5},
  pages = {051101},
  numpages = {4},
  year = {2005},
  publisher = {American Physical Society},
  doi = {10.1103/PhysRevLett.94.051101},
  url = {https://link.aps.org/doi/10.1103/PhysRevLett.94.051101}
}

@article{FOOW:2023,
   title={{T}he {S}tability of {R}elativistic {F}luids in {L}inearly {E}xpanding {C}osmologies},
   volume={2024},
   ISSN={1687-0247},
   url={http://dx.doi.org/10.1093/imrn/rnad241},
   DOI={10.1093/imrn/rnad241},
   number={5},
   journal={International Mathematics Research Notices},
   publisher={Oxford University Press (OUP)},
   author={Fajman, D. and Ofner, M. and Oliynyk, T.A. and Wyatt, Z.},
   year={2023},
   month=oct, pages={4328–4383} }

@article{BMO:2024,
   title={Past {I}nstability of {FLRW} {S}olutions of the {E}instein-{E}uler-{S}calar {F}ield {E}quations for {L}inear {E}quations of {S}tate {$p=K\rho$} with {$0\leq K<1/3$}},
   volume={110},
   ISSN={2470-0029},
   url={http://dx.doi.org/10.1103/PhysRevD.110.044060},
   DOI={10.1103/physrevd.110.044060},
   number={4},
   journal={Physical Review D},
   publisher={American Physical Society (APS)},
   author={Beyer, F. and Marshall, E. and Oliynyk, T.A.},
   year={2024}}

@article{ElstUggla:1997,
   title={General {R}elativistic {O}rthonormal {F}rame {A}pproach},
   volume={14},
   ISSN={1361-6382},
   url={http://dx.doi.org/10.1088/0264-9381/14/9/021},
   DOI={10.1088/0264-9381/14/9/021},
   number={9},
   journal={Classical and Quantum Gravity},
   publisher={IOP Publishing},
   author={van Elst, H. and Uggla, C.},
   year={1997},
   month=sep, pages={2673–2695} }

@article{MacCallum:1973,
    author = {MacCallum, M.A.H.},
    editor = {Schatzman, E.},
    title = {Cosmological {M}odels from a {G}eometric {P}oint of {V}iew},
    eprint = {2001.11387},
    archivePrefix = {arXiv},
    journal = {Cargese Lect. Phys.},
    volume = {6},
    pages = {61-174},
    year = {1973},
    primaryClass={gr-qc},
      url={https://arxiv.org/abs/2001.11387}
}

@article{EllisMaccallum:1969,
author = {G.F.R. Ellis and M.A.H. MacCallum},
title = {{A {C}lass of {H}omogeneous {C}osmological {M}odels}},
volume = {12},
journal = {Communications in Mathematical Physics},
number = {2},
publisher = {Springer},
pages = {108 -- 141},
year = {1969},
}

@article{Fournodavlos_et_al:2024,
      title={{F}uture {S}tability of {P}erfect {F}luids with {E}xtreme {T}ilt and {L}inear {E}quation of {S}tate $p=c_s^2\rho$ for the {E}instein-{E}uler {S}ystem with {P}ositive {C}osmological {C}onstant: {T}he {R}ange $\frac{1}{3}<c_s^2<\frac{3}{7}$}, 
      author={G. Fournodavlos and E. Marshall and T.A. Oliynyk},
      year={2024},
      eprint={2404.06789},
      archivePrefix={arXiv},
      primaryClass={math.AP},
      url={https://arxiv.org/abs/2404.06789}}

@article{Lim_et_al:2004,
   title={{A}symptotic {I}sotropization in {I}nhomogeneous {C}osmology},
   volume={69},
   ISSN={1550-2368},
   url={http://dx.doi.org/10.1103/PhysRevD.69.103507},
   DOI={10.1103/physrevd.69.103507},
   number={10},
   journal={Physical Review D},
   publisher={American Physical Society (APS)},
   author={Lim, W.C. and van Elst, H. and Uggla, C. and Wainwright, J.},
   year={2004},
   month=may }

@article{Hewitt_et_al:2001,
   title={{T}he {A}symptotic {R}egimes of {T}ilted {B}ianchi {II} {C}osmologies},
   volume={33},
   ISSN={1572-9532},
   url={http://dx.doi.org/10.1023/A:1002075902953},
   DOI={10.1023/a:1002075902953},
   number={1},
   journal={General Relativity and Gravitation},
   publisher={Springer Science and Business Media LLC},
   author={Hewitt, C.G. and Bridson, R. and Wainwright, J.},
   year={2001},
   month=jan, pages={65–94} }

@article{HernStewart:1998,
   title={{T}he {G}owdy {C}osmologies {R}evisited},
   volume={15},
   ISSN={1361-6382},
   url={http://dx.doi.org/10.1088/0264-9381/15/6/014},
   DOI={10.1088/0264-9381/15/6/014},
   number={6},
   journal={Classical and Quantum Gravity},
   publisher={IOP Publishing},
   author={Hern, S.D. and Stewart, J.M.},
   year={1998},
   month=jun, pages={1581–1593} }

@article{GarfinklePretorius:2020,
   title={{S}pike {B}ehavior in the {A}pproach to {S}pacetime {S}ingularities},
   volume={102},
   ISSN={2470-0029},
   url={http://dx.doi.org/10.1103/PhysRevD.102.124067},
   DOI={10.1103/physrevd.102.124067},
   number={12},
   journal={Physical Review D},
   publisher={American Physical Society (APS)},
   author={Garfinkle, D. and Pretorius, F.},
   year={2020},
   month=dec }

@article{BergerMoncrief:1998b,
   title={{E}vidence for an {O}scillatory {S}ingularity in {G}eneric {$U(1)$} {S}ymmetric {C}osmologies on {$T^{3}\times R$}},
   volume={58},
   ISSN={1089-4918},
   url={http://dx.doi.org/10.1103/PhysRevD.58.064023},
   DOI={10.1103/physrevd.58.064023},
   number={6},
   journal={Physical Review D},
   publisher={American Physical Society (APS)},
   author={Berger, B.K. and Moncrief, V.},
   year={1998},
   month=aug }

@article{BergerMoncrief1998a,
   title={{N}umerical {E}vidence that the {S}ingularity in {P}olarized {$U(1)$} {S}ymmetric {C}osmologies on {$T^{3}\times R$} is {V}elocity {D}ominated},
   volume={57},
   ISSN={1089-4918},
   url={http://dx.doi.org/10.1103/PhysRevD.57.7235},
   DOI={10.1103/physrevd.57.7235},
   number={12},
   journal={Physical Review D},
   publisher={American Physical Society (APS)},
   author={Berger, B.K. and Moncrief, V.},
   year={1998},
   month=jun, pages={7235–7240} }

@book{Strikwerda:2004,
  title={{F}inite {D}ifference {S}chemes and {P}artial {D}ifferential {E}quations},
  author={Strikwerda, J.C.},
  year={2004},
  publisher={SIAM}
}

@article{ColeyLim:2013,
   title={General {R}elativistic {D}ensity {P}erturbations},
   volume={31},
   ISSN={1361-6382},
   url={http://dx.doi.org/10.1088/0264-9381/31/1/015020},
   DOI={10.1088/0264-9381/31/1/015020},
   number={1},
   journal={Classical and Quantum Gravity},
   publisher={IOP Publishing},
   author={Lim, W.C. and Coley, A.A.},
   year={2013},
   month=nov, pages={015020} }

@article{ColeyLim:2012,
  title = {{G}enerating {M}atter {I}nhomogeneities in {G}eneral {R}elativity},
  author = {Coley, A.A. and Lim, W.C.},
  journal = {Phys. Rev. Lett.},
  volume = {108},
  issue = {19},
  pages = {191101},
  numpages = {5},
  year = {2012},
  publisher = {American Physical Society},
  doi = {10.1103/PhysRevLett.108.191101},
  url = {https://link.aps.org/doi/10.1103/PhysRevLett.108.191101}
}

@article{Oliynyk:2024,
  title={{O}n the {F}ractional {D}ensity {G}radient {B}low-{U}p {C}onjecture of {R}endall},
  author={Oliynyk, T.A.},
  journal={Communications in Mathematical Physics},
  volume={405},
  number={8},
  pages={197},
  year={2024},
  publisher={Springer}
}

@article{Hervik_et_al:2010,
  title={{F}uture {A}symptotics of {T}ilted {B}ianchi {T}ype {II} {C}osmologies},
  author={Hervik, S. and Lim, W.C. and Sandin, P. and Uggla, C.},
  journal={Class. Quantum Gravity},
  volume={27},
  number={18},
  pages={185006},
  year={2010},
  publisher={IOP Publishing}
}

@article{Fajman_et_al:2025,
      title={Stability of {F}luids in {S}pacetimes with {D}ecelerated {E}xpansion}, 
      author={D. Fajman and M. Ofner and T.A. Oliynyk and Z. Wyatt},
      year={2025},
      eprint={2501.12798},
      archivePrefix={arXiv},
      primaryClass={gr-qc},
      url={https://arxiv.org/abs/2501.12798}, 
}

@article{Mondal:2024,
  title={The {N}on-{L}inear {S}tability of (n+1)-{D}imensional {FLRW} {S}pacetimes},
  author={Mondal, P.},
  journal={J. Hyperbolic Differ. Equ.},
  volume={21},
  number={02},
  pages={329--422},
  year={2024},
  publisher={World Scientific}
}

@article{Mondal:2021,
  title={The {L}inear {S}tability of the (n+1)-{D}imensional {FLRW} {S}pacetimes},
  author={Mondal, P.},
  journal={Class. Quantum Gravity},
  volume={38},
  number={22},
  pages={225009},
  year={2021},
  publisher={IOP Publishing}
}

@phdthesis{Lim_Thesis:2004,
    title={The {D}ynamics of {I}nhomogeneous {C}osmologies}, 
    author={W.C. Lim},
    year={2004},
    school = {University of {W}aterloo},
    eprint={gr-qc/0410126},
      archivePrefix={arXiv},
      primaryClass={gr-qc},
      url={https://arxiv.org/abs/gr-qc/0410126}
}

@article{Ellis:1967,
  title={{D}ynamics of {P}ressure-{F}ree {M}atter in {G}eneral {R}elativity},
  author={Ellis, G.F.R.},
  journal={Journal of Mathematical Physics},
  volume={8},
  number={5},
  pages={1171--1194},
  year={1967},
  publisher={American Institute of Physics}
}

@article{FMO:2025,
      title={{F}uture {S}tability of {T}ilted {T}wo-{F}luid {B}ianchi {I} {S}pacetimes}, 
      author={G. Fournodavlos and E. Marshall and T.A. Oliynyk},
      year={2025},
      eprint={2508.15155},
      archivePrefix={arXiv},
      primaryClass={gr-qc},
      url={https://arxiv.org/abs/2508.15155}, 
}

@article{Wainwright:1981,
  title={{E}xact {S}patially {I}nhomogeneous {C}osmologies},
  author={Wainwright, J.},
  journal={Journal of Physics A: Mathematical and General},
  volume={14},
  number={5},
  pages={1131},
  year={1981},
  publisher={IOP Publishing}
}

@article{Marshall:2026,
  title = {Mixmaster {F}luids {N}ear the {B}ig {B}ang},
  author = {Marshall, E.},
  journal = {Phys. Rev. D},
  volume = {113},
  issue = {2},
  pages = {024053},
  numpages = {17},
  year = {2026},
  publisher = {American Physical Society},
  doi = {10.1103/86cf-6yct}
}

@article{Moncrief:1986,
  title = {Reduction of {{Einstein}}'s {E}quations for {V}acuum {S}pace-{T}imes with {S}pacelike {{U}}(1) {I}sometry {G}roups},
  author = {Moncrief, V.},
  year = 1986,
  journal = {Ann. Phys.},
  volume = {167},
  number = {1},
  pages = {118--142},
  issn = {00034916},
  doi = {10.1016/S0003-4916(86)80009-4},
  urldate = {2026-02-06},
  langid = {english},
  note = {DOI:~\href{https://doi.org/10.1016/S0003-4916(86)80009-4}{10.1016/S0003-4916(86)80009-4}}
}

@article{Garfinkle_et_al:2023,
  title={Initial {C}onditions {P}roblem in {C}osmological {I}nflation {R}evisited},
  author={Garfinkle, D. and Ijjas, A. and Steinhardt, P.J.},
  journal={Physics Letters B},
  volume={843},
  pages={138028},
  year={2023},
  publisher={Elsevier}
}

@article{Ijjas_et_al:2024,
  title={Smoothing and {F}lattening the {U}niverse {T}hrough {S}low {C}ontraction {V}ersus {I}nflation},
  author={Ijjas, A. and Steinhardt, P.J. and Garfinkle, D. and Cook, W.G.},
  journal={Journal of Cosmology and Astroparticle Physics},
  volume={2024},
  number={07},
  pages={077},
  year={2024},
  publisher={IOP Publishing}
}

@incollection{Miller:1980,
	author = {J.G. Miller},
    editor = {A.R. Marlow},
	booktitle = {Quantum {T}heory and {G}ravitation},
	pages = {221--232},
	publisher = {Academic Press, New York},
	title = {{K}aluza and {K}lein's {F}ive-{D}imensional {R}elativity},
	year = {1980}
}

@article{Parker:1984,
  title={On {S}ome {T}heorems of {G}eroch and {S}tiefel},
  author={Parker, P.E.},
  journal={Journal of mathematical physics},
  volume={25},
  number={3},
  pages={597--599},
  year={1984},
  publisher={American Institute of Physics}
}

@article{Dong:2026,
      title={The {P}ast {S}tability of {K}asner {S}ingularities for the $(3+1)$-{D}imensional {E}instein {V}acuum {S}pacetime under {P}olarized ${U}(1)$-{S}ymmetry}, 
      author={K. Dong},
      year={2026},
      eprint={2601.06957},
      archivePrefix={arXiv},
      primaryClass={gr-qc},
      url={https://arxiv.org/abs/2601.06957}, 
}

@article{Moncrief:1990,
  title={Reduction of the {E}instein-{M}axwell and {E}instein-{M}axwell-{H}iggs {E}quations for {C}osmological {S}pacetimes with {S}pacelike {U(1)} {I}sometry {G}roups},
  author={Moncrief, V.},
  journal={Classical and Quantum Gravity},
  volume={7},
  number={3},
  pages={329--352},
  year={1990}
}

@article{ChoquetBruhatMoncrief:2001,
   title={Future {G}lobal in {T}ime {E}insteinian {S}pacetimes with {U(1)} {I}sometry {G}roup},
   volume={2},
   ISSN={1424-0637},
   url={http://dx.doi.org/10.1007/s00023-001-8602-5},
   DOI={10.1007/s00023-001-8602-5},
   number={6},
   journal={Annales Henri Poincaré},
   publisher={Springer Science and Business Media LLC},
   author={Choquet-Bruhat, Y. and Moncrief, V.},
   year={2001},
   month=Dec, pages={1007–1064} }

@article{ChoquetBruhat:2003,
  title={Future {C}omplete {S1} {S}ymmetric {E}insteinian {S}pacetimes, the {U}npolarized {C}ase},
  author={Choquet-Bruhat, Y.},
  journal={Comptes rendus. Math{\'e}matique},
  volume={337},
  number={2},
  pages={129--136},
  year={2003}
}

@article{York:1972,
  title={Role of {C}onformal {T}hree-{G}eometry in the {D}ynamics of {G}ravitation},
  author={York Jr, J.W.},
  journal={Physical review letters},
  volume={28},
  number={16},
  pages={1082},
  year={1972},
  publisher={APS}
}

@article{Lichnerowicz:1944,
  title={L’int{\'e}gration des {\'e}quations de la gravitation relativiste et le probl{\`e}me des $ n $ corps},
  author={Lichnerowicz, A.},
  journal={Journal de mathematiques pures et appliquees},
  volume={23},
  pages={37--63},
  year={1944}
}

@article{Carrasco_et_al:2012,
  title={Turbulent {F}lows for {R}elativistic {C}onformal {F}luids in 2+1 {D}imensions},
  author={Carrasco, F. and Lehner, L. and Myers, R.C. and Reula, O. and Singh, A.},
  journal={Physical Review D—Particles, Fields, Gravitation, and Cosmology},
  volume={86},
  number={12},
  pages={126006},
  year={2012},
  publisher={APS}
}

@Article{numpy:2020,
 title         = {Array {P}rogramming with {NumPy}},
 author        = {C.R. Harris and K.J. Millman and S.J.
                 van der Walt and R. Gommers and P. Virtanen and D.
                 Cournapeau and E. Wieser and J. Taylor and S.
                 Berg and N.J. Smith and R. Kern and M. Picus
                 and S. Hoyer and M.H. van Kerkwijk and M.
                 Brett and A. Haldane and J. Fern{\'{a}}ndez del
                 R{\'{i}}o and M. Wiebe and P. Peterson and P.
                 G{\'{e}}rard-Marchant and K. Sheppard and T. Reddy and
                 W. Weckesser and H. Abbasi and C. Gohlke and
                 T.E. Oliphant},
 year          = {2020},
 month         = sep,
 journal       = {Nature},
 volume        = {585},
 number        = {7825},
 pages         = {357--362},
 doi           = {10.1038/s41586-020-2649-2},
 publisher     = {Springer Science and Business Media {LLC}},
 url           = {https://doi.org/10.1038/s41586-020-2649-2}
}

@article{Bernhardt:2025,
  title={Future {S}tability of {S}olutions of the {E}instein-{N}onlinear {S}calar {F}ield {S}ystem with {D}ecelerated {E}xpansion},
  author={Bernhardt, L.},
  journal={arXiv preprint arXiv:2508.15303},
  year={2025}
}

@article{KehleMoschidis:2026,
      title={Weakly {T}urbulent {D}ynamics on {S}chwarzschild-{A}d{S} {B}lack {H}ole {S}pacetimes}, 
      author={C. Kehle and G. Moschidis},
      year={2026},
      eprint={2604.12118},
      archivePrefix={arXiv},
      primaryClass={gr-qc},
      url={https://arxiv.org/abs/2604.12118}, 
}

@article{Moschidis:2020,
   title={A {P}roof of the {I}nstability of {A}d{S} for the {E}instein-{N}ull {D}ust {S}ystem with an {I}nner {M}irror},
   volume={13},
   ISSN={2157-5045},
   url={http://dx.doi.org/10.2140/apde.2020.13.1671},
   DOI={10.2140/apde.2020.13.1671},
   number={6},
   journal={Analysis \& PDE},
   publisher={Mathematical Sciences Publishers},
   author={Moschidis, G.},
   year={2020},
   pages={1671–1754} }

@article{Moschidis:2023,
	author = {Moschidis, G.},
	doi = {10.1007/s00222-022-01152-7},
	id = {Moschidis2023},
	isbn = {1432-1297},
	journal = {Inventiones mathematicae},
	number = {2},
	pages = {467--672},
	title = {A {P}roof of the {I}nstability of {A}d{S} for the {E}instein-{M}assless {V}lasov {S}ystem},
	url = {https://doi.org/10.1007/s00222-022-01152-7},
	volume = {231},
	year = {2023}}

@article{BizonRostworowski:2011,
  title={Weakly {T}urbulent {I}nstability of {A}nti--de {S}itter {S}pacetime},
  author={Bizo{\'n}, P. and Rostworowski, A.},
  journal={Physical Review Letters},
  volume={107},
  number={3},
  pages={031102},
  year={2011},
  publisher={APS}
}

@article{Beyer:2026,
  title={The {E}xtremely-{T}ilted {F}luid {R}egime {N}ear {A}symptotically {K}asner {B}ig {B}ang {S}ingularities},
  author={Beyer, F.},
  journal={arXiv preprint arXiv:2602.19361},
  year={2026}
}

@article{Zheng:2026,
  title={Localized {B}ig {B}ang {S}tability of {S}pacetime {D}imensions $n\geq4$},
  author={Zheng, W.},
  journal={arXiv preprint arXiv:2601.21677},
  year={2026}
}

@article{Bizon_et_al:2015,
  title={Resonant {D}ynamics and the {I}nstability of {A}nti--de {S}itter {S}pacetime},
  author={Bizo{\'n}, P. and Maliborski, M. and Rostworowski, A.},
  journal={Physical {R}eview {L}etters},
  volume={115},
  number={8},
  pages={081103},
  year={2015},
  publisher={APS}
}

@article{Jalmuzna_et_al:2011,
  title={AdS {C}ollapse of a {S}calar {F}ield in {H}igher {D}imensions},
  author={Ja{\l}mu{\.z}na, J. and Rostworowski, A. and Bizo{\'n}, P.},
  journal={Physical Review D—Particles, Fields, Gravitation, and Cosmology},
  volume={84},
  number={8},
  pages={085021},
  year={2011},
  publisher={APS}
}

\appendix

\section{Berger-Moncrief Parameterisation of U(1)-Symmetric Metrics}
\label{app:BergerMoncrief_FrameExample}
In \cites{BergerMoncrief1998a,BergerMoncrief:1998b}, Berger and Moncrief numerically investigated the behaviour of $U(1)$-symmetric cosmologies near initial spacelike singularities. Building on work of \cite{Moncrief:1986}, they used the following parametrisation of a general $U(1)$-symmetric metric
\begin{align}
\label{eqn:BergerMoncrief_U1_Metric}
    g &= -e^{-2\varphi + 2\Lambda -4\tau}d\tau^{2} + e^{-2\tau+\Lambda-2\varphi}\Big[\frac{1}{2}(e^{2z}+e^{-2z}(1+x)^{2})du^{2} + (e^{2z}+e^{-2z}(x^{2}-1)dudv \nonumber \\
    &+ \frac{1}{2}(e^{2z} + e^{-2z}(1-x)^{2})dv^{2}\Big] + e^{2\varphi}(dx^{3} + \beta_{1}du + \beta_{2}dv)^{2},
\end{align}
where the spacelike Killing vector is given by $\xi = \frac{\del}{\del x^{3}}$ and $\varphi$, $\Lambda$, $z$, $x$, $\beta_{1}$, and $\beta_{2}$ are all functions of $(\tau,u,v)$ alone. It is straightforward to verify that an orthogonal dual frame $\theta^{a}$ is given by
\begin{equation}
\begin{gathered}
\label{eqn:MoncriefBerger_DualFrame}
    \theta^{0} = e^{-\varphi + \Lambda -2\tau}d\tau, \quad \theta^{1} = \frac{1}{\sqrt{2}}e^{-\tau + \frac{\Lambda}{2}-\varphi}\big(e^{z}du + e^{-z}dv), \\
    \theta^{2} = \frac{1}{\sqrt{2}}e^{-\tau + \frac{\Lambda}{2}-\varphi}\big(e^{-z}(x+1)du + e^{-z}(x-1)dv\big), \quad \theta^{3} = e^{\varphi}(dx^{3} + \beta_{1}du + \beta_{2}dv).
\end{gathered}
\end{equation}
Similarly, the corresponding orthonormal frame vectors $e_{a}$ are given by
\begin{equation}
\begin{gathered}
\label{eqn:MoncriefBerger_Frame}
    e_{0} = e^{2\tau-\Lambda + \varphi}\del_{	\tau}, \quad e_{1} = \frac{1}{\sqrt{2}}e^{\tau-z-\frac{\Lambda}{2}+\varphi}\Big( (1-x)\del_{u} + (1+x)\del_{v} + \big((x-1) \beta_{1} - (1+x)\beta_{2}\big)\del_{x^{3}}\Big) \\
    e_{2} = \frac{1}{\sqrt{2}}e^{\tau+z-\frac{\Lambda}{2}+\varphi}\Big(\del_{u} -\del_{v} +(\beta_{2}-\beta_{1})\del_{x^{3}}\Big), \quad e_{3} = e^{-\varphi}\del_{x^{3}}.
\end{gathered}
\end{equation}
In particular, \eqref{eqn:MoncriefBerger_Frame} is an example of a group-invariant frame with $e_{3}=\lambda^{\frac{-1}{2}}\xi$. \newline \par

Next, using the decomposition of the commutator coefficients \cite{RöhrUggla:2005}*{Page 6}
\begin{equation*}
\begin{gathered}
    H = -\frac{1}{3}\tensor{c}{_0^A_A}, \;\; \sigma_{AB} = -\tensor{c}{_0^C_{\langle A}}\delta_{B \rangle C}, \;\; \dot{u}_{A} = \tensor{c}{_0^0_A}, \\
    w_{A} = \frac{1}{4}\tensor{\eps}{_A^{BC}}\tensor{c}{_B^0_C}, \;\; w_{A} + \Omega_{A} = \frac{1}{2}\tensor{\eps}{_{AB}^C}\tensor{c}{_0^B_C}, \;\; n^{AB} = \frac{1}{2}\eps^{ab(A}\tensor{c}{_a^{B)}_b}, \;\; a_{A} = \frac{1}{2}\tensor{c}{_A^B_B},
\end{gathered}
\end{equation*}
the identity 
\begin{align*}
\tensor{c}{_a^c_b} = \tensor{\omega}{_a^c_b} -\tensor{\omega}{_b^c_a},
\end{align*}
and the formula for the connection coefficients \cite{Wald:1984}*{Page 50},
\begin{align*}
\omega_{abc} = e_{a}^{\gamma}e_{b}^{\mu}\nabla_{\gamma}e_{c\mu},
\end{align*}
we obtain the identities
\begin{align}
\label{eqn:MoncriefBerger_U1_FrameDecomposition}
\dot{u}_{3} &= a_{3} = n_{11} = n_{12} = n_{22} = \omega_{A} =  0, \nonumber \\
\dot{u}_{1} &= -\frac{1}{\sqrt{2}}e^{\tau-x-\frac{1}{2}\Lambda + \varphi}\Big[-(1+x)\del_{v}\Lambda + (1+x)\del_{v}\varphi + (x-1)(\del_{u}\Lambda - \del_{u}\varphi)\Big], \nonumber \\
\dot{u}_{2} &= \frac{1}{\sqrt{2}}e^{\tau+z-\frac{1}{2}\Lambda + \varphi}\Big(-\del_{v}\Lambda + \del_{v}\varphi + \del_{u}\Lambda - \del_{u}\varphi\Big), \nonumber \\
a_{1} &= -\frac{1}{4\sqrt{2}}e^{\tau-z-\frac{1}{2}\Lambda+\varphi}\Big(2\del_{v}x - 2(1+x)\del_{v}z + (1+x)\del_{v}\Lambda - 2\del_{u}x + (x-1)(2\del_{u}z - \del_{u}\Lambda)\Big), \nonumber \\
a_{2} &= \frac{1}{2\sqrt{2}}e^{\tau+z-\frac{1}{2}\Lambda + \varphi}(\del_{v}\varphi - \del_{u}\varphi), \quad H = \frac{1}{3}e^{2t-\Lambda+\varphi}(-2 + \del_{\tau}\Lambda - \del_{\tau}\varphi), \nonumber \\
\sigma_{11} &= \frac{1}{6}e^{2\tau - \lambda + \varphi}(-2 + 6\del_{\tau}z + \del_{\tau}\Lambda - 4\del_{\tau}\varphi) , \quad \sigma_{12} = \frac{1}{2}e^{2\tau-2z-\Lambda+\varphi}\del_{\tau}x, \nonumber \\
\sigma_{13} &= \frac{1}{2\sqrt{2}}e^{-z-\frac{3}{2}\lambda + 3(\tau+\varphi)}\big((1-x)\del_{\tau}\beta_{1} + (1+x)\del_{\tau}\beta_{2}\big), \\
\sigma_{22} &= \frac{1}{6}e^{2\tau -\Lambda +\varphi}(-2 -6\del_{\tau}z + \del_{\tau}\Lambda - 4\del_{\tau}\varphi),\quad \sigma_{23} = \frac{1}{2\sqrt{2}}e^{z-\frac{3}{2}\Lambda + 3(\tau+\varphi)}(\del_{\tau}\beta_{1} - \del_{\tau}\beta_{2}), \nonumber \\
\sigma_{33} &= \frac{-1}{3}e^{2\tau - \Lambda + \varphi}(-2 + \del_{\tau}\Lambda -4\del_{\tau}\varphi), \nonumber \\
n_{13} &= \frac{1}{4\sqrt{2}}e^{\tau+z-\frac{1}{2}\Lambda+\varphi}(-2\del_{v}z - \del_{v}\Lambda + 4\del_{v}\varphi + 2\del_{u}z + \del_{u}\Lambda -4\del_{u}\varphi), \nonumber \\
n_{23} &= \frac{1}{2\sqrt{2}}e^{\tau-z-\frac{1}{2}\Lambda + \varphi}\Big(-2\del_{v}x + 2(1+x)\del_{v}z - (1+x)\del_{v}\Lambda + 4\del_{v}\varphi + 2\del_{u}x + 2\del_{u}z - \del_{u}\Lambda \nonumber \\
&+ x\big(4\del_{v}\varphi - 2\del_{u}z + \del_{u}\Lambda - 4\del_{u}\varphi\big) + 4\del_{u}\varphi\Big), \nonumber \\
n_{33} &= e^{2\tau -\Lambda + 3\varphi}(\del_{u}\beta_{2} - \del_{v}\beta_{1}), \quad \Omega_{1} = \frac{1}{2\sqrt{2}}e^{z-\frac{3}{2}\Lambda + 3(\tau+\varphi)}(\del_{\tau}\beta_{1} - \del_{\tau}\beta_{2}), \nonumber \\
\Omega_{2} &= \frac{1}{2\sqrt{2}}e^{-z - \frac{3}{2}\Lambda + 3(\tau+\varphi)}\Big((x-1)\del_{\tau}\beta_{1} - (1+x)\del_{\tau}\beta_{2}\Big), \quad \Omega_{3} = \frac{1}{2}e^{2\tau - 2z - \Lambda + \varphi}\del_{\tau}x. \nonumber
\end{align}
Observe that the conditions $\dot{u}_{3}=a_{3}=n_{11}=n_{12}=n_{22}=\omega_{A}=0$ are consistent with the simplifications obtained from using a group invariant frame \eqref{eqn:U1_Frame_Simplifications} in Section \ref{sec:U1_OrthonormalFrame_Formalism}. Similarly, from \eqref{eqn:MoncriefBerger_U1_FrameDecomposition} we see that  
\begin{align*}
    \Omega_{2} = -\sigma_{13},  \quad \Omega_{1} = \sigma_{23},
\end{align*}
which is also consistent with \eqref{eqn:Group_Invariant_Commutator_Simplifications}. 
The condition 
\begin{align*}
    \Omega_{3} = \sigma_{12}
\end{align*}
is a gauge choice for the rotation of the spatial frame. Note that a polarised $U(1)$-symmetric spacetime is obtained by setting $\beta_{1}=\beta_{2}=0$ in \eqref{eqn:BergerMoncrief_U1_Metric}. In this case, using \eqref{eqn:MoncriefBerger_U1_FrameDecomposition}, we find
\begin{align*}
\sigma_{13}=\sigma_{23}=n_{33}=\Omega_{1}=\Omega_{2} =0
\end{align*}
which, again, is consistent with the simplifications obtained from $\xi$ being hypersurface orthogonal \eqref{eqn:Polarised_U1_Frame_Simplifications_a}-\eqref{eqn:Polarised_U1_Frame_Simplifications_b} in Section \ref{sec:U1_OrthonormalFrame_Formalism}.

\section{3+1 Orthonormal Frame Equations}
\label{app:3+1_Frame_Eqns}
For a general orthonormal frame, we can introduce local coordinates such that
\begin{align*}
    e_{0} = \alpha^{-1}(\del_{t} - \beta^{\Omega}\del_{\Omega}), \quad e_{A} = e^{\Omega}_{A}\del_{\Omega}.
\end{align*}
The field equations, without symmetry assumptions and $w_{A}=0$, are given by \cite{RöhrUggla:2005}*{Pages 6-7}
\begin{align}
\label{eqn:frame_evo_3+1}
e_{0}(e_{A}^{\Sigma}) &= -(H\delta_{A}^{B} + \sigma_{A}^{B}+\tensor{\eps}{_A^B_C}\Omega^{C})e_{B}^{\Sigma}, \\
e_{0}(a_{A}) &= -e_{A}(H) + \frac{1}{2}e_{B}(\tensor{\sigma}{_{A}^{B}}) +\frac{1}{2}\tensor{\eps}{_A^B_C}e_{B}(\Omega^{C}) - H(\dot{u}_{A}+a_{A}) \nonumber \\
&+ \tensor{\sigma}{_A^B}(\frac{1}{2}\dot{u}_{B}-a_{B}) +\tensor{\eps}{_A^B_C}\Omega^{C}(\frac{1}{2}\dot{u}_{B}-a_{B}), \\
e_{0}(n^{AB}) &= -Hn^{AB} + 2\tensor{n}{^{(A}_{C}}\sigma^{B)C} +2\tensor{\eps}{^C_D^{(A}}\Omega_{C}n^{B)D} \nonumber \\
&- \epsilon^{CD(A}\Big[e_{C}(\tensor{\sigma}{_D^{B)}})+\dot{u}_{C}\tensor{\sigma}{_D^{B)}}\Big]  +e_{C}(\delta^{C(A}\Omega^{B)}-
\delta^{AB}\Omega^{C}) \nonumber \\
&+ \dot{u}_{C}(\delta^{C(A}\Omega^{B)}-\Omega^{C}\delta^{AB}), \\
e_{0}(H) &= -H^{2}+\frac{1}{3}e_{A}(\dot{u}^{A}) + \frac{1}{3}\dot{u}_{A}(\dot{u}^{A}-2a^{A})-\frac{1}{3}\sigma_{AB}\sigma^{AB}\nonumber \\
&-\frac{1}{6}(T_{00}+\tensor{T}{_A^A}) , \\
\label{eqn:sigma_evo_3+1}
e_{0}(\sigma_{AB}) &= -3H\sigma_{AB} +e_{\langle A}(\dot{u}_{B \rangle}) +\dot{u}_{\langle A}\dot{u}_{B \rangle} -e_{\langle A}(a_{B \rangle})+a_{\langle A}\dot{u}_{B \rangle} \nonumber \\
&+\epsilon_{CD(A}\Big[e^{C}(\tensor{n}{_{B)}^D})+\tensor{n}{_{B)}^D}\dot{u}^{C}-2\tensor{n}{_{B)}^D}a^{C} +2\tensor{\sigma}{_{B)}^{D}}\Omega^{C}\Big] \nonumber \\
&-2\tensor{n}{_{\langle A}^C}\tensor{n}{_{B \rangle}_{C}} +nn_{\langle AB \rangle} +T_{\langle AB \rangle}.
\end{align}
The constraint equations are 
\begin{align}
\label{eqn:C1_Constraint_3+1}
0 = \mathcal{C}_{1} &:= 2e_{[A}(e_{B]}^{\Sigma}) - 2a_{[A}e^{\Sigma}_{B]} - \epsilon_{ABD}n^{DC}e^{\Sigma}_{C}, \\
0 = \mathcal{C}_{2} &:= 2e_{[B}(\dot{u}_{A]}) + 2a_{[A}\dot{u}_{B]} + \epsilon_{ABC}n^{CD}\dot{u}_{D}, \\
\label{eqn:C3_Constraint_3+1}
0 = \mathcal{C}_{3} &:= \dot{u}_{A} - \alpha^{-1}e_{A}(\alpha), \\
0 = \mathcal{C}_{M} &:= e_{B}(\tensor{\sigma}{_A^B}) -2e_{A}(H) -3\tensor{\sigma}{_A^B}a_{B} - \epsilon_{ABC}n^{BD}\tensor{\sigma}{_D^C} - T_{0A} , \\
\label{eqn:Hamiltonian_Constraint_3+1}
0 = \mathcal{C}_{H} &:= 4e_{A}(a^{A}) +6H^{2} - 6a^{A}a_{A} -n^{AB}n_{AB} +\frac{1}{2}n^{2} -\sigma_{AB}\sigma^{AB} \nonumber \\
&-2T_{00}.
\end{align}

\section{Frame Expressions for Additional Quantities}
\label{app:Frame_Quantities}
The twist $w_{\alpha}$ of the Killing vector $\xi^{\mu}$ is defined by
\begin{align*}
    w_{\alpha} = \eps_{\alpha\beta\gamma\delta}\xi^{\beta}\nabla^{\gamma}\xi^{\delta}.
\end{align*}
Expressed in terms of our group invariant, orbit aligned orthonormal frame we then have
\begin{align*}
    w_{a} &= \tensor{\eps}{_{a3}^{cd}}\lambda^{\frac{1}{2}}\Big[e_{c}(\lambda^{\frac{1}{2}}\delta_{3d}) + \lambda^{\frac{1}{2}}\tensor{\omega}{_c^3_d} \Big] \nonumber \\
    &= \lambda \tensor{\eps}{_{a3}^{cd}}\tensor{\omega}{_c^3_d} \nonumber \\
    &= \lambda \Big[\eps_{a123}\big(\tensor{\omega}{_1^3_2} - \tensor{\omega}{_2^3_1}\big) + \eps_{0a23}\big(\tensor{\omega}{_0^3_2} - \tensor{\omega}{_2^3_0}\big) + \eps_{01a3}\big(\tensor{\omega}{_1^3_0} - \tensor{\omega}{_1^3_0}\big)\Big] \nonumber \\
    &= \lambda \Big[\delta_{a}^{0}\tensor{c}{_1^3_2} + \delta_{a}^{1}\tensor{c}{_0^3_2}  + \delta_{a}^{2}\tensor{c}{_1^3_0} \Big].
\end{align*}
where we have used the fact that the frame component of $\xi^{a}$ is given by $\lambda^{\frac{1}{2}}\delta^{a}_{3}$. From this expression we see that $e_{3}(w_{a}) = \xi(w_{a})=0$ since both $\lambda$ and $\tensor{c}{_a^c_b}$ are constant on the group orbits (cf. \eqref{eqn:lambda_invariance}, \eqref{eqn:e3_commutator_independence}). Moreover, we see that the twist vanishes when 
\begin{align*}
   \tensor{c}{_1^3_2} =  \tensor{c}{_0^3_2}  = \tensor{c}{_1^3_0} = 0,
\end{align*}
which is exactly the condition for hypersurface orthogonality of $\xi$ derived from the Lie bracket in \eqref{eqn:HO_Killing_Commutator_Simplifications}, as required. \newline \par

In terms of the variables from Section \ref{sec:Sym_Hyp_U1_Eqns}, the spatial Ricci scalar is given by
\begin{align}
\label{eqn:Spatial_Ricci_Scalar_NewVars}
    {}^{(3)}R &= 4\sqrt{3}e_{1}(\hat{a}_{1}) + 4\sqrt{3}e_{2}(\hat{a}_{2}) - 18(\hat{a}_{1}^{2} + \hat{a}_{2}^{2}) - 6(\hat{n}^{2} + n_{\times}^{2} + n_{+}^{2}).
\end{align}

\section{Numerically Solving the Elliptic Constraint}
\label{app:Numerical_Constraint_Solver}
In this section we outline the process for numerically solving the elliptic equation for the conformal factor $\psi$,
\begin{align}
\label{eqn:elliptic_constraint_app}
    \Delta\psi - \frac{3}{4}\psi^{5}H^{2} + \frac{1}{4}\psi^{5}(\Lambda + T_{00}^{\text{fl}}) + \frac{3}{4}\psi^{-7}(Z_{+}^{2} + Z_{-}^{2} + Z_{2}^{2}) = 0.
\end{align}
After discretising \eqref{eqn:elliptic_constraint_app}, we obtain a system of non-linear equations which can be schematically written as 
\begin{align}
\label{eqn:Functional_ConstraintForm}
    F(\psi) := A\psi - f(\psi) = 0
\end{align}
where $A$ is the fourth order accurate discrete Laplacian with $\Delta x := \Delta x^{1} = \Delta x^{2}$,
\begin{align*}
  (A\psi)_{i,j} = \frac{-\psi_{i-2,j} + 16\psi_{i-1,j} + 16\psi_{i+1,j} - \psi_{i+2,j} -\psi_{i,j-2} + 16\psi_{i,j-1} + 16\psi_{i,j+1} - \psi_{i,j+2} - 60\psi_{i,j}}{12(\Delta x)^{2}} 
\end{align*}
and $f(\psi)$ is given by
\begin{align*}
    f(\psi_{i,j}) :=  -\Big(\frac{1}{4}(\Lambda + T_{00}^{\text{fl}}) -\frac{3}{4}H^{2} 
    \Big)\psi_{i,j}^{5} - \frac{3}{4}(Z_{+}^{2} + Z_{-}^{2} + Z_{2}^{2})\psi_{i,j}^{-7}.
\end{align*}
We can obtain an approximate solution to \eqref{eqn:Functional_ConstraintForm} using Newton's method. In this case, the update formula is given by
\begin{align}
\label{eqn:Newton_Update_Elliptic}
    \psi^{(n+1)} = \psi^{(n)} - \delta\psi^{(n)}
\end{align}
where $\delta\psi^{(n)}$ solves the linear system
\begin{align*}
    J^{(n)}\delta\psi^{(n)} = F(\psi^{(n)})
\end{align*}
and $J^{(n)} := F^{\prime}(\psi^{(n)})$ is the Jacobian matrix evaluated at the current iterate. This linear system can be explicitly written as 
\begin{align*}
    \Delta \delta\psi^{(n)} &+ 5\Big(\frac{1}{4}(\Lambda + T_{00}^{\text{fl}}) -\frac{3}{4}H^{2} 
    \Big)(\psi^{(n)})^{4}\delta\psi^{(n)} - \frac{21}{4}(Z_{+}^{2} + Z_{-}^{2} + Z_{2}^{2})(\psi^{(n)})^{-8}\delta\psi^{(n)} \nonumber \\
    &= \Delta \psi^{(n)} + \Big(\frac{1}{4}(\Lambda + T_{00}^{\text{fl}}) -\frac{3}{4}H^{2} 
    \Big)(\psi^{(n)})^{5} + \frac{3}{4}(Z_{+}^{2} + Z_{-}^{2} + Z_{2}^{2})(\psi^{(n)})^{-7}.
\end{align*}
We solve this system using successive-over-relaxation, cf. \cite{Strikwerda:2004}*{Chapter 13}
\begin{align*}
    \delta\psi^{(n,k+1)}_{i,j} = (1-\omega)\delta\psi^{(n,k)}_{i,j} + \omega\frac{(\Delta x)^{2}F(\psi^{(n)}_{i,j}) - \frac{1}{12}\gamma^{(n,k)}}{-5 + (\Delta x)^{2}\Big[5\big(\frac{1}{4}(\Lambda + T_{00}^{\text{fl}}) -\frac{3}{4}H^{2} 
    \big)(\psi_{i,j}^{(n)})^{4} - \frac{21}{4}(Z_{+}^{2} + Z_{-}^{2} + Z_{2}^{2})(\psi_{i,j}^{(n)})^{-8}\Big]}
\end{align*}
where $\omega \in [1,2)$ is the relaxation parameter, $k$ is the iteration index for the linear problem, and $\gamma^{(n,k)}$ is defined by
\begin{align*}
    \gamma^{(n,k)} := -\delta\psi^{(n,k)}_{i-2,j} + 16\delta\psi^{(n,k)}_{i-1,j}  + 16\delta\psi^{(n,k)}_{i+1,j}  - \delta\psi^{(n,k)}_{i-2,j} -\delta\psi^{(n,k)}_{i,j-2} + 16\delta\psi^{(n,k)}_{i,j-1}  + 16\delta\psi^{(n,k)}_{i,j+1}  - \delta\psi^{(n,k)}_{i,j+2}.
\end{align*}
Once we have obtained a suitably accurate approximation of $\delta\psi^{(n)}$, we compute $\psi^{(n+1)}$ using the Newton update \eqref{eqn:Newton_Update_Elliptic}.

\end{document}